\documentclass[12pt]{article}

\usepackage{amsmath,amssymb,amsthm}

\numberwithin{equation}{section}

\begin{document}

\newtheorem{theorem}{Theorem}[section]
\newtheorem{corollary}[theorem]{Corollary}
\newtheorem{lemma}[theorem]{Lemma}
\newtheorem{proposition}[theorem]{Proposition}

\newcommand{\be}{\begin{equation}}
\newcommand{\ee}{\end{equation}}
\newcommand{\bea}{\begin{eqnarray}}
\newcommand{\eea}{\end{eqnarray}}
\newcommand{\sh}{{\rm sh}}
\newcommand{\ch}{{\rm ch}}
\newcommand{\einde}{$\ \ \ \Box$ \vspace{5mm}}
\newcommand{\De}{\Delta}
\newcommand{\de}{\delta}
\newcommand{\Z}{{\mathbb Z}}
\newcommand{\N}{{\mathbb N}}
\newcommand{\C}{{\mathbb C}}
\newcommand{\Cs}{{\mathbb C}^{*}}
\newcommand{\R}{{\mathbb R}}
\newcommand{\Q}{{\mathbb Q}}
\newcommand{\T}{{\mathbb T}}
\newcommand{\re}{{\rm Re}\, }
\newcommand{\im}{{\rm Im}\, }
\newcommand{\cW}{{\cal W}}
\newcommand{\cJ}{{\cal J}}
\newcommand{\cE}{{\cal E}}
\newcommand{\cA}{{\cal A}}
\newcommand{\cR}{{\cal R}_{\rm ren}}
\newcommand{\cP}{{\cal P}}
\newcommand{\cM}{{\cal M}}
\newcommand{\cN}{{\cal N}}
\newcommand{\cI}{{\cal I}}
\newcommand{\cMs}{{\cal M}^{*}}
\newcommand{\cB}{{\cal B}}
\newcommand{\cD}{{\cal D}}
\newcommand{\cC}{{\cal C}}
\newcommand{\cL}{{\cal L}}
\newcommand{\cF}{{\cal F}}
\newcommand{\cG}{{\cal G}}
\newcommand{\cH}{{\cal H}}
\newcommand{\cO}{{\cal O}}
\newcommand{\cS}{{\cal S}}
\newcommand{\cT}{{\cal T}}
\newcommand{\cU}{{\cal U}}
\newcommand{\cQ}{{\cal Q}}
\newcommand{\cV}{{\cal V}}
\newcommand{\cK}{{\cal K}}
\newcommand{\cZ}{{\cal Z}}
\newcommand{\intR}{\int_{-\infty}^{\infty}}
\newcommand{\intp}{\int_{0}^{\infty}}
\newcommand{\limp}{\re r \to \infty}
\newcommand{\Ogamp}{O(\exp(-\gamma  r))}
\newcommand{\limn}{\re r \to -\infty}
\newcommand{\Ogamn}{O(\exp(\gamma  r))}
\newcommand{\limpn}{\lim_{|\re x| \to \infty}}
\newcommand{\diag}{{\rm diag}}
\newcommand{\Ln}{{\rm Ln}}
\newcommand{\Arg}{{\rm Arg}}
\newcommand{\LHP}{{\rm LHP}}
\newcommand{\RHP}{{\rm RHP}}
\newcommand{\UHP}{{\rm UHP}}
\newcommand{\Res}{{\rm Res}}
\newcommand{\ep}{\epsilon}
\newcommand{\ga}{\gamma}
\newcommand{\ka}{\kappa}
\newcommand{\sing}{{\rm sing}}
\newcommand{\rE}{{\mathrm E}}
\newcommand{\rF}{{\mathrm F}}
\newcommand{\sumt}{\sum_{\tau=+,-}}
\newcommand{\sumtt}{\sum_{\tau,\tau'=+,-}}

\title{Hilbert space theory for relativistic dynamics with reflection.  Special cases}
\author{  Steven Haworth and Simon Ruijsenaars \\ School of Mathematics, \\ University of Leeds, Leeds LS2 9JT, UK}

\date{}

\maketitle

\begin{abstract}
We present and study a novel class of one-dimensional Hilbert space eigenfunction transforms that diagonalize analytic difference operators encoding the (reduced) two-particle relativistic hyperbolic Calogero-Moser dynamics. The scattering is described by reflection and transmission amplitudes~$t$ and~$r$ with function-theoretic features that are quite different from nonrelativistic  amplitudes.   The axiomatic Hilbert space analysis in the appendices is inspired by and applied to the attractive two-particle relativistic Calogero-Moser dynamics for a sequence of special couplings. Together with the scattering function $u$ of the repulsive case, this leads to a triple of amplitudes $u, t, r$ satisfying the Yang-Baxter equations.  
\end{abstract}

\tableofcontents

\newpage

\renewcommand{\theequation}{\thesection.\arabic{equation}}

\setcounter{equation}{0}

\section{Introduction}
This paper can be viewed from different perspectives, and indeed it has a two-fold purpose. On the one hand, it yields (in its four appendices) a study of unitary eigenfunction transforms of a novel type, generalizing the transmission-reflection and bound state picture associated with one-dimensional nonrelativistic Schr\"odinger operators
\be\label{Schr}
- \frac{d^2}{dr^2}+V(r).
\ee
Here we are thinking of real-valued potentials $V(r)$, $r\in \R$,   that are smooth and decay rapidly for $r\to\pm\infty$. As is well known, for such potentials the operator~\eqref{Schr} can have finitely many  bound states with negative eigenvalues, whereas the scattering can be encoded in the so-called Jost solutions, which have eigenvalue $k^2>0$. (We view $r$ and~$k$ as dimensionless variables.) Specifically, the Jost solutions $J_{\pm}(r,k)$ satisfy the time-independent Schr\"odinger equation
\be\label{Jost}
(-\partial_r^2+V(r))J_{\pm}(r,k)=k^2J_{\pm}(r,k),\ \ \ k>0, 
\ee
and are characterized by the asymptotic behavior
\be
J_+(r,k)\sim 
\left\{
\begin{array}{ll}
a(k)e^{-irk}+b(k)e^{irk},  &  r\to\infty, \\
e^{-irk},  &  r\to -\infty, 
\end{array}
\right.
\ee
\be
J_-(r,k)\sim  \left\{
\begin{array}{ll}
e^{irk},  &  r\to\infty, \\
a(k)e^{irk}-b(-k)e^{-irk},  &  r\to -\infty. 
\end{array}
\right.
\ee
It is then customary to define the transmission and reflection coefficients by
\be
t(k)=1/a(k),\ \ \ r(k)=b(k)/a(k).
\ee

The generalized Jost functions we study in the appendices share the asymptotic behavior just recalled, but their transmission and reflection coefficients have a (purely imaginary) period in the spectral variable~$k$. For certain special choices they converge to ordinary Jost functions as this period goes to $i\infty$.   

On the other hand, the paper is concerned with special cases of the eigenfunctions  arising for two particles of opposite charge in the relativistic hyperbolic Calogero-Moser system, and with the associated Hilbert space theory. We intend to return to the Hilbert space theory for the general case in a companion paper. In the present one, however, we already lay the groundwork for the general case by introducing and discussing the pertinent eigenfunctions, both in this section and in more detail in Section~2. In Section~2 we are also better placed to put these Hilbert space aspects  in a wider context and summarize previous literature in the area.

We proceed by recalling the two types of interaction for the hyperbolic Calogero-Moser two-particle system in its center-of-mass frame. 
For the nonrelativistic Calogero-Moser case they are given by the potentials
\be\label{Vra}
V_s(r)=\lambda(\lambda-1)/\sinh^2(r),\ \ V_o(r)=-\lambda(\lambda-1)/\cosh^2(r).\ \ \ 
\ee
(In the physics literature, these are often referred to as P\"oschl-Teller potentials, see for example~\cite{F74}.)
Restricting attention to $\lambda> 1$, this corresponds to a  description of a particle pair with the same/opposite charge having a repulsive/attractive interaction, the charge interpretation being borrowed from the electromagnetic force. 

The relativistic generalizations of the repulsive and attractive nonrelativistic Hamiltonians associated with~\eqref{Vra} can be taken to be of the form
\be\label{Hs}
H_s(\rho,\tau;r)=\left(\frac{\sinh(r+i\tau)}{\sinh(r)}\right)^{1/2}\exp(i\rho \partial_r)\left(\frac{\sinh(r-i\tau)}{\sinh(r)}\right)^{1/2}+(r\to -r),  
\ee
\be\label{Ho}
H_o(\rho,\tau;r)=\left(\frac{\cosh(r+i\tau)}{\cosh(r)}\right)^{1/2}\exp(i\rho \partial_r)\left(\frac{\cosh(r-i\tau)}{\cosh(r)}\right)^{1/2}+(r\to -r). 
\ee
Here we have $\rho,\tau>0$, and the nonrelativistic differential operators given by~\eqref{Schr} and~\eqref{Vra} arise from these analytic difference operators in the limit $\rho \to 0$ with $\tau=\rho\lambda$. (For a comprehensive survey of the relativistic Calogero-Moser $N$-particle systems together with their nonrelativistic and Toda limits we refer to the lecture notes~\cite{R94}.)

The repulsive nonrelativistic Hamiltonian can be diagonalized by using the conical (or Mehler) function specialization of the hypergeometric function. The attractive potential arises from the repulsive one by the two analytic continuations $r\to r\pm i\pi/2$. The attractive nonrelativistic Hamiltonian can then be diagonalized by employing a suitable linear combination of the two correspondingly continued conical functions, taking also bound states into account, which do not arise for the repulsive case. (More details on these diagonalizations from disparate viewpoints can be found for instance in~\cite{F74}, Problems 38 and 39, \cite{T62}, Sections 4.18 and 4.19, and~\cite{K84}.)

The repulsive Hamiltonian~$H_{s}$ can be diagonalized by a relativistic generalization of the conical function. A detailed study of this relativistic conical function and its various limits can be found in~\cite{R11}. We use this function as a starting point to arrive at the eigenfunctions of the attractive counterpart~$H_{o}$ that are the key to the Hilbert space reinterpretation of~$H_o$. As we shall show, they are once more obtained by suitable linear combinations of the $r\to r\pm i\pi/2$ continuations of the relativistic conical function.  

For the special coupling constants $\tau=(N+1)\pi$, $N\in\N$, the associated attractive eigenfunction transforms have certain periodicity features that motivated the general framework that is set out in the appendices. The latter yield a largely self-contained account, which can be read independently of the main text. 

We stress, however, that the assumptions we make in Appendix~A and~B would seem far-fetched without their concrete realizations coming from the main text, as they are not satisfied for the eigenfunctions of the nonrelativistic operators~\eqref{Schr} with a nontrivial $V(r)$. Since one might question whether they can be realized at all, we have included the simplest explicit examples in the appendices, so that one need not  delve into the main text to see that the assumptions are not vacuous.

The Hilbert space theory associated with the opposite-charge Hamiltonian~$H_{o}$ for the case of vanishing reflection has been worked out before in Section~4 of~\cite{R00}. (This case corresponds to couplings of the form~$\tau=(N+1)\rho$, $N\in\N$.) The reflectionless case was also handled via a slightly more general framework, cf.~Section~2 in~\cite{R00}. The latter setting is subsumed by our account in the appendices of the present paper, in the sense that the assumptions we make in Appendix~A and~B  are satisfied by the eigenfunctions studied in~\cite{R00}. 

More specifically, in the appendices we start from wave functions
\be\label{Psipm}
\Psi(r,k)=\Psi^+(r,k)e^{irk} +\Psi^-(r,k)e^{-irk},
\ee
which are counterparts of the above function~$J_-(r,k)/a(k)$, and impose various requirements on the plane wave coefficients $\Psi^{\pm}(r,k)$ of a function-theoretic and asymptotic nature. (The case $\Psi^-(r,k)=0$   yields reflectionless transforms.) The assumptions ensure that we can push through a comprehensive Hilbert space analysis for the associated transforms, without mentioning any Hamiltonians until the last Appendix~D, which is devoted to time-dependent scattering theory.

With the functional analysis relegated to the appendices, the main text demonstrates that the assumptions are satisfied by the pertinent  eigenfunctions of the attractive Hamiltonian~$H_{o}$ when the coupling constant~$\tau$ is  of the form~$(N+1)\pi$, $N\in\N$. (The reflectionless $\tau$-choices~$(N+1)\rho$, $N\in\N$, are briefly reviewed as well.) Therefore, when viewed as an analytic difference operator---as opposed to a Hilbert space operator---the Hamiltonian reduces to the free one
\be
H_f=\exp(i\rho \partial_r)+\exp(-i\rho \partial_r).
\ee
As it stands, the latter can of course be diagonalized by the Fourier transform with kernel $\exp(irk)$, yielding multiplication by $2\cosh(\rho k)$. The crux is, however, that the $H_o$-eigenfunctions for general $\tau$ and $\rho$, specialized to $\tau=(N+1)\pi$, do not reduce to a plane wave at all. Instead, they are of the form~\eqref{Psipm}, with nontrivial coefficients $\Psi^{\pm}(r,k)$ that are $i\rho$-periodic in~$r$, and they give rise to reflection and transmission coefficients $r(k)$ and $t(k)$ that, together with the scattering function $u(k)$ for the repulsive case, satisfy the Yang-Baxter equations.

The latter equations  actually hold for the arbitrary coupling scattering amplitudes, cf.~\eqref{t}--\eqref{YB2}. This state of affairs lends further credence to the long-standing conjecture that the many-particle generalizations of the Hamiltonians~\eqref{Hs} and~\eqref{Ho} lead to a factorized $S$-matrix. For the classical versions of the systems this solitonic scattering has been proved in~\cite{R93}, alongside a detailed account of the connection to several solitonic field theories.  

For the special case $\tau=\pi/2$, a prime example of the latter is the sine-Gordon classical field theory.  For this choice of $\tau$, the  scattering amplitudes $t(k)$, $r(k)$ and $u(k)$ coincide with those for quantum sine-Gordon opposite-charge and equal-charge fermions~\cite{ZZ79,S92}. This state of affairs is in agreement with earlier work from a related viewpoint~\cite{R01}, to which we intend to return in the companion paper dealing with general $\tau$-values.
 
Among the special cases considered in this paper, of particular interest is the case $N=0$, yielding $\tau=\pi$. Indeed, for this choice the repulsive eigenfunction (essentially given by the relativistic conical function) reduces to the plane wave sum~$\sin(rk)$, so that $u(k)=1$. By contrast, the scattering coefficients $t(k)$ and $r(k)$ are nontrivial. (They amount to specializations of~$T_+(k)$~\eqref{tnu} and~$R_+(k)$~\eqref{rnu}.) Physically speaking, this parameter choice yields a system in which particles of the same charge do not interact, whereas oppositely charged particles have an attractive interaction. In particular, an oppositely-charged pair can form a bound state, provided $\rho$ is suitably restricted. 

In order to tie in the main text with the appendices, we need to handle the residue sums that are spawned by contour shifts arising for the pertinent transforms and their adjoints. For isometry of the transforms, it is necessary that these sums vanish, whereas eventual bound states show up via nonzero residue sums associated with the adjoint transforms. At the end of the main text, we also present an in-depth study of   isometry breakdown for sufficiently small $\rho$, a phenomenon without a nonrelativistic counterpart. 

Having presented a bird's eye view of the aims and contents of the paper, we proceed with a more detailed sketch of its results and organization. 

As already mentioned, our starting point is the relativistic conical function. More specifically, we work with a renormalized version~$\cR(a_+,a_-,b;x,y)$. (We recall one of the many integral representations for this function below, cf.~\eqref{Rrdef}.) From a quantum-mechanical perspective, all of its five variables have dimension [position]. More specifically, the parameters~$a_+$ and $a_-$   can be viewed as the interaction length and Compton wave length $\hbar/mc$, while the parameter~$b$ and~variable~$y$ play the role of coupling constant and spectral variable. The unorthodox choice of the latter (inasmuch as it is customary to choose a momentum variable as spectral variable) is inspired by the symmetry of this function under the interchange of $x$ and~$y$ (self-duality). It is written entirely in terms of the hyperbolic gamma function~$G(a_+,a_-;z)$ from~\cite{R97}, which is invariant under swapping $a_+$ and $a_-$, cf.~\eqref{modinv} below. This `modular invariance' is inherited by~$\cR(a_+,a_-,b;x,y)$. The renormalized version has the advantage of being meromorphic in~$b$, $x$ and~$y$, with no poles in~$b$ that do not depend on~$x$ and~$y$. In more detail, its poles can only be located at
\be\label{cRpoles}
\pm z=2ia-ib+ika_++ila_-,\ \ \  z=x,y,\ \ \ k,l\in \N,
\ee
where
\be
a\equiv (a_++a_-)/2.
\ee

Thanks to these symmetry properties, this function satisfies not only the A$\De$E (analytic difference equation)
\be\label{Rade}
A(a_+,a_-,b;x)\cR(a_+,a_-,b;x,y)=2\cosh(\pi y/a_+)\cR(a_+,a_-,b;x,y),
\ee
where the A$\De$O (analytic difference operator) is defined by
\be\label{defA}
A(a_+,a_-,b;z)\equiv V(a_+,b;z)\exp(-ia_-\partial_z) +(z\to -z),
\ee
with
\be\label{defV}
V(a_+,b;z)\equiv \frac{\sinh(\pi (z-ib)/a_+)}{\sinh(\pi z/a_+)},
\ee
but also three more A$\De$Es obtained by swapping~$a_+$ and~$a_-$, and/or~$x$ and~$y$ in~\eqref{Rade}. 

Each of the resulting four A$\De$Os has two more avatars, obtained by similarity transformations involving the generalized Harish-Chandra function
\be\label{defc}
c(b;z)\equiv G(z+ia-ib)/G(z+ia),
\ee
and weight function
\be\label{defw}
w(b;z)\equiv 1/c(b;z)c(b;-z).
\ee
Here and often below, we suppress the dependence on~$a_+$ and~$a_-$ when no ambiguity can arise. (For example, the $c$- and~$w$-functions just introduced are invariant under swapping $a_+$ and~$a_-$.) 
From now on, we also use the abbreviations
\be
c_{\de}(z)\equiv \cosh(\pi z/a_{\de}),\ \ s_{\de}(z)\equiv \sinh(\pi z/a_{\de}),\ \ e_{\de}(z)\equiv \exp(\pi z/a_{\de}),\ \ \de=+,-.
\ee

For our present purposes, we only need two among these eight additional operators. The most important one is the Hamiltonian
\be\label{HA}
H(a_+,a_-,b;x)  \equiv  w(b;x)^{1/2}A(a_+,a_-,b;x) w(b;x)^{-1/2},\ \ \ (b,x)\in(0,2a)\times(0,\infty),
\ee
but we also have occasion to use
\be\label{cA}
\cA(a_+,a_-,b;y)  \equiv  c(b;y)^{-1}A(a_+,a_-,b;y) c(b;y).
\ee
These operators can also be written as 
\be\label{Hb}
  H(b;x)= V(a_+,b;-x)^{1/2}\exp(ia_-\partial_x) V(a_+,b;x)^{1/2}+(x\to -x),
\ee
and
\be\label{cAb}
\cA(b;y)=\exp(-ia_-\partial_y) + V(a_+,b;-y)\exp(ia_-\partial_y) V(a_+,b;y).
\ee
From these formulas one can read off that they are formally self-adjoint, cf.~\eqref{defV}.

The $w$-function~\eqref{defw} is positive for $(b,x)\in(0,2a)\times(0,\infty)$, and throughout the paper we take positive square roots of positive functions. This positivity assertion and the alternative formulas for~$H(b;x)$ and~$\cA(b;y)$ readily follow from key features of the hyperbolic gamma function. For completeness, we briefly digress to summarize the salient properties. 
 
The hyperbolic gamma function can be defined as the meromorphic solution to one of the first order A$\De$Es
\be\label{Gades}
\frac{G(z+ia_{\de}/2)}{G(z-ia_{\de}/2)}=2c_{-\de}(z),\ \ \de=+,-,\ \ a_+,a_->0,
\ee
which is uniquely determined by  
the normalization~$G(0)=1$ and a minimality property;  the second A$\De$E is then satisfied as well.  (This property amounts to requiring absence of zeros and poles in a certain $|\im z|$-strip, and `optimal' asymptotics for $|\re z|\to\infty$, cf.~\cite{R97}.) In the strip $|\im z|<a$ it has the integral representation
\be\label{Gint}
G(a_+,a_-;z)=\exp\left( i
\int_0^\infty\frac{dy}{y}\left(\frac{\sin 2yz}{2\sinh(a_{+}y)\sinh(a_{-}y)} - \frac{z}{a_{+}a_{-} y}\right)\right),
\ee
from which one reads off absence of zeros and poles in this strip and the properties
\be\label{modinv}
 G(a_-,a_+;z) = G(a_+,a_-;z),\ \ \  ({\rm modular\ invariance}),
\ee
\be\label{refl}	
G(-z) = 1/G(z),\ \ \ ({\rm reflection\ equation}),
\ee
\be\label{Gcon}
\overline{G(a_+,a_-;z)}=G(a_+,a_-;-\overline{z}).
\ee 
The hyperbolic gamma function has its poles at 
\be\label{Gpo}
  -ia -ika_+-il a_-,\ \ \ \ k,l\in\N,\ \ \ (G{\rm -poles}),
\ee
and its zeros at
\be\label{Gze}
ia +ika_++il a_-, \ \ \ \ k,l\in\N,\ \ \ \ (G{\rm -zeros}).
\ee
The pole at $-ia$ is simple, and so is the zero at $ia$. 
Finally, we list the asymptotic behavior for $\re (z)\to\pm \infty$:
\be\label{Gas}
	G(a_+,a_-;z) = \exp \big(\mp i\left(\chi+ \pi z^2/2a_+a_-\right)\big)\big(1 + O(\exp(-r |\re(z)|))\big).
\ee
Here, we have
\be\label{chi}
 \chi \equiv \frac{\pi}{24}\left(\frac{a_+}{a_-} + \frac{a_-}{a_+}\right),
\ee
and the decay rate can be any positive number satisfying 
\be
r <2\pi \min(1/a_+,1/a_-).
\ee

The A$\De$O~\eqref{HA} can be viewed as the defining Hamiltonian for the same-charge center-of-mass two-particle system. The opposite-charge defining Hamiltonian is given by
\be\label{tHtA}
\tilde{H}(b;x)= \tilde{w}(b;x)^{1/2}\tilde{A}(b;x) \tilde{w}(b;x)^{-1/2},\ \ \ (b,x)\in(-a_+/2,a_+/2+a_-)\times \R,
\ee
where
\be\label{tA}
\tilde{A}(b;x)\equiv \tilde{V}(b;x)\exp(-ia_-\partial_x) +(x\to -x),
\ee
\be\label{tV}
\tilde{V}(b;x)\equiv \frac{\cosh(\pi (x-ib)/a_+)}{\cosh(\pi x/a_+)},
\ee
and the attractive weight function is given by
\be\label{tilw}
\tilde{w}(b;x)\equiv \prod_{\sigma=+,-}\frac{G(\sigma x+ia_-/2)}{G(\sigma x+ia_-/2-ib)}.
\ee
Hence the attractive Hamiltonian can also be written as
\be\label{tHb} 
\tilde{H}(b;x)  =  \tilde{V}(b;-x)^{1/2}\exp(ia_-\partial_x) \tilde{V}(b;x)^{1/2}+(x\to -x).
\ee
Clearly, the attractive coefficient~$\tilde{V}(b;x)$ arises from the repulsive coefficient~$V(a_+,b;x)$
 by either one of the analytic continuations $x\to x\pm i a_+/2$. Hence the same is true for $\tilde{A}$ and~$A$; for~$\tilde{H}$ and~$H$ this is also the case, as follows from the alternative representations~\eqref{Hb} and~\eqref{tHb}. 
 
Even so, the weight function $\tilde{w}(b;x)$ does not arise in this way from~$w(b;x)$. Indeed, the two functions $w(b;x\pm i a_+/2)$   are not positive for $(b,x)\in(-a_+/2,a_+/2+a_-)\times \R$, whereas $\tilde{w}(b;x)$ does have this positivity feature. (As before, this readily follows from~\eqref{Gades}--\eqref{Gcon}.) On the other hand, we clearly have
 \be\label{deftw}
\tilde{w}(b;x)=1/\tilde{c}(b;x)\tilde{c}(b;-x),
\ee
where
\be\label{deftc}
\tilde{c}(b;x)\equiv G(x+ia_-/2-ib)/G(x+ia_-/2)=c(b;x-ia_+/2).
\ee
In words, the attractive $c$-function does arise from the repulsive $c$-function~\eqref{defc} by one of the analytic continuations at issue. We also point out that the conjugation relation~\eqref{Gcon} entails that the $\tilde{c}$-function has the same conjugation property
\be
\overline{\tilde{c}(b;x)}=\tilde{c}(b;-x),\ \ \ (b,x)\in \R\times\R,
\ee
as the $c$-function; by contrast to the latter, however, it is regular at the origin (for generic~$b$), whereas~$c(b;x)$ has a simple pole.
 
 Because the scale parameter $a_+$ and variable~$x$ are singled out, modular invariance and self-duality are not preserved for the attractive regime. Therefore, it is no longer crucial to insist on the symmetric parametrization of the repulsive regime. Even so, for reasons of notational economy we continue to employ it for most of the main text.
 
In the appendices, however, it is more convenient to work with dimensionless variables~$r$ and~$k$ such that the plane wave~$\exp(i\pi xy/a_+a_-)$ becomes~$\exp(irk)$. Specifically,
we take
\be\label{xryk}
x/a_-\to r/\rho,\ \ \ y/a_-\to k/\kappa,
\ee
with
\be\label{rhoka}
\rho\kappa =\pi a_-/a_+.
\ee
The coupling parameter~$b$ can be traded for a dimensionless parameter~$\tau$ given by
\be\label{tau}
\tau \equiv \pi b/a_+.
\ee
The Hamitonians~$H_s$~\eqref{Hs} and~$H_o$~\eqref{Ho}  then arise from~$H$ and~$\tilde{H}$ by using these dimensionless quantities together with the choice $\kappa=1$.

In Section~2 we begin by summarizing pertinent results on the repulsive eigenfunctions. Their abundance of symmetries leaves no doubt concerning their uniqueness, and indeed there already exists a fully satisfactory Hilbert space theory for $H(b;x)$~\cite{R11,R03III}. 

For the attractive case, however, there is an enormous ambiguity, and it is no longer obvious how to proceed. Once we have put this problem in perspective, we are prepared to survey the previous literature that has a bearing on the key problem of promoting \emph{formally} self-adjoint A$\De$Os to self-adjoint Hilbert space operators, using an explicit unitary eigenfunction transform that yields a concrete realization of the spectral theorem~\cite{RS72}.

After this terse literature overview (which can be found above~\eqref{Rpmades}), we detail how our attractive eigenfunctions result from the repulsive ones by taking a suitable linear combination of the two analytically continued functions~$\cR(a_+,a_-,b;x\pm ia_+/2,y)$. We exemplify this for the simplest case $b=a_+$. This yields the elementary function~\eqref{psizero}, all of whose function-theoretic and asymptotic features can be directly read off. 

In view of the ambiguity in obtaining the opposite-charge eigenfunctions, our choice may seem unmotivated at face value. It is, however, singled out by its very special and desirable features. We have not attempted to show that this renders the choice unique, but we have little doubt that this is true. 

To lend credence to this conviction, we obtain already in this paper the dominant $|\re x|\to\infty$ asymptotics of the joint eigenfunction~$\psi(b;x,y)$~\eqref{defpsi}, cf.~Proposition~2.1. Specifically, together with the scattering function~$u(b;y)$~\eqref{u} of the repulsive case, the resulting attractive transmission and reflection coefficients~$t(b;y)$~\eqref{t} and~$r(b;y)$~\eqref{r} satisfy the Yang-Baxter equations~\eqref{sjk}--\eqref{YB2}. General-$b$ Hilbert space aspects, however, are beyond the scope of this paper.

In the remaining sections we focus on the special $b$-values $(N+1)a_+$. Section~3 contains a detailed account of function-theoretic and asymptotic properties of the corresponding eigenfunctions~$\psi_N(x,y)$~\eqref{psiN}. Their elementary character is inherited from that of the relativistic conical functions~$R_N(x,y)$~\eqref{RN}. The latter were studied in~\cite{R99}. We need various results from this source, which we summarize in the equations~\eqref{RNexp}--\eqref{SNasm}. In the remainder of the section  we obtain the properties of~$\psi_N(x,y)$ that served as a template for the axiomatic account in the appendices. 

In Section~4 we then  apply the Hilbert space results from the appendices to the eigenfunction transforms associated with~$\psi_N(x,y)$, and to the adjoints of the transforms. First, we study the residue sums arising for the transforms, showing that they vanish for the parameter interval~$a_-\in((N+1/2)a_+,\infty)$ in Theorem~4.1. With isometry proved, it follows from Appendix~D that the transforms can be viewed as the incoming wave operators from time-dependent scattering theory. Then we study the adjoint transforms for parameters in this isometry interval. Here a distinction arises: For the subinterval~$a_-\in[(N+1)a_+,\infty)$ we prove in Theorem~4.2 that the associated residue sums vanish, so that the transforms are unitary. By contrast, for~$a_-\in((N+1/2)a_+,(N+1)a_+)$, Theorem~4.4 reveals the presence of a bound state~$\psi_N(x)$, given by~\eqref{psiNx}. As we detail at the end of Section~4, it arises as (a multiple of) the residue of $\psi_N(x,y)$ at its only $y$-pole $i(N+1)a_+-ia_-$ in the strip $\im y\in(0,a_-)$.

After these `constructive' results in Section~4, the final Section~5 furnishes some `destructive' results pertaining to the remaining $a_-$-interval $(0,(N+1/2)a_+]$. In particular, we prove that for $a_-\in(Na_+,(N+1/2)a_+)$ isometry of the transform breaks down in a way that we make quite explicit. For $N>0$ it is not an easy matter to handle any positive $a_-$, and we shall not do so. Indeed, as $a_-$ decreases, there is a cascade of ever larger and more baroque deviations from isometry, much like for the reflectionless cases handled in~\cite{R00}. For $N=0$ we do work out the details, which are summarized in Proposition~5.1.

As already mentioned, the appendices contain an axiomatic account of the functional-analytic aspects, with the assumptions inspired by the features of the eigenfunctions $\psi_N(x,y)$~\eqref{psiN}. Appendix~A deals with a general transform~$\cF$ from `momentum space' to `position space'. The former space is described in terms of two-component square-integrable functions of a (dimensionless) variable $k\in(0,\infty)$, so as to lead to a $2\times 2$ $S$-matrix $S(k)$~\eqref{Sm}. The assumptions are sufficiently restrictive to lead to a picture of the transform as being isometric away from a subspace encoded in residue terms, cf.~Theorem~A.1. In the axiomatic setting, however, it is far from clear whether this subspace is necessarily finite-dimensional, as is the case for the transforms associated with $\psi_N(x,y)$. 

The adjoint transform~$\cF^*$ is analyzed in Appendix~B, with additional assumptions on the $k$-dependence of the plane wave coefficients enabling us to arrive at the counterpart Theorem~B.1 of Theorem~A.1. Simple explicit examples are included in this appendix via Propositions~B.3--B.5, with Proposition B.4 presenting example transforms that go beyond the main text. 

A pivotal technical ingredient in the proofs of Theorems~A.1 and B.1 is relegated to Appendix~C, namely, the asymptotic analysis of the boundary terms arising from the rectangular integration contours at issue in the proofs. 

The paper is concluded with Appendix~D, in which it is shown that when the transforms are unitary, then they may be viewed as the incoming wave operators for a large class of self-adjoint Hilbert space dynamics. More specifically, the transform~$\cF$ from Appendix~A plays this role for a class containing the (Hilbert space version of the) opposite-charge Hamiltonian~$\tilde{H}((N+1)a_+;x)$ given by~\eqref{HN}, whereas for its adjoint~$\cF^*$ the class contains a Hamiltonian that arises from the dual A$\De$O $\cS((N+1)a_+;y)$ given by~\eqref{cSN}. It should be stressed that in the axiomatic framework of Appendix~D the self-adjoint dynamics at issue are \emph{defined} via the transforms. Likewise, since we rely on the results of the appendices to associate self-adjoint operators to the A$\De$Os \eqref{HN} and~\eqref{cSN}, the feasibility of doing so hinges on the \emph{isometry} of their eigenfunction transforms.

Below Theorem~D.1 we also discuss the issue of parity and time reversal symmetry. Both symmetries are present for the transforms arising from the main text, but we leave the question open whether time reversal symmetry holds in the general setting  of the appendices.


\section{Repulsive and attractive eigenfunctions for general coupling}

As  is detailed in~\cite{R11}, the relativistic conical function has a great many distinct representations. For completenes, we quote one of them, namely (cf.~Eq.~(1.3) and Eq.~(1.6) in~\cite{R11}),
\be\label{Rrdef}
\cR(b;x,y)=\frac{G(ia-ib)}{\sqrt{a_+a_-}}  \int_{\R}dz \frac{G(z+ (x-y)/2-ib/2)G(z- (x-y)/2-ib/2)}{G(z+ (x+y)/2+ib/2)G(z- (x+y)/2+ib/2)},
\ee
where we take  $(b,x,y)\in (0,2a)\times \R^2$ to begin with. (Note that the $G$-asymptotics~\eqref{Gas} entails the convergence of the integral.) This is the only known representation that is both manifestly self-dual and modular invariant. Moreover, evenness in~$x$ and~$y$ follows by using the reflection equation~\eqref{refl}, and real-valuedness results from the conjugation relation~\eqref{Gcon}. 

The function
\be\label{FR}
{\mathrm F}(b;x,y) \equiv w(b;x)^{1/2}w(b;y)^{1/2}\cR(b;x,y),\ \ \ (b,x,y)\in(0,2a)\times (0,\infty)^2,
\ee
is also self-dual and modular invariant, so it is a joint eigenfunction of four independent A$\De$Os, namely~$H(a_+,a_-,b;x), H(a_+,a_-,b;y), H(a_-,a_+,b;x), H(a_-,a_+,b;y)$. In fact, it has yet another symmetry, which is not manifest:
It satisfies
\be\label{bsym}
{\mathrm F}(2a-b;x,y) ={\mathrm F}(b;x,y). 
\ee
This can be understood from the four Hamiltonians having this symmetry property. By contrast, the A$\De$O $A(a_+,a_-,b;z)$~\eqref{defA} and the three factors of~${\mathrm F}(b;x,y)$ are not invariant under $b\to 2a-b$.

It follows from previous results~\cite{R11,R03III} that the transform
\be\label{cFs}
(\cF_s(b) f)(x)\equiv \frac{1}{\sqrt{2a_+a_-}} \int_0^{\infty}dy\, {\mathrm F}(b;x,y)f(y),\ \ \ b\in(0,2a),
\ee
yields a unitary transformation from~$L^2((0,\infty),dy)$ onto~$L^2((0,\infty),dx)$. It satisfies
\be\label{Fssym} 
\cF_s(b)^*=\cF_s(b),\ \ \cF_s(2a-b)=\cF_s(b),
\ee
with the self-adjointness following from real-valuedness and self-duality of its kernel, and the $b$-symmetry from~\eqref{bsym}. For $b$ equal to~$a_+$ or~$a_-$ the transform amounts to the sine transform, and its $b=0$ and $b=2a$ limits amount to the cosine transform.

For~$b\in(0,2a)$ the even weight function has a second order zero at the origin (cf.~\eqref{Gze}), so its square root continues from~$(0,\infty)$ to an odd function. Therefore, the transform can also be viewed as a unitary involution on the odd subspace of $L^2(\R)$. More generally, the eigenfunctions for $N$ particles with the same charge are antisymmetric under permutations~\cite{HR14}: Equal-charge particles obey fermionic statistics.

By contrast, a fermion and an antifermion can be distinguished by their charge, so we should aim for eigenfunctions of~$\tilde{H}(b;x)$ that have nontrivial even and odd reductions. This opposite-charge Hamiltonian, however, has no `modular partner',  so we can no longer insist on invariance under interchange of $a_+$ and~$a_-$, a requirement that renders the above same-charge eigenfunction~${\mathrm F}(b;x,y)$ essentially unique. Now any $\tilde{H}(b;x)$-eigenfunction remains an eigenfunction upon multiplication by a function $M(x,y)$ that is meromorphic and $ia_-$-periodic in~$x$ and has an arbitrary $y$-dependence. (This is the ambiguity alluded to in the Introduction.) A priori, however, there is no reason to expect that in this infinite-dimensional eigenfunction space there exist any special ones with the requisite orthogonality and completeness properties. 

More generally, to date there exists no general Hilbert space theory for A$\De$Os. There is, however, a growing supply of explicit (mostly one-variable) A$\De$Os that admit a reinterpretation as bona fide self-adjoint Hilbert space operators. This hinges on the existence of special eigenfunctions that give rise to the unitary transform featuring in the spectral theorem for self-adjoint operators. We continue by surveying the literature in this field, which is quite scant by comparison to the vast literature dealing with Hilbert space aspects of (linear) differential operators and discrete difference operators.

By far the simplest cases are analytic difference operators that admit  
  orthogonal polynomials as eigenfunctions. 
The earliest example is given by the Askey-Wilson polynomials~\cite{AW85}. From the perspective of Calogero-Moser type systems, these can be viewed as the special eigenfunctions of the $BC_1$ relativistic trigonometric Calogero-Moser system that give rise to a reinterpretation of the Askey-Wilson 4-parameter A$\De$O as a self-adjoint Hilbert space operator. Likewise, the multi-variable orthogonal polynomials introduced by Macdonald~\cite{M95} yield the sought-for joint eigenfunctions of the commuting A$\De$Os arising for the $A_{N-1}$  (i.~e., $N$-particle) relativistic trigonometric Calogero-Moser system, and the Koornwinder polynomials those for the $BC_N$ case~\cite{K92, D95}.

A large class of one-variable A$\De$Os yielding reflectionless unitary eigenfunction transforms has been introduced in~\cite{R05}. (The reflectionless A$\De$Os studied in~\cite{R00} form a tiny subclass.) Their eigenfunctions are closely connected to soliton solutions of various nonlocal evolution equations. The unitarity proof for the associated Hilbert space transforms makes essential use of previous results on the connection between the classical hyperbolic relativistic Calogero-Moser $N$-particle systems and the KP and 2D Toda hierarchies.

Further unitary eigenfunction transforms for 4-parameter A$\De$Os of Askey-Wilson type have been constructed in~\cite{KS01, R03III}, and for the elliptic 8-parameter A$\De$Os introduced by van Diejen~\cite{D94} in~\cite{R15}. Other relevant papers are~\cite{R03, R04, FT15}. 

In all of these cases, the eigenfunctions have some very special properties that play a pivotal role for pushing through the associated Hilbert space theory.
In spite of this store of examples,  a general theory has not emerged yet. 

Returning to the problem at hand, it is clear from~\eqref{Rade}--\eqref{defV} and~\eqref{tA}--\eqref{tV} that we have
\be\label{Rpmades}
\tilde{A}(b;x)\cR(b;x\pm ia_{+}/2,y)=2c_+(y) \cR(b;x\pm ia_+/2,y).
\ee
By virtue of~\eqref{tHtA}, this implies that we get two independent $\tilde{H}(b;x)$-eigenfunctions,
\be
\tilde{H}(b;x)\tilde{w}(b;x)^{1/2}\cR(b;x\pm ia_{+}/2,y)=2c_+(y)\tilde{w}(b;x)^{1/2} \cR(b;x\pm ia_+/2,y).
\ee
However, these functions remain eigenfunctions when they are multiplied by functions that are $ia_-$-periodic in~$x$ and that have an arbitrary $y$-dependence, so it is at this point that we need further constraints to reduce the ambiguity. 

To this end, consider the auxiliary function
\be\label{cZ}
\cZ(b;x, y)\equiv \cR(b;x,y)/c(b;-y).
\ee
In view of~\eqref{cA} it satisfies
\be
\cA(b;-y)\cZ(b;x,y)=2c_+(x)\cZ(b;x,y),
\ee
which entails
\be
\cA(b;-y)\cZ(b;x\pm ia_+/2,y)=\pm 2is_+(x)\cZ(b;x\pm ia_+/2,y).
\ee
We now introduce the formally self-adjoint A$\De$O
\be\label{cS}
\cS(b;y)\equiv V(a_+,b;y)\exp(-ia_-\partial_y)V(a_+,b;-y)-\exp(ia_-\partial_y),
\ee
and observe that in view of~\eqref{cAb} we have similarities
\be\label{cScA}
\frac{e_-(\pm (ib-y)/2)}{s_-(ib-y)}\cA(b;-y)\frac{s_-(ib-y)}{e_-(\pm (ib-y)/2)}=\pm i\cS(b;y).
\ee
Therefore, the functions
\be\label{psipm}
\psi_{\pm}(b;x,y)\equiv \pm \tilde{w}(b;x)^{1/2}
\frac{e_-(\pm (ib-y)/2)}{2s_-(ib-y)}\cZ(b;x\pm ia_+/2,y),
\ee
satisfy not only the same $\tilde{H}$-A$\De$E,
\be\label{tHade}
\tilde{H}(b;x)\psi_{\pm}(b;x,y) =2c_+(y)\psi_{\pm}(b;x,y),
\ee
but also the same $\cS$-A$\De$E,
\be\label{cSade}
\cS(b;y)\psi_{\pm}(b;x,y) =2s_+(x)\psi_{\pm}(b;x,y).
\ee
As a consequence, the function
\begin{multline}\label{defpsi}
\psi(b;x,y)\equiv \psi_+(b;x,y)+\psi_-(b;x,y)
\\
=\frac{\tilde{w}(b;x)^{1/2}}{2s_-(ib-y)c(b;-y)}\big( e_-((ib-y)/2)\cR(x+ia_+/2,y)
\\
-e_-((y-ib)/2)\cR(x-ia_+/2,y)\big),
\end{multline}
satisfies both A$\De$Es, too. (We recall $c(b;-y)$ is given by~\eqref{defc}, and $\tilde{w}(b;x)$ by~\eqref{tilw}.)

As will transpire, the function~$\psi(b;x,y)$ we just defined is the sought-for attractive eigenfunction. We shall substantiate this for the special cases $b=(N+1)a_+$, $N\in\N$, in the present paper, whereas the Hilbert space theory for the general-$b$ case will be dealt with in a companion paper. The special cases $b=(N+1)a_-$, $N\in\N$, have been treated before in~Section~4 of~\cite{R00}, but in view of the above-mentioned ambiguity  it is not immediate that the functions occurring there are basically the same as the functions $\psi((N+1)a_-;x,y)$. At the end of Section~3 we shall show that this holds true.

To be sure, at this point it is far from clear that even for the simplest cases $b=a_+$ and $b=a_-$ the eigenfunction $\psi(b;x,y)$ has all of the desired properties. We continue by working out the details for these $b$-values, since this involves little effort and the resulting formulas are illuminating.  

As announced above, for~$b=a_+$ we are dealing with an equal-charge (reduced) 2-particle system that is `free', in the sense that no scattering occurs. Specifically, we have (cf.~Eq.~(4.7) in~\cite{R11})
\be\label{N0}
\cR(a_+;x,y)=\frac{\sin(\pi xy/a_+a_-)}{2s_-(x)s_-(y)}.
\ee
Also, \eqref{defc} and~\eqref{Gades} imply
\be\label{cwsp}
c(a_+;z)=1/2is_-(z),\ \ \ w(a_+;z)=4s_-(z)^2,
\ee
so that \eqref{FR} yields
\be
\rF(a_+;x,y)=2\sin(\pi xy /a_+a_-).
\ee
This entails that $\cF_s(a_+)$~\eqref{cFs} amounts to the sine transform, as mentioned before.

On the other hand, from~\eqref{deftc} and~\eqref{deftw} we have
\be
\tilde{c}(a_+;x)=1/2is_-(x-ia_+/2),\ \ \tilde{w}(a_+;x)=4s_-(x+ia_+/2)s_-(x-ia_+/2).
\ee
Also, \eqref{cZ} gives
\be\label{cZsp}
\cZ(a_+;x,y)=-i\sin(\pi xy /a_+a_-)/s_-(x),
\ee
so from~\eqref{defpsi} we obtain
\begin{multline}\label{psizero}
\psi(a_+;x,y)  =  \frac{(s_-(x+ia_+/2)s_-(x-ia_+/2))^{1/2}}{2s_-(ia_+-y)}\sum_{\tau=+,-}\tau 
\frac{e_-(\tau (y-ia_+)/2)}{s_-(x-i\tau a_+/2)}
  \\
  \times \big(e_-(\tau y/2) \exp(i\pi xy/a_+a_-)-e_-(-\tau y/2) \exp(-i\pi xy/a_+a_-).
\end{multline}
With the substitutions~\eqref{xryk} and~\eqref{rhoka} in force, this yields the wave function~$\Psi(r,k)$ featuring in Prop.~B.3, cf.~\eqref{cFp0}.

Turning to the special case $b=a_-$, we can again make use of the formulas~\eqref{N0}, \eqref{cwsp} and~\eqref{cZsp}, but now with $a_+$ and~$a_-$ swapped. Then, \eqref{deftc} and~\eqref{deftw} yield the quite different outcome
\be
\tilde{c}(a_-;x)=1/2c_+(x),\ \ \ \tilde{w}(a_-;x)=4c_+(x)^2.
\ee
Using~\eqref{defpsi}, we now obtain the attractive wave function
\be\label{psifree}
\psi(a_-;x,y)=\exp(i\pi xy/a_+a_-),
\ee
which is manifestly `free', just as its repulsive counterpart.

Before specializing to the $b$-values~$(N+1)a_+$ for~$N>0$, it is expedient to derive already in this paper the general-$b$ asymptotic behavior of~$\psi(b;x,y)$ as $|\re x|\to\infty$. Indeed, this will illuminate how we arrived at the above $y$-dependence. 

To begin with, the $G$-function asymptotics~\eqref{Gas}
implies that the $c$-function~\eqref{defc} satisfies
\be\label{cas}
c(b;z)\sim \phi(b)^{\pm 1}\exp(\mp \pi bz/a_+a_-),\ \ \ \re z\to\pm\infty,
\ee
where we have introduced the constant
\be\label{phib}
\phi(b)\equiv \exp(i\pi b(b-2a)/2a_+a_-).
\ee
Introducing next
\be\label{u}
u(b;z)\equiv -c(b;z)/c(b;-z)=-\prod_{\de=+,-}G(z+i\de (a-b))/G(z+ i\de a),
\ee
we deduce
\be\label{uas}
u(b;z)\sim -\phi(b)^{\pm 2},\ \ \ \ \re  z\to  \pm\infty.
\ee
Also, the reflection equation~\eqref{refl} and the complex conjugation relation~\eqref{Gcon}
entail
\be
u(b;-z)u(b;z)=1,\ \ \ |u(b;z)|=1,\ \ b,z\in\R.
\ee
The function~$u(b;y)$ encodes the scattering associated with the A$\De$O $\cA (a_{+},a_-,b;x)$ and its modular partner, reinterpreted as commuting self-adjoint operators on the Hilbert space $L^2((0,\infty),dx)$. (To be sure, this reinterpretation requires $b\in[0,2a]$, cf.~\cite{R11}.)

The assertion just made hinges on the asymptotic behavior 
of the joint eigenfunction~$\rE(b;x,y)$ of these A$\De$Os, defined by 
\be\label{ER}
\rE(b;x,y)\equiv \frac{\phi(b)}{c(b;x)c(b;y)}\cR(b;x,y).
\ee
 By contrast to $\cR(b;x,y)$ and~$\rF(b;x,y)$, this function is not even, but satisfies
\be\label{Erefl}
\rE(b;-x,y)=-u(b;x)\rE(b;x,y).
\ee
It follows from~Theorem~1.2 in~\cite{R03II} that the $\rE$-function has asymptotics 
\be\label{EasR}
\rE(b;x,y)\sim e^{i\pi xy/a_+a_-}- u(b;-y)e^{-i\pi xy/a_+a_-},\ \ (b,y)\in\R\times(0,\infty), \ \ \re x\to \infty.
\ee
(The specialization of the `$BC_1$'-functions from~\cite{R03II,R03III} to the `$A_1$'-functions of this paper is detailed in Section~2 of~\cite{R11}.) Using \eqref{Erefl} and~\eqref{uas}, this yields
\be\label{EasL}
\rE(b;x,y)\sim \phi(b)^2\big( e^{-i\pi xy/a_+a_-}- u(b;-y) e^{i\pi xy/a_+a_-}\big),\  \ (b,y)\in\R\times(0,\infty), \  \re x\to -\infty.
\ee

To complete our preparation for the last result of this section, we define transmission and reflection coefficients by
\be\label{t}
t(b;y)\equiv \frac{s_-(y)}{s_-(ib-y)}u(b;y),
\ee 
\be\label{r}
r(b;y)\equiv \frac{s_-(ib)}{s_-(ib-y)}u(b;y).
\ee
This entails that when we set
\be\label{sjk}
s_{jk}\equiv s(y_j-y_k),\ \ \ s=u,t,r,\ \ \ 1\le j<k\le 3,
\ee
then the well-known $(u,t,r)$-Yang-Baxter equations given by
\be\label{YB1}
r_{12}t_{13}u_{23}=t_{23}u_{13}r_{12}+r_{23}r_{13}t_{12},
\ee
and
\be\label{YB2}
u_{12}r_{13}u_{23}=t_{23}r_{13}t_{12}+r_{23}u_{13}r_{12},
\ee
are easily verified. We recall that these equations encode the consistent factorization of the multi-particle $S$-matrix into a product of two-particle scattering amplitudes~\cite{YB90}.  

\begin{proposition}
The dominant large-$|\re x|$ asymptotic behavior of $\psi(b;x,y)$ is given by
\be\label{psias}
\psi(b;x,y)\sim
\left\{
\begin{array}{ll}
t(b;y)\exp(i\pi xy/a_+a_-),  &  \re x\to\infty,  \\
\exp(i\pi xy/a_+a_-)-r(b;y)\exp(-i\pi xy/a_+a_-) ,  & \re x\to-\infty ,
\end{array}
\right.
\ee
where $b\in(-a_+/2,a_-+a_+/2)$ and $y>0$. Furthermore, we have an identity
\be\label{psirev}
\psi(b;x,y)=t(b;y)\psi(b;-x,-y)-r(b;y) \psi(b;x,-y).
\ee
\end{proposition}
\begin{proof}
Using the evenness of $\cR(b;x,y)$ in~$x$ and~$y$, the identity~\eqref{psirev} follows from~\eqref{defpsi} by a straightforward calculation. (It encodes `time reversal invariance', cf.~the discussion below Theorem~D.1.)

In order to prove~\eqref{psias}, we first note that from \eqref{cZ} we have  
\be
\cZ(b;x,y)=-\phi(b)^{-1}c(b;x)u(b;y)\rE(b;x,y).
\ee
Combining this relation with~\eqref{cas}, \eqref{EasR} and~\eqref{EasL}, we obtain
\be
\cZ(b;x,y)\sim
\left\{
\begin{array}{ll}
e^{-\pi bx/a_+a_-}\big(e^{-i\pi xy/a_+a_-}-u(b;y)e^{i\pi xy/a_+a_-} \big),  &  \re x\to\infty,  \\
e^{\pi bx/a_+a_-}\big(e^{i\pi xy/a_+a_-}-u(b;y)e^{-i\pi xy/a_+a_-} \big),  & \re x\to-\infty .
\end{array}
\right.
\ee
Now from \eqref{deftw}, \eqref{deftc} and \eqref{cas} we get
\be
\tilde{w}(b;x)^{1/2}\sim \exp(\pm \pi bx/a_+a_-),\ \ \ \re x\to\pm \infty,
\ee
where we need $b\in(-a_+/2,a_-+a_+/2)$ to ensure $\tilde{w}(b;0)>0$. From this we deduce
\begin{multline}
\tilde{w}(b;x)^{1/2}\cZ(b;x\pm ia_+/2,y)\sim \\
e_-(\mp i b/2)\big(e_-(\pm y/2)e^{-i\pi xy/a_+a_-}-e_-(\mp y/2)u(b;y)e^{i\pi xy/a_+a_-} \big),\ \re x\to\infty,
\end{multline}
\begin{multline}
\tilde{w}(b;x)^{1/2}\cZ(b;x\pm ia_+/2,y)\sim \\
e_-(\pm i b/2) \big(e_-(\mp y/2)e^{i\pi xy/a_+a_-}-e_-(\pm y/2)u(b;y)e^{-i\pi xy/a_+a_-} \big) ,\ \re x\to-\infty.
\end{multline} 
Hence \eqref{psipm} entails
\begin{multline}
2\psi_{\pm}(b;x,y)s_-(ib-y)\sim \\
\pm\big(e^{-i\pi xy/a_+a_-}-e_-(\mp y)u(b;y)e^{i\pi xy/a_+a_-} \big),\  \ \ \re x\to\infty,  
\end{multline} 
\begin{multline}
2\psi_{\pm}(b;x,y)s_-(ib-y)\sim \\
 \pm e_-(\pm ib)\big(e_-(\mp y)e^{i\pi xy/a_+a_-}-u(b;y)e^{-i\pi xy/a_+a_-} \big) , \ \ \   \re x\to-\infty .
 \end{multline}  
Recalling \eqref{defpsi}, we now obtain~\eqref{psias}.
\end{proof}


\section{A close-up of the special eigenfunctions $\psi_N(x,y)$}

In this section we specialize the coupling constant~$b$ to the sequence of values $(N+1)a_+$, $N\in\N$, and study the opposite-charge eigenfunctions
\be\label{psiN}
\psi_N(x,y)\equiv \psi((N+1)a_+;x,y),\ \ \ N\in\N.
\ee
They are obtained from the same-charge eigenfunctions
\be\label{RN}
R_N(x,y)\equiv \cR((N+1)a_+;x,y),\ \ \ N\in\N,
\ee
via~\eqref{defpsi}. By contrast to the case of a generic coupling $b\in\R$, the $c$-function and $\tilde{w}$-function featuring in these formulas are elementary periodic functions, given by
\be\label{crec}
1/c((N+1)a_+;-y)=\prod_{j=0}^N (-2i)s_-(y+ija_+),
\ee
\be\label{twN}
\tilde{w}((N+1)a_+;x)=\prod_{j=0}^N 4s_-(x+i(j+1/2)a_+)s_-(x-i(j+1/2)a_+).
\ee
(This easily follows from the $G$-A$\De$Es~\eqref{Gades}.)

The crux of the special $b$-values is that the functions \eqref{RN} are elementary functions, too. Specifically, from Eqs.~(4.8) and (4.5) in~\cite{R11} we have
\be\label{RNexp}
R_N(x,y)=(-i)^{N+1}(K_N(x,y)-K_N(x,-y))\Big/\prod_{j=-N}^N 4s_-(x+ija_+)s_-(y+ija_+),
\ee
with (cf.~Eqs.~(4.10)--(4.12) in~\cite{R11})
\be\label{KN}
K_N(x,y)  \equiv \exp(i\pi xy/a_+a_-)\Sigma_N(x,y),
\ee
\be\label{SigN}
\Sigma_N(x,y)\equiv\sum_{k,l=0}^Nc_{kl}^{(N)}(e_-(ia_+))e_-((N-2k)x+(N-2l)y)).
\ee
The coefficients in the sum are specified in Section~II of~\cite{R99}, where the functions $K_N(x,y)$ and $R_N(x,y)$ have been studied in great detail. (See also~\cite{DK00} for a related account.)  
In the sequel we need substantial information concerning the above quantities, which we proceed to collect. 

First, the coefficients $c_{kl}^{(N)}(q)$ are Laurent polynomials in~$q$ with integer coefficients. In particular, we have~$c^{(0)}_{00}=1$ (so that $R_0(x,y)$ reduces to~\eqref{N0}), and
\be\label{ckl1}
c^{(1)}_{00}=c^{(1)}_{11}=q,\ \ c^{(1)}_{01}=c^{(1)}_{10}=-q^{-1},\ \ \ q=e_-(ia_+).
\ee

Second, the coefficients have symmetries 
 \be\label{csym}
c^{(N)}_{kl}=c^{(N)}_{lk}=c^{(N)}_{N-k,N-l}=(-)^N\overline{c^{(N)}_{k,N-l}},\ \ \ k,l=0,\ldots,N.
\ee
Clearly, these are equivalent to the following features of $K_N$ and~$\Sigma_N$:
\be\label{Sigsd}
K_N(x,y)=K_N(y,x),\ \ \ \Sigma_N(x,y)=\Sigma_N(y,x),
\ee
\be\label{Sigrefl}
K_N(x,y)=K_N(-x,-y),\ \ \ \Sigma_N(x,y)=\Sigma_N(-x,-y),
\ee
\be\label{Sigcon}
K_N(x,y)=(-)^N\overline{K_N(x,-y)},\ \ \Sigma_N(x,y)=(-)^N\overline{\Sigma_N(x,-y)},\ \ \ x,y\in\R.
\ee

Third, the function $K_N(x,y)$ satisfies the A$\De$E
\be\label{KNade}
s_-(x+iNa_{+})K_N(x-ia_+,y)+s_-(x-iNa_{+})K_N(x+ia_+,y)=2s_-(x)c_-(y)K_N(x,y).
\ee
(This is a similarity transformed version of the A$\De$E~\eqref{Rade} with~$b=(N+1)a_+$.)

Fourth, we need the following explicit evaluations:
\be\label{Kspec}
K_N(x,\pm iNa_+)=K_N(\pm iNa_+,y)=\prod_{j=N+1}^{2N}2s_-(ija_+),
\ee
\be\label{KBy}
K_N(\pm i(N-k)a_+,y)=i^NB^{(N)}_k(c_-(y)),\ \ \ k=0,\ldots,N,
\ee
\be\label{KBx}
K_N(x,\pm i(N-k)a_+)=i^NB^{(N)}_k(c_-(x)),\ \ \ k=0,\ldots,N.
\ee
Here, $B^{(N)}_k(u)$ is a polynomial of degree~$k$ and parity~$(-)^k$ with real coefficients, and the restrictions
\be\label{arestr}
ja_+\notin a_-\N,\ \ \ j=1,\ldots, 2N,
\ee
must be imposed for the degree-$k$ property to hold true.
 In particular, from~\eqref{ckl1} we easily get
\be\label{Bsp}
K_1(\pm ia_+,y)=2s_-(2ia_+),\ \ K_1(0,y)=4s_-(ia_+)c_-(y),
\ee
and the restrictions $a_+\neq la_-/2$, $l\in\N$, ensure that the coefficients do not vanish.

The above features can all be gleaned from~Theorem~II.1 in~\cite{R99}. Finally, we also need the summation identity
\be\label{SNasp}
\sum_{l=0}^Nc_{0l}^{(N)}e_-((N-2l)y)=\prod_{j=1}^N2s_-(y+ija_+),
\ee
which results upon combining Eqs.~(2.19)--(2.21) with Eq.~(2.55) in~\cite{R99}. From this we obtain a second identity
\be\label{SNasm}
\sum_{l=0}^Nc_{Nl}^{(N)}e_-((N-2l)y)=(-)^N\prod_{j=1}^N2s_-(y-ija_+),
\ee
by using $c^{(N)}_{Nl}=c^{(N)}_{0,N-l}$ and relabeling.
 
We are now in the position to work out an explicit expression for $\psi_N(x,y)$ by using the above building blocks. From this and the summation identities~\eqref{SNasp} and~\eqref{SNasm} we can then determine the asymptotic behavior of~$\psi_N(x,y)$. First, however, we  do so for 
\be\label{EN}
\rE_N(x,y)\equiv \rE ((N+1)a_+;x,y),
\ee
since this yields a simple template for the application of the above formulas. 

From~\eqref{phib} we calculate
\be\label{phiN}
  \phi((N+1)a_+)=(-i)^{N+1}e_-(i(N+1)Na_+/2),
\ee
and then~\eqref{ER} and~\eqref{crec} yield
\be\label{ENexp}
\rE_N(x,y)=e_-(i(N+1)Na_+/2)  
(K_N(x,y)-K_N(x,-y))\Big/\prod_{j=1}^N 4s_-(x+ija_+)s_-(y+ija_+).
\ee
Combining \eqref{KN} and \eqref{SigN} with~\eqref{SNasp}, it is now easy to verify
\be\label{ENas}
\rE_N(x,y)= e^{i\pi xy/a_+a_-}-u_N(-y)e^{-i\pi xy/a_+a_-}+O(e_-(-2 x)),\ \ \ \re x\to\infty,
\ee
where
\be\label{uN}
u_N(z)\equiv u((N+1)a_+;z)=\prod_{j=1}^N \frac{s_-(ija_++z)}{s_-(ija_+-z)},
\ee
and where the implied constant can be chosen uniformly for~$\im x$ varying over $\R$-compacts and~$y$ varying over compact subsets of $\C$ with the $y$-poles removed.
This agrees with the specialization of~\eqref{EasR}, and by using~\eqref{SNasm} we readily obtain a more precise version of~\eqref{EasL} as well.

Turning to $\psi_N(x,y)$, we begin by defining a weight function~$w_N(x)$ that plays the role of $w(r)$ in Appendix~A. It is given by (recall~\eqref{twN})
\be\label{wN}
w_N(x)\equiv 1/\tilde{w}((N+1)a_+;x)=1\Big/ \prod_{j=0}^N 4s_-(x+i(j+1/2)a_+)s_-(x-i(j+1/2)a_+).
\ee
Next, we introduce a function $v_N(y)$ that corresponds to the function $v(k)$ from Appendix~B. To ease the notation, we define its reciprocal:
\be\label{vyrec}
1/v_N(y)\equiv  2s_-(y-i(N+1)a_+)c(((N+1)a_+;-y)\prod_{j=-N}^N2s_-(y+ija_+).
\ee
Using~\eqref{crec} and simplifying, this yields
\be\label{vhy}
v_N(y)=1\Big/\prod_{j=1}^{N+1}2is_-(y-ija_+).
\ee
From~\eqref{defpsi} we now obtain, using the above formulas,
\begin{multline}\label{psiNexp}
\psi_N(x,y)=(-)^Ni^{N+1}w_N(x)^{1/2}v_N(y) \\
\times\sum_{\de=+,-}2\de s_-(x-i\de (N+1/2)a_+)e_-(\de(i(N+1)a_+-y)/2)\\
\times \Big( \exp(i\pi xy/a_+a_-)e_-(-\de y/2)\Sigma_N(x+i\de a_+/2,y) -(y\to -y)\Big).
\end{multline}
As a check, note this coincides with \eqref{psizero} for $N=0$. 

We proceed to derive properties of~$\psi_N(x,y)$ from these explicit formulas. A first conclusion is that we have 
\be\label{psiNcon}
\overline{\psi_N(x,y)}=\psi_N(x,-y),\ \ \ (x,y)\in\R^2.
\ee
Indeed, this readily follows by combining \eqref{wN}--\eqref{psiNexp} with the conjugation relation~\eqref{Sigcon} for $\Sigma_N$. 

Next, we obtain the large-$|\re x|$ asymptotics of $\psi_N(x,y)$. Just as for $E_N(x,y)$, the identities \eqref{SNasp} and \eqref{SNasm} are the key ingredients. A little more effort is needed due to the extra factors, but it is still straightforward to deduce
\be\label{psiNasx}
\psi_N(x,y)=
\left\{
\begin{array}{ll}
t_N(y)e^{i\pi xy/a_+a_-}+O(e_-(-2x)),   &  \re x\to\infty,  \\
e^{i\pi xy/a_+a_-}-r_N(y)e^{-i\pi xy/a_+a_-} +O(e_-(2x)) ,  & \re x\to-\infty ,
\end{array}
\right.
\ee
with
\be\label{tNrN}
t_N(y)\equiv\frac{s_-(y)}{s_-(i(N+1)a_+-y)}u_N(y),\ \ \ r_N(y)\equiv\frac{s_-(i(N+1)a_+)}{s_-(i(N+1)a_+-y)}u_N(y),
\ee
where the implied constants can be chosen uniformly for~$\im x$ varying over $\R$-compacts and~$y$ varying over $\C$-compacts with the $y$-poles removed. This renders the dominant general-$b$ asymptotics given by~\eqref{t}--\eqref{psias} more precise.

We can also use the dual versions of~\eqref{SNasp}--\eqref{SNasm} (i.~e., $y\to x$ and $c_{0l}/c_{Nl}\to c_{l0}/c_{lN}$) to obtain the large-$|\re y|$ asymptotics of $\psi_N(x,y)$. This again needs a bit of work, the result being
\be\label{psiNasy}
\psi_N(x,y)=
\left\{
\begin{array}{ll}
w_N(x)^{1/2}C_N(x)\exp(i\pi xy/a_+a_-)+O(e_-(-y) ),  &  \re y\to\infty,  \\
w_N(x)^{1/2}\overline{C_N(x)}\exp(i\pi xy/a_+a_-)+O(e_-(y) ),  & \re y\to-\infty ,
\end{array}
\right.
\ee
where
\be\label{CN}
C_N(x)\equiv (-)^{N+1}e_-(i(N+1)^2a_+/2)
\prod_{j=0}^N2s_-(x+i(j+1/2)a_+),
\ee
and the implied constants are uniform for~$\im y$ in $\R$-compacts and~$x$ varying over $\C$-compacts with the $x$-singularities removed. Thus we obtain a reflectionless asymptotics of the form
\be
\psi_N(x,y)\sim U_N(x)^{\pm 1/2} \exp(i\pi xy/a_+a_-),\ \ \ \re y\to\pm\infty,
\ee
with
\be\label{UN}
U_N(x)\equiv e_-(i(N+1)^2a_+)
\prod_{j=0}^N\frac{s_-(x+i(j+1/2)a_+)}{s_-(x-i(j+1/2)a_+)}.
\ee

We stress that there are no restrictions on $a_+$ and $a_-$ in~\eqref{psiNasx}--\eqref{UN}, save for our standing positivity assumption. This reveals a major flaw of time-independent (as opposed to time-dependent) scattering theory in the present setting: The unitary asymptotics exhibited by these formulas does not show any anomaly. However, as we shall see later on, we must \emph{restrict} the scale parameters for (a suitable multiple of) $\psi_N(x,y)$ to yield the integral kernel of a \emph{unitary} eigenfunction transform.

In order to clarify the unitarity issue (`orthogonality and completeness' in quantum-mechanical parlance), we rely on the results in the appendices. To make contact with their setup,  we need to verify the various assumptions made there. To this end we rewrite~\eqref{psiNexp} as
\be\label{psiNalt}
\psi_N(x,y)=w_N(x)^{1/2}v_N(y)\sum_{\tau=+,-}\ell_N^{\tau}(x,y)\exp(i\tau \pi xy/a_+a_-),
\ee
so that we obtain entire coefficients
\begin{multline}\label{ellNm}
\ell_N^{\tau}(x,y)=(-)^Ni^{N+1} \tau \sum_{\de=+,-}2\de s_-(x-i\de (N+1/2)a_+)
\\
\times e_-(\de(i(N+1)a_+-y )/2)e_-(-\de\tau y/2)
\Sigma_N(x+i\de a_+/2,\tau y).
\end{multline}
From this we read off the (anti)periodicity features
\be\label{perx}
\ell_N^{\tau}(x+ia_-,y)=(-)^{N+1}\ell_N^{\tau}(x,y),\ \ \ \tau=+,-,
\ee
\be\label{pery}
\ell_N^{\tau}(x,y+ia_-)=\tau(-)^{N+1}\ell_N^{\tau}(x,y),\ \ \ \tau=+,-.
\ee
 Also, from~\eqref{wN} it is clear that~$w_N(x)$ is $ia_-$-periodic, while~\eqref{vhy} shows that we have
 \be\label{vNper}
  v_N(y+ia_-)=(-)^{N+1}v_N(y).
  \ee

When we now make the substitutions~\eqref{xryk}, then it follows from this that the assumptions in Appendix~A and~B concerning periodicity/antiperiodicity in the variables~$r$ and~$k$ are satisfied. Likewise, \eqref{Psiconj} is clear from~\eqref{psiNcon}, and the assumptions regarding asymptotic behavior in~$r$ and~$k$ are satisfied, too.  
Finally, we show that the critical evenness assumptions hold true. 

\begin{proposition}
With the substitutions~\eqref{xryk} in $\ell^{\pm}_N(x,y)$~\eqref{ellNm} and $v_N(y)$~\eqref{vhy}, the evenness assumptions \eqref{Lass} and~\eqref{Mass} hold true.
\end{proposition}
\begin{proof}
Staying with the variables~$x$ and~$y$, 
\eqref{Lass} amounts to 
invariance of
\be\label{LNtt}
L^{\tau,\tau'}_N(y,x,x')\equiv \ell_N^{\tau}(x,y)\ell_N^{\tau'}(x',-y)+\ell_N^{-\tau}(-x,y)\ell_N^{-\tau'}(-x',-y),
\ee
under taking $y,x,x'\to -y,-x,-x'$. To check this symmetry property, we suppress the subscripts in~\eqref{ellNm} and use the notation
\be
\eta_N\equiv i(N+1/2)a_+,\ \ \ \nu=i(N+1)a_+/2.
\ee
Then, \eqref{ellNm} entails
\begin{multline}
\ell^{\tau}(x,y)\ell^{\tau'}(x',-y)\sim \tau\tau'\sum_{\de,\de'}\de\de' s(x-\de\eta_N)s(x'-\de'\eta_N)
\Sigma(x+\de \eta_0,\tau y)
\Sigma(x'+\de' \eta_0,-\tau' y)
\\
\times e((\de'\tau'-\de\tau)y/2)
e((\de+\de')\nu)e((\de'-\de)y/2).
\end{multline}
From this we deduce by using~\eqref{Sigrefl},
\begin{multline}\label{Lttex}
L^{\tau,\tau'}_N(y,x,x')\sim
\sum_{\de,\de'}\de\de' s(x-\de\eta_N)s(x'-\de'\eta_N)
\Sigma(x+\de \eta_0,\tau y)
\Sigma(x'+\de' \eta_0,-\tau' y)
\\
\times e((\de'\tau'-\de\tau)y/2)
\big[e((\de+\de')\nu)e((\de'-\de)y/2)+(\de,\de'\to -\de,-\de')\big]
\\
=\sum_{\de}s(x-\de\eta_N)s(x'-\de\eta_N)
\Sigma(x+\de \eta_0,\tau y)
\Sigma(x'+\de \eta_0,-\tau' y)
e(\de(\tau'-\tau)y/2)\cdot 2c(2\nu)
\\
-\sum_{\de}s(x-\de\eta_N)s(x'+\de\eta_N)
\Sigma(x+\de \eta_0,\tau y)
\Sigma(x'-\de \eta_0,-\tau' y)
e(-\de(\tau'+\tau)y/2)\cdot 2c(y).
\end{multline}
Invoking~\eqref{Sigrefl} once more, we now see that both sums are invariant under taking $y,x,x'$ to~$-y,-x,-x'$. 
Having verified~\eqref{Lass}, the weaker assumptions~\eqref{Mass}   follow, too.  
\end{proof}

As promised below~\eqref{defpsi}, we conclude this section by deriving explicit expressions for the attractive eigenfunctions~$\psi((N+1)a_-;x,y)$ with $N>0$, and comparing them to Section~4 in~\cite{R00}. (Recall we already calculated~$\psi(a_-;x,y)$, cf.~\eqref{psifree}.) To this end we can make use of~\eqref{crec} and~\eqref{RNexp}--\eqref{SigN} with~$a_+$ and~$a_-$ swapped, whereas~\eqref{twN} is replaced by
\be\label{twNrep}
\tilde{w}((N+1)a_-;x)=2c_+(x)\prod_{j=-N}^N 2c_+(x-ija_-).
\ee
Employing these formulas in combination with~\eqref{defpsi}, we obtain
\begin{multline}\label{psiNr=0}
\psi((N+1)a_-;x,y)=\exp(i\pi xy/a_+a_-)\\
\times
\frac{\sum_{k,l=0}^N(-)^kc_{kl}^{(N)}(e_+(ia_-))e_+((N-2k)x+(N-2l)y))}{\prod_{j=1}^N2s_+(y-ija_-)[4c_+(x-ija_-)c_+(x+ija_-)]^{1/2} }.
\end{multline}
We point out that these functions are not only joint eigenfunctions of $\tilde{H}((N+1)a_-;x)$ and~$\cS((N+1)a_-;y)$ (cf.~\eqref{tHade} and~\eqref{cSade}), but also of the two `free' A$\De$Os
\be\label{2free}
\exp(-ia_+\partial_x)+\exp(ia_+\partial_x),\ \ \ \exp(-ia_+\partial_y)+\exp(ia_+\partial_y).
\ee
(Indeed, the functions multiplying the plane wave are $ia_+$-periodic in~$x$ and~$y$.) By contrast, for other $b$-values no additional independent A$\De$Os appear to exist for which $\psi(b;x,y)$ are eigenfunctions.

The connection with the functions $F_a(\nu,\beta;x,p)$ given by Eq.~(1.32) of~\cite{R00} can be made by setting
\be\label{RIMS}
a_+=\pi/\nu,\ \ a_-=\hbar\beta,\ \ y=\beta p/\nu.
\ee
Then the dimensionless parameter $a=\hbar\beta\nu$ used there equals $\pi a_-/a_+$, and the coefficient matrices are related by
\be 
(-)^mc_{mn}^{(N)}(e_+(ia_-))=i^Nc^{(a)}_{mn},\ \ \ m,n=0,\ldots,N,
\ee
cf.~Eq.~(1.39) in~\cite{R00}. With these reparametrizations, we wind up with the relation
\be\label{psiFa}
F_a(\nu,\beta;x,p)=(-i)^N\psi((N+1)a_-;x,y)\prod_{j=1}^N[s_+(y-ija_-)/s_+(y+ija_-)]^{1/2}.
\ee

We mention that in~\cite{DK00} these reflectionless eigenfunctions were tied in with basic hypergeometric series. It is an interesting question whether this link can still be made for the eigenfunctions~$\psi((N+1)a_+;x,y)$.


\section{The transforms associated with $\psi_N(x,y)$}
 
In this section we focus on the Hilbert space aspects of the opposite-charge A$\De$Os
\be\label{HN}
\tilde{H}_N(x)\equiv \tilde{H}((N+1)a_+;x)=\exp(-ia_-\partial_x) +\exp(ia_-\partial_x),
\ee
\be\label{cSN}
 \cS_N(y)\equiv \cS((N+1)a_+;y)=\exp(-ia_-\partial_y)  -\exp(ia_-\partial_y),
\ee
(cf.~\eqref{tHb} and~\eqref{cS}), and their joint eigenfunctions~$\psi_N(x,y)$~\eqref{psiNexp}. To this end we invoke the results in Appendix~A and~B, with $\psi_N(x,y)$ (reparametrized by~\eqref{xryk})  in the role of $\Psi(r,k)$. We are entitled to do so, since we have shown in the previous section that the assumptions made there are satisfied for generic scale parameters $a_+,a_-$. We have not isolated the exceptional parameters yet, and now proceed by studying this issue. 

 To start with, it  easily follows from \eqref{psiNexp} that the transform resulting from~\eqref{cF}--\eqref{Fm} is bounded whenever $w_N(x)$ and $v_N(y)$ have no pole at the origin. Indeed, the integral kernels are then given by a finite sum of terms of the form
\be\label{cFform}
B(r)\exp(i\tau rk)\hat{B}(k),\ \ \ \tau \in \{ +,-\},
\ee
where $B(r)$ and  $\hat{B}(k)$ are bounded functions on~$\R$ and~$[0,\infty)$, resp. Clearly, each such term gives rise to a product of three bounded operators.

The only eventual constraint on $a_+$ and $a_-$ encountered thus far comes from the   condition `no poles at the origin' just mentioned. We are using the adjective `eventual', since we do not know whether it can ever happen that the wave function~$\psi(b;x,y)$ itself has a pole for~$x=0$ or~$y=0$. We can illustrate this with~$\psi_0(x,y)$~\eqref{psizero}: Its factors are singular at the origin if and only if $a_+\in a_-\N^*$, but in fact we easily calculate
\be\label{psi0free}
\psi_0(x,y)=(-)^{l-1}\exp(i\pi xy/a_+a_-),\ \ \ a_+=la_-,\ \ \ l\in\N^*.
\ee
 
More generally, whenever the parameter~$b$ is simultaneously a multiple of~$a_+$ and of~$a_-$ (so that the factors of~$\psi_N(x,y)$ have poles at the origin), the function $\psi_N(x,y)$ is actually a (constant) phase multiple of the plane wave~$\exp(i\pi xy/a_+a_-)$. This can be concluded from the features of the functions $F_a(\nu,\beta;x,p)$~\eqref{psiFa} expounded at the end of~Section~4 in~\cite{R00}. 

Fortunately our ignorance about this absence of real poles is of little consequence, as we now explain. First, it is clear from its definition~\eqref{wN} that $w_N(x)$ has no real poles for 
\be\label{unitary}
a_->(N+1/2)a_+=b-a_+/2.
\ee
For the remainder of this section, we restrict attention to this parameter interval.
As we shall show in the next one, for $a_-\le (N+1/2)a_+$ the transform associated with~$\psi_N(x,y)$ is not isometric, save for a discrete parameter set, hence not acceptable from a quantum-mechanical viewpoint. 

Second, from~\eqref{vhy}  we read off that~$v_N(y)$ has no real poles for~$a_->(N+1)a_+$. For~$a_-=(N+1)a_+$, however, $\psi_N(x,y)$ reduces to the plane wave~\eqref{psifree}, so the corresponding transform amounts to the Fourier transform. For the remaining interval~$a_-\in ((N+1/2)a_+,(N+1)a_+)$ we do not get any real poles in~$v_N(y)$, so the upshot is that there is no real pole problem for the parameters~\eqref{unitary} at issue in this section.

In Appendix~A and~B, however, we have additional assumptions concerning poles, which amount to the poles of the weight functions~$w_N(x)$ and~$v_N(y)v_N(-y)$ being simple. Clearly, in the interval~\eqref{unitary} $w_N(x)$ has double poles in the critical strip~$\im x\in (0,a_-)$ whenever the poles at $x=i(j+1/2)a_+$, $j=0,\ldots,N$, collide with the poles at $x=ia_--i(k+1/2)a_+$, $k=0,\ldots,N$. However, this is of no consequence for the first major result of this section.

\begin{theorem}
With the parameter restriction~\eqref{unitary} in effect, the transform~$\cF_N(\rho,\kappa)$ given by~\eqref{cF}--\eqref{Fm} with kernel
\be
\Psi(r,k) =\psi_N(a_- r/\rho,a_- k/\kappa ),\ \ \ \rho\kappa =\pi a_-/a_+,
\ee
is an isometry.
\end{theorem}
\begin{proof}
 The requirement~\eqref{unitary} amounts to
\be\label{rkint}
\rho\kappa\in ((N+1/2)\pi,\infty).
\ee
There is a finite set~$\cE_N$ of exceptional $\rho\kappa$-values in this interval for which at least one pole of (cf.~\eqref{wN})
\be\label{wr}
w(r)=1\Big/ \prod_{j=0}^N 4\sinh\Big(\frac{\pi }{\rho}\Big( r+i(j+1/2)\frac{\pi}{\kappa}\Big)\Big)
\sinh\Big(\frac{\pi }{\rho}\Big( r-i(j+1/2)\frac{\pi}{\kappa}\Big)\Big),
\ee
is not simple. But the bounded multiplication operators featuring in the transform kernel (cf.~\eqref{cFform}) are clearly strongly continuous in $\rho$ on the whole interval~\eqref{rkint}, so the same is true for the transform~$\cF_N(\rho,\kappa)$. Thus we need only show
isometry for the non-exceptional subintervals to conclude isometry for all of~\eqref{rkint}.

To prove isometry for
\be
\rho\kappa \in ((N+1/2)\pi,\infty)\setminus \cE_N,
\ee
we may invoke Theorem~A.1, since all of its assumptions are satisfied, with $L$ equal to $N+1$. In view of its Corollary~A.2, it suffices to show that each of the~$N+1$ terms in the sum on the rhs of~\eqref{Rdk} vanishes. 

To this end, consider the dependence on the indices~$\nu$ and~$\nu'$ of one of the terms. For it to vanish, it is sufficient that the bracketed expression on the second line of~\eqref{Rdk} is proportional to~$\nu\nu'$. Indeed, assuming it is, we can invoke the identity~\eqref{id1} to infer that the term vanishes.

Now in the case at hand, we have
\be
 \mu^{\tau}(\rho x/a_-,\kappa y/a_-)=
 v_N(y)\ell^{\tau}_N(x,y)\exp(i\tau \pi xy/a_+a_-),\ \ \ \tau=+,-,
 \ee
where~$\ell^{\tau}_N(x,y)$  is given by~\eqref{ellNm}. Recalling~\eqref{KN}, we see that the only dependence on~$\tau$ comes from the factor
\be
\tau K_N(x+i\de a_+/2,\tau y).
\ee
 Moreover, for the present case we may choose as pole locations
\be\label{rj}
r_j\equiv i\pi (j+1/2)/\kappa\Leftrightarrow x_j\equiv i(j+1/2)a_+,\ \ \ j=0,\ldots,N.
\ee
Thus we only encounter the values
\be\label{Kval}
K_N(\pm ima_+,\tau y),\ \ \ m=0,\ldots,N+1,
\ee
with the case $m=N+1$ corresponding to~$x_N$ and $\de=+$. In the latter case, however, the factor~$s_-(x-i\de(N+1/2)a_+)$ in~\eqref{ellNm} vanishes. 

With the parameter restrictions~\eqref{arestr} in force, we can invoke~\eqref{KBy} to conclude that the values~\eqref{Kval} with~$m<N+1$ are even in~$y$. By continuity this is still true for the excluded parameters. Hence these values do not depend on~$\tau$.

The upshot is that the two $\mu^{\tau}$-products in~\eqref{Rdk} solely depend on~$\nu$ and~$\nu'$ via a factor $\nu\nu'$, so that all terms in the sum vanish.
\end{proof}

The proof of the above theorem only made use of Appendix~A. To obtain the next theorem, we need the results of Appendix~B.

\begin{theorem}
Assume
\be\label{parunit}
a_-\in[(N+1)a_+,\infty).
\ee
Then the transform~$\cF_N(\rho,\kappa)$ of Theorem~4.1 is unitary. 
\end{theorem}
\begin{proof}
By Theorem~4.1 the transform is isometric, and it equals the Fourier transform at the interval endpoint, cf.~\eqref{psifree}. Letting next
\be
a_-\in ((N+1)a_+,\infty) \Leftrightarrow \rho\kappa \in ((N+1)\pi,\infty),
\ee
there is a finite set~$\hat{\cE}_N$ of exceptional $\rho\kappa$-values in this interval for which at least one pole of (cf.~\eqref{vhy})
\be\label{hatw}
\hat{w}(k)=\Big( \prod_{j=1}^{N+1} 4\sinh\big(\frac{\pi k}{\kappa}+ij\frac{\pi^2}{\rho\kappa}\big)
\sinh\big(\frac{\pi k}{\kappa}-ij\frac{\pi^2}{\rho\kappa}\big) \Big)^{-1},
\ee
is not simple. By continuity, we need only prove unitarity for
\be
\rho\kappa \in ((N+1)\pi,\infty)\setminus \hat{\cE}_N.
\ee
For such parameters all assumptions of Theorem~B.1 are satisfied, with~$\hat{L}=N+1$.

As a consequence, it suffices to prove that each of the $N+1$ terms on the rhs of~\eqref{defR} vanishes. We choose as pole locations
\be\label{kj}
k_j\equiv i\pi j/\rho\Leftrightarrow y_j\equiv ija_+,\ \ \ j=1,\ldots,N+1.
\ee
Now from~\eqref{ellNm} we get
\begin{multline}\label{lamN}
\lambda^{\tau}_N(x,y)\equiv \exp(i\tau\pi xy/a_+a_-)\ell_N^{\tau}(x,y) =(-)^Ni^{N+1} \tau \sum_{\de=+,-}2\de s_-(x-i\de (N+1/2)a_+)
\\
\times e_-(\de(i(N+1)a_+-y )/2) 
K_N(x+i\de a_+/2,\tau y).
\end{multline}
The dependence on $\tau$ and $\tau'$ of the two $\lambda$-products in~\eqref{defR} is therefore given by an overall factor $\tau\tau'$ and by the arguments of the two $K_N$-functions.
But just as in the previous proof, it follows from~\eqref{KBx} that the values
\be\label{Kvalres}
K_N(x,\pm \tau ija_+),\ \ \ j=1,\ldots,N,
\ee
do not depend on~$\tau$.
In view of the identity~\eqref{id1} with $A,A'\to iA,iA'$,  we then deduce that  the terms in the sum on the rhs of~\eqref{defR} vanish for $j=1,\ldots,N$.

Next, we set
\begin{multline}\label{Lam}
\Lambda_N^{\tau,\tau'}(y,x,x')\equiv \lambda_N^{\tau}(x, y )\lambda_N^{\tau'}(x',-y  )+\lambda_N^{-\tau}(- x,y )\lambda_N^{-\tau'}(-x',- y )
\\
=L_N^{\tau,\tau'}(y,x,x')\exp(i\pi (\tau x-\tau' x')y/a_+a_-),
\end{multline}
and consider~$\Lambda_N^{\tau,\tau'}(y_{N+1},x,x')$. From~\eqref{lamN} we get
\begin{multline}
\Lambda_N^{\tau,\tau'}(y_{N+1},x,x')=(-)^{N+1}8\tau\tau' c_-(i(N+1)a_+)\sum_{\de,\de'=+,-}\de\de' s_-(x-i\de (N+1/2)a_+)\\
\times  s_-(x'-i\de' (N+1/2)a_+) 
 K_N(x+i\de a_+/2,\tau y_{N+1})K_N(x'+i\de' a_+/2,-\tau' y_{N+1}),
\end{multline}
where we have used \eqref{Sigrefl} to rewrite the second $\lambda$-product. (Note that this agrees with~\eqref{Lttex}.) Invoking~\eqref{lamN} once more, this can be rewritten as
\be\label{Lamalt}
\Lambda_N^{\tau,\tau'}(y_{N+1},x,x')=-2c_-(i(N+1)a_+)\lambda_N^{\tau}(x,i(N+1)a_+)\lambda_N^{-\tau'}(x',i(N+1)a_+).
\ee

We continue by evaluating~$\lambda_N^{\tau}(x,i(N+1)a_+)$ explicitly.  First, from~\eqref{lamN} we have
\be\label{lamid}
\lambda_N^{\tau}(x,i(N+1)a_+)= (-i)^{N+1}\tau p_N,
\ee
where we have introduced
\be\label{pN}
p_N\equiv 2\sum_{\nu=+,-}\nu s_-(x+i\nu (N+1/2)a_+)K_N(x-i\nu a_+/2,i\tau (N+1)a_+).
\ee
We are suppressing the manifest $\tau$- and $x$-dependence of the entire functions on the rhs for the new quantity~$p_N$, since it does not depend on these variables. We shall prove this claim shortly. Taking it for granted, we may take $x=i(N+1/2)a_+ $ and invoke~\eqref{Kspec} to obtain the product formula
\be\label{pNexp}
p_N=\prod_{j=N+1}^{2N+1}2s_-(ija_+).
\ee
Hence we have
\be
\lambda_N^{\tau}(x,i(N+1)a_+)=\tau  \prod_{j=N+1}^{2N+1}2\sin(\pi ja_+/a_-),
\ee
so that
\be\label{Lamspec}
\Lambda_N^{\tau,\tau'}(y_{N+1},x,x')=2\tau\tau'\cos(\pi(N+1)a_+/a_-)\prod_{j=N+1}^{2N+1}4\sin(\pi ja_+/a_-)^2.
\ee
Thus the~$j= N+1$ term in the sum on the rhs of~\eqref{defR} vanishes, too, and the theorem follows.

It remains to prove our constancy claim. To this end, consider the $x$-dependence of the sum $S(x)$ on the rhs of~\eqref{pN}. It is clear from the definitions~\eqref{KN} and~\eqref{SigN} that $S(x)$ has period~$2ia_-$. Next, we have
\begin{multline}
S(x+ia_+/2)-S(x-ia_+/2)=\\
s_-(x+i(N+1)a_+)K_N(x,\tau y_{N+1} )
-s_-(x-iNa_+)K_N(x+i a_+,\tau y_{N+1})\\
-s_-(x+iNa_+)K_N(x-i a_+,\tau y_{N+1})+s_-(x-i(N+1)a_+)K_N(x,\tau y_{N+1} ) ,
\end{multline}
so from the A$\De$E~\eqref{KNade} we deduce that the rhs equals
\be
\big( s_-(x+i(N+1)a_+)+s_-(x-i(N+1)a_+)-2s_-(x)c_-(y_{N+1})\big) K_N(x,\tau y_{N+1} )=0.
\ee
Therefore $S(x)$ has period~$ia_+$, too. Taking $a_+/a_-$ irrational, it follows that $S(x)$ is constant. By virtue of real-analyticity in $x,a_+,a_-$, it then follows that $S(x)$ can only depend on $a_+$ and~$a_-$, and so our claim is proved.
\end{proof} 

For parameters satisfying~\eqref{parunit}, this theorem enables us to promote the A$\De$Os~\eqref{HN} and~\eqref{cSN} to self-adjoint Hilbert space operators. Indeed, for this parameter range all assumptions of Appendix~D are satisfied, so that we may associate to $\tilde{H}_N(x)$ the operator~$M_{CM}$ (cf.~\eqref{muCM} and~\eqref{HCM}), and to $\cS_N(y)$ the operator~$\hat{D}_{CM}$ (cf.~\eqref{dCM} and~\eqref{HhCM}). From this the following result is nearly immediate.

\begin{corollary}
Assume $a_+,a_-$ satisfy~\eqref{parunit}. Then the self-adjoint operator $M_{CM}$ associated to the A$\De$O $\tilde{H}_N(x)$ has absolutely continuous spectrum $[2,\infty)$ with multiplicity two, whereas the self-adjoint operator~$\hat{D}_{CM}$ associated to the A$\De$O $\cS_N(y)$  has absolutely continuous spectrum $(-\infty,\infty)$ with multiplicity one. 
\end{corollary}
\begin{proof}
Indeed, the unitary transform and its adjoint intertwine the respective operators with multiplication operators for which these spectral features are plain. 
\end{proof}

Even though the action of the above operators on the functions in their domains is just the free A$\De$O action, they are quite different from the Hilbert space versions of the A$\De$Os defined via the Fourier transform~$\cF_0$, as evidenced by the nontrivial $S$-operators at hand. 

A further comment on this distinction may be in order: It arises from the vastly different definition domains. In particular, viewing~$M_{CM}$ as a self-adjoint operator on $L^2(\R,dx)$  by undoing the reparametrization~\eqref{xryk},   its definition domain consists of functions that have an analytic continuation to the strip $|\im x|<a_-$, but the locations of their square-root branch points for~$x\in i(-a_-,a_-)$ (which can be read off from the pertinent eigenfunction transform) entail that a pairwise intersection of domains for different parameters consists of the zero function. This domain behavior is radically different from that for a pair of distinct Hilbert space versions~$\hat{D}_1,\hat{D}_2$ of the same free (i.~e., constant-coefficient) differential operator~$D$. Indeed, typically their Hilbert space domains are encoded in distinct boundary conditions, so that the domain intersection is still a dense subspace, but neither a core for~$\hat{D}_1$ nor for~$\hat{D}_2$.

We proceed with the last theorem of this section. 

\begin{theorem}
Letting
\be\label{parbs}
a_-\in((N+1/2)a_+,(N+1)a_+),
\ee
the transform~$\cF_N(\rho,\kappa)$ of Theorem~4.1
satisfies
\be\label{cFN}
\cF_N(\rho,\kappa)^*\cF_N(\rho,\kappa)={\bf 1}_{ \hat{\cH }},\ \ \ \cF_N(\rho,\kappa)\cF_N(\rho,\kappa)^*={\bf 1}_{ \cH }-\Psi_N\otimes \Psi_N/(\Psi_N,\Psi_N)_1.
\ee
Here, we have
\be\label{PsiN}
\Psi_N(r)\equiv 2\cosh(\kappa r) w_N(a_- r/\rho)^{1/2},
\ee
with~$w_N(x)$ defined by~\eqref{wN}, and the inner product is given by
\be\label{ipN}
(\Psi_N,\Psi_N)_1=(-)^{N+1}\frac{\pi\prod_{j=1}^N\sin (j\pi^2/\rho\kappa)}{\kappa\prod_{j=N+1}^{2N+1}\sin (j\pi^2/\rho\kappa)},\ \ \ \ \ \ \rho\kappa\in((N+1/2)\pi,(N+1)\pi).
\ee
\end{theorem}
\begin{proof}
Proceeding as in the proof of Theorem~4.2, we choose again as simple pole locations in the strip $\im k\in (0,\kappa)$ the numbers
\be\label{kjres}
k_j\equiv ij\pi /\rho\Leftrightarrow y_j\equiv ija_+,\ \ \ j=1,\ldots,N,
\ee
but now we need to choose
\be\label{kspec}
k_{N+1}\equiv i (N+1)\pi /\rho-i\kappa\Leftrightarrow y_{N+1}\equiv i(N+1)a_+-ia_-,
\ee
since we have $(N+1)\pi/\rho\in (\kappa,2\kappa)$. Thus the residue terms with $j=1,\ldots,N$ vanish as before, but we need to reconsider the $j=N+1$ contribution.

Now from \eqref{lamN} and~\eqref{pery} we obtain
\be\label{lamper}
\lambda^{\tau}_N(x,y-ia_-)=\tau (-)^{N+1}e_+(\tau x)\lambda^{\tau}_N(x,y),
\ee
so \eqref{Lam} entails
\be
\Lambda_N^{\tau,\tau'}(y-ia_-,x,x') = \tau\tau'e_+(\tau x-\tau'x')\Lambda_N^{\tau,\tau'}(y,x,x').
\ee
On account of~\eqref{Lamspec} we then get
\be\label{Lambs}
\Lambda_N^{\tau,\tau'}(i(N+1)a_+-ia_-,x,x')= 2e_+(\tau x-\tau' x')\cos(\pi(N+1)a_+/a_-)\prod_{j=N+1}^{2N+1}4\sin(\pi ja_+/a_-)^2.
\ee

The upshot is that the residue sum~\eqref{defR} reduces to
\begin{multline}\label{defRN}
R(r,r')=2i\hat{w}_{N+1}\cos((N+1)\pi^2/\rho\kappa)\prod_{j=N+1}^{2N+1}4\sin( j\pi^2/\rho\kappa )^2 \\
\times\sum_{\tau,\tau'=+,-}\frac{ \exp(\kappa(\tau r-\tau' r')}{1-\tau\tau'\exp(\kappa (\tau' r'-\tau r))} .
\end{multline}
Next, from \eqref{hatw} we calculate
\be
\hat{w}_{N+1}=(-)^N\frac{\kappa}{i\pi}\frac{\sin((N+1)\pi^2/\rho\kappa)}{\prod_{j=1}^{2N+2}2\sin(j\pi^2/\rho\kappa)},
\ee
and via the identity resulting from~\eqref{AAp} and~\eqref{S0} we obtain the sum in~\eqref{defRN}. Simplifying, we get the final result
\be
R(r,r')=(-)^N\frac{4\kappa}{\pi}\cosh(\kappa r)\cosh(\kappa r')\frac{\prod_{j=N+1}^{2N+1}\sin (j\pi^2/\rho\kappa)}{\prod_{j=1}^N  \sin(j\pi^2/\rho\kappa)}.
\ee
From this we arrive at~\eqref{cFN} and~\eqref{ipN}  by the same reasoning as for the special case $N=0$, cf.~\eqref{Psi0}--\eqref{inpr0}.
\end{proof}

To conclude this section, we derive further information concerning the bound state $\Psi_N(r)$ given by~\eqref{PsiN}. It is expedient to do so in its guise (cf.~\eqref{xryk})
\begin{multline}\label{psiNx}
\psi_N(x)\equiv \Psi_N(\rho x/a_-)= 2c_+(x)w_N(x)^{1/2}
\\
=2c_+(x)\Big/\prod_{j=0}^{N}\big[4s_-(x-i(j+1/2)a_+)s_-(x+i(j+1/2)a_+)\big]^{1/2}.
\end{multline}
Let us consider the poles of $\psi_N(x,y)$ in the strip $\im y\in(0,a_-)$. Recalling~\eqref{vhy} and~\eqref{psiNexp}, we see that poles of $\psi_N(x,y)$ can only occur at the locations $y\equiv ija_+$ (mod $ia_-$), with $j=1,\ldots ,N+1$. 

Requiring first $a_-\in ((N+1)a_+,\infty)$, it follows that the only eventual pole locations in the critical strip are the numbers
\be\label{yj}
y_j := ija_+,\ \ \ j=1,\ldots,N+1.
\ee
Now $\psi_N(x,y)$ has factors $K_N(x\pm ia_+/2,y)-K_N(x\pm ia_+/2,-y)$ (cf.~\eqref{psiNexp}), and on account of~\eqref{KBx}, these factors vanish at $y_j$ for $j=1,\ldots,N$. Therefore $\psi_N(x,y)$ is regular at $y_1,\ldots,y_N$. 
Next, we recall that we have
\be
\psi_N(x,y)=w_N(x)^{1/2}v_N(y)\sum_{\tau=+,-}\lambda^{\tau}_N(x,y).
\ee
In view of the identity~\eqref{lamid}, this entails that $\psi_N(x,y)$ is regular at $y_{N+1}$ as well.

The upshot is that $\psi_N(x,y)$ is holomorphic in the critical strip $\im y\in(0,a_-)$ for the $a_-$-interval featuring in Theorem~4.2. (Recall $\psi_N(x,y)$ is equal to the plane wave~\eqref{psifree} at the endpoint.)

For the $a_-$-interval~\eqref{parbs} of Theorem~4.4, however, the pole locations $y_1,\ldots,y_N$   are still in the critical strip, whereas~$y_{N+1}$ is not. Obviously, $\psi_N(x,y)$ is still regular at $y_1,\ldots,y_N$. But now we also have a pole of $v_N(y)$ in the strip at $y_{N+1}-ia_-$. We can determine its residue by combining~\eqref{lamper} and~\eqref{lamid}. Indeed, from these identities  we deduce
\be
\lambda_N^{\tau}(x,y_{N+1}-ia_-)=i^{N+1}p_Ne_+(\tau x),\ \ \tau=+,-.
\ee
Therefore, the residue is proportional to the bound state~\eqref{psiNx}. To be specific, from~\eqref{vhy} we calculate
\be
{\rm Res}\  v_N(y)\Big|_{y=y_{N+1}-ia_-}=i^{N+1}a_-/\pi \prod_{j=1}^N 2s_-(ija_+),
\ee
and together with~\eqref{pNexp} this yields
\be
{\rm Res}\  \psi_N(x,y)\Big|_{y=y_{N+1}-ia_-}=
(-)^{N+1} \frac{ia_-\prod_{j=N+1}^{2N+1}\sin (ja_+/a_- )}{\pi\prod_{j=1}^N  \sin(ja_+/a_- )}\psi_N(x).
\ee

The eigenvalue of the A$\De$O $\tilde{H}_N(x)$ \eqref{HN} on $\psi_N(x)$~\eqref{psiNx} is therefore given by
\be
2c_+(i(N+1)a_+-ia_-)=2(-)^{N+1}\cos(\pi a_-/a_+).
\ee
(This can be easily checked directly.)  Of course, we define the operator on $L^2(\R,dx)$ associated to $\tilde{H}_N(x)$  to have the same eigenvalue, and then the last result of this section easily follows.

\begin{corollary}
Assuming \eqref{parbs}, the self-adjoint operator associated to $\tilde{H}_N(x)$ has a nondegenerate positive eigenvalue
\be
 E_N\equiv 2(-)^{N+1}\cos(\pi a_-/a_+)\in (0,2),\ \ \ a_-\in((N+1/2)a_+,(N+1)a_+),
 \ee
  below its absolutely continuous spectrum $[2,\infty)$ with multiplicity~2. The corresponding eigenfunction~$\psi_N(x)$~\eqref{psiNx} has norm given by~\eqref{ipN}. 
\end{corollary}

  In view of Theorem~D.1, Theorem~4.2 and Theorem~4.4,  the self-adjoint Hamiltonian on~$L^2(\R,dx)$ associated to the A$\De$O $\tilde{H}_N(x)$ yields a well-defined time-dependent scattering theory for $a_->(N+1/2)a_+$, with its $S$-operator encoded in the transmission and reflection coefficients $t_N(y)$ and~$r_N(y)$ given by~\eqref{tNrN} and~\eqref{uN}. Note that the `Jost function' $\psi_N(x,y)/t_N(y)$ converges to a multiple of the bound state~\eqref{psiNx} as~$y$ converges to $i(N+1)a_+-ia_-$.

\section{Isometry breakdown}

As we have seen in the proof of Theorem 4.3, the reason for the transform~$\cF$ not being unitary is the presence of a nonzero residue term. It emerges when $a_-$ decreases beyond the critical value $(N+1)a_+$, at which the transform reduces to the Fourier transform~$\cF_0$. Since the adjoint $\cF^*$ of~$\cF$ is not isometric for $a_-\in((N+1/2)a_+,(N+1)a_+)$, we can no longer use~$\cF^*$ to associate a self-adjoint operator to $\cS_N(y)$~\eqref{cSN}. Even so, since $\cF$ is still an isometry in this $a_-$-interval, $\cF^*$ is still a partial isometry.

As we shall make clear in this section, for generic $a_-\in (0,(N+1/2)a_+)$ isometry of the eigenfunction transform $\cF$ breaks down, so that we cannot use it any longer to associate a self-adjoint operator to $\tilde{H}_N(x)$~\eqref{HN}. This is due to   nonzero residue terms that emerge  when the numbers (cf.~\eqref{rj})
\be
x_j:=i(j+1/2)a_+,\ \ \ j=0,\ldots,N,
\ee
move out of the strip $\im x\in(0,a_-)$ as $a_-$ decreases. More precisely, the nonzero terms are spawned by poles $x_j-in_ja_-$, with $n_j>0$ chosen such that these locations are in the critical strip. (Except for the $N=0$ case, we do not consider the $a_-$-values such that the origin is among the locations $x_j-ila_-$, $l\in\N^*$.)  

We proceed to analyze the state of affairs for the $N=0$ case in complete detail. (On first reading, the reader may wish to skip to Proposition~5.1, in which the results are summarized.) We first recall that the transform $\cF_0(\rho,\kappa)$ of Theorem~4.1 equals the transform~$\cF_+(\phi_0)$ of Proposition~B.3, and that we have special cases
\be
\cF_0(\rho,\kappa)=(-)^{l-1}\cF_0,\ \ \ \  \ a_-=a_+/l\Leftrightarrow \rho\kappa=\pi/l,\ \ \ l\in\N^*,
\ee
cf.~\eqref{psi0free}. Thus the transform is not only isometric for $a_+<2a_-$, but also for~$a_+$ equal to an arbitrary multiple of~$a_-$. 

Consider now the remaining $\rho\kappa$-intervals
\be\label{Ipm}
 I_n^-\equiv \Big(\frac{\pi}{2n+2},\frac{\pi }{2n+1}\Big),\ \  I_n^+\equiv \Big(\frac{\pi }{2n+1},\frac{\pi }{2n}\Big),\ \ \ \ n\in\N^*.
\ee
We have (cf.~\eqref{Psi0})
\be
w_0(r)= 1\Big/4\sinh \Big(\frac{\pi}{\rho}\big( r+i\frac{\pi}{2\kappa}\big)\Big)\sinh \Big(\frac{\pi}{\rho}\big( r-i\frac{\pi}{2\kappa}\big)\Big),
\ee
so we get a simple $w_0(r)$-pole in the critical strip $\im r\in(0,\rho)$ at the location
\be
r_1^{(n)}:= i\pi/2\kappa -in\rho\ne i\rho/2,
\ee
with residue
\be
w_1=\frac{-i\rho}{4\pi \sin(\pi^2/\rho\kappa)}.
\ee
Also, from~\eqref{cFp0} we read off the summands $\mu^{\pm}(r,k)$, and then a straightforward calculation yields
\be\label{mu0res}
\mu^{\tau}(\de r_1^{(n)},k)=\tau (-)^{n+1}\exp(\de\tau n\rho k)\sinh\big(i\pi^2/\rho\kappa\big)\frac{\exp\big(\de(\pi k/2\kappa-i\pi^2/2\rho\kappa)\big) }{\sinh(\pi k/\kappa-i\pi^2/\rho\kappa)}.
\ee
Therefore, the residue sum~\eqref{Rdk} becomes
\begin{multline}
\hat{R}_{\de,\de'}(k,k')= \frac{\rho\sin(\pi^2/\rho\kappa)\hat{s}_n(k,k')}{4\pi\sinh\big(\pi k/\kappa+i\pi^2/\rho\kappa\big)\sinh\big(\pi k'/\kappa-i\pi^2/\rho\kappa\big) }
\\
\times \Big(\exp\Big(\de\Big(\frac{\pi k}{2\kappa}+\frac{i\pi^2}{2\rho\kappa}\Big)-\de'\Big(\frac{\pi k'}{2\kappa}-\frac{i\pi^2}{2\rho\kappa}\Big)\Big) +(\de,\de'\to-\de,-\de')\Big),
\end{multline}
where
\bea
\hat{s}_n(k,k') & \equiv & \sum_{\nu,\nu'=+,-}\frac{\nu\nu'\exp(n\rho(\nu k-\nu' k'))}{1-\exp(\rho(\nu' k'-\nu k))}
\nonumber \\
  & = & \frac{\sinh((n+1/2)\rho(k-k'))}{\sinh( \rho(k-k')/2)}-
   \frac{\sinh((n+1/2)\rho(k+k'))}{\sinh( \rho(k+k')/2)}.
   \eea
Now we have a recurrence
\bea
(\hat{s}_n-\hat{s}_{n-1})(k,k') & = & 2\cosh(n\rho (k-k'))-2\cosh(n\rho (k+k'))
\nonumber \\
 & = & -4\sinh(n\rho k)\sinh(n\rho k'),
 \eea
whence we easily deduce
\be
\hat{s}_n(k,k')=-4\sum_{j=1}^n \sinh(j\rho k)\sinh(j\rho k').
\ee
Introducing 
\be
\chi_{\de}^{(\alpha,j)}(k)\equiv \frac{\sinh(j\rho k)}
{\sinh(\pi k/\kappa+i\pi^2/\rho\kappa)}\exp\Big(\alpha\de\Big(\frac{\pi k}{2\kappa}+\frac{i\pi^2}{2\rho\kappa}\Big)\Big),\ \ \ \alpha,\de=+,-,
\ee
we therefore obtain
\be\label{Rchi}
\hat{R}_{\de,\de'}(k,k')=- \frac{\rho\sin(\pi^2/\rho\kappa)}{\pi}\sum_{\alpha=+,-}\sum_{j=1}^n
\chi_{\de}^{(\alpha,j)}(k)\overline{ \chi_{\de'}^{(-\alpha,j)}(k')}.
\ee

Finally, we use the parity operator~$\hat{\cP}$~\eqref{pari} to first substitute
\be
\chi^{(-,j)}=-\hat{\cP}\chi^{(+,j)},
\ee
and then we rewrite~\eqref{Rchi} in terms of the even and odd functions
\bea\label{chie}
\chi^{(e,j)}(k) & \equiv & \chi^{(+,j)}(k)+\hat{\cP}\chi^{(+,j)}(k)
 \nonumber \\
&  = & \frac{\sinh(j\rho k)}
{\cosh(\pi k/2\kappa+i\pi^2/2\rho\kappa)} \left(\begin{array}{c}
1 \\ 
-1 
\end{array}\right),\ \ \  \ \ k>0,
\eea
\bea\label{chio}
\chi^{(o,j)}(k) & \equiv & \chi^{(+,j)}(k)-\hat{\cP}\chi^{(+,j)}(k)
 \nonumber \\
&  = & \frac{\sinh(j\rho k)}
{\sinh(\pi k/2\kappa+i\pi^2/2\rho\kappa)} \left(\begin{array}{c}
1 \\ 
1 
\end{array}\right),\ \ \ \ \ k>0.
\eea
 As a result, we obtain the manifestly self-adjoint rank-$(2n)$ kernel 
\be\label{hatRfin}
\hat{R}_{\de,\de'}(k,k')= \frac{\rho\sin(\pi^2/\rho\kappa)}{2\pi} \sum_{j=1}^n\Big(
\chi_{\de}^{(e,j)}(k)\overline{ \chi_{\de'}^{(e,j)}(k')}-\chi_{\de}^{(o,j)}(k)\overline{ \chi_{\de'}^{(o,j)}(k')}\Big), 
\ee
with $\de,\de'=+,-$, and $ k,k'>0$.

We continue by calculating the residue sum~$R(r,r')$~\eqref{defR}. The dual weight function is given by
\be\label{wod}
\hat{w}_0(k)=1\Big/4\sinh \Big(\frac{\pi}{\kappa}\big( k+i\frac{\pi}{\rho}\big)\Big)\sinh \Big(\frac{\pi}{\kappa}\big( k-i\frac{\pi}{\rho}\big)\Big),
\ee
and we consider its simple poles in the critical strip $\im k\in(0,\kappa)$ of the form
\be
k_1^{(m)}=i\pi/\rho -im\kappa \ne i\kappa/2,\ \ \ m\in\N^*,
\ee
with residue
\be\label{hatres}
\hat{w}_1=\frac{-i\kappa}{4\pi \sin(2\pi^2/\rho\kappa)}.
\ee
More specifically, we choose $\rho\kappa$ in the interval (cf.~\eqref{Ipm})
\be\label{Im}
I_m\equiv \Big(\frac{\pi}{m+1},\frac{\pi}{m}\Big)
=\left\{\begin{array}{cc}
I_{m/2}^+, & m\ \mathrm{even}, \\ 
I_{(m-1)/2}^-, & m\ \mathrm{odd}, 
\end{array}\right.
\ee
 with~$\rho\kappa\ne \pi/(m+1/2)$ to avoid the double pole location. Using~\eqref{Psi0} we calculate
\be
\lambda^+(r,k_1^{(m)})=2(-)^me^{m\kappa r}\sin(\pi^2/\rho\kappa),
\ee
\be
\lambda^+(r,-k_1^{(m)})=2(-)^me^{-m\kappa r}\big(\sin(2\pi^2/\rho\kappa)-e^{2\pi r/\rho}\sin(\pi^2/\rho\kappa)\big),
\ee
\be
\lambda^-(r,k_1^{(m)})=-2\sin(\pi^2/\rho\kappa)e^{-m\kappa r},
\ee
\be
\lambda^-(r,-k_1^{(m)})=-2\sin(\pi^2/\rho\kappa)e^{m\kappa r}e^{-2\pi r/\rho},
\ee
and then the definition~\eqref{Lamdef} yields
\be\label{Lameq}
\Lambda^{\tau,\tau}(k_1^{(m)},r,r')=4 \sin(2\pi^2/\rho\kappa)\sin(\pi^2/\rho\kappa)\exp(\tau m\kappa( r-r')),
\ee
\be\label{Lamdi}
\Lambda^{\tau,-\tau}(k_1^{(m)},r,r')= (-)^{m+1}4\sin(2\pi^2/\rho\kappa)\sin(\pi^2/\rho\kappa)\exp (\tau m\kappa( r+r')).
\ee
(Note this checks with~\eqref{Lamspec} for $m=0$ and with~\eqref{Lambs} for~$m=1$.)

The residue sum~$R(r,r')$~\eqref{defR} can therefore be written
\be
R(r,r')=\frac{\kappa\sin(\pi^2/\rho\kappa)}{\pi}w_0(r)^{1/2}w_0(r')^{1/2}s_m(r,r'),
\ee
where
\bea
s_m(r,r') & \equiv & \sum_{\tau=+,-}\Big( \frac{\exp(\tau m\kappa(r-r'))}{1-\exp(\tau\kappa(r'-r))}+(-)^{m+1}\frac{\exp(\tau m\kappa(r+r'))}{1-\exp(-\tau\kappa(r'+r))}
\nonumber \\
  & = & \frac{\sinh((m+1/2)\kappa(r-r'))}{\sinh( \kappa(r-r')/2)}+(-)^{m+1}
   \frac{\cosh((m+1/2)\kappa(r+r'))}{\cosh( \kappa(r+r')/2)}.
   \eea
This sum obeys the recurrence
\bea
(s_m-s_{m-1})(r,r') & = & 2\cosh(m\kappa (r-r'))+(-)^{m+1}2\cosh(m\kappa (r+r'))
\nonumber \\
 & = & \left\{\begin{array}{cc}
 -4\sinh(m\kappa r)\sinh(m\kappa r'), & m\ \mathrm{even}, \\ 
4\cosh(m\kappa r)\cosh(m\kappa r') , & m\ \mathrm{odd}, 
\end{array}\right. 
 \eea
whose unique solution reads
\be
s_m(r,r')=\sum_{j=0}^{m-1} (-)^jh_j(r)h_j(r'),
\ee
with
\be\label{hj}
h_j(r)=\left\{\begin{array}{cc}
 2\sinh(j\kappa r), & j\ \mathrm{odd}, \\ 
2\cosh(j\kappa r) , & j\ \mathrm{even}. 
\end{array}\right.
\ee
As a consequence, we obtain the manifestly self-adjoint rank-$m$ kernel
\be\label{Rfin}
R(r,r')=\frac{\kappa\sin(\pi^2/\rho\kappa)}{\pi}\sum_{j=0}^{m-1} (-)^j\Psi^{(j)}(r)\Psi^{(j)}(r'),\ \ \ \rho\kappa\in I_m\setminus \{ \pi/(m+1/2)\},
\ee
where
\be\label{Psij}
\Psi^{(j)}(r)\equiv h_j(r)w_0(r)^{1/2}.
\ee

We now summarize the above $N=0$ results.

\begin{proposition}
Letting~$\rho\kappa\le\pi/2$, the transform $\cF_0(\rho,\kappa)=\cF_+(\phi_0)$ of Proposition~B.3 has the following properties.  For~$\rho\kappa$ in the intervals~$I_n^-$ and~$I_n^+$ defined by~\eqref{Ipm}, it satisfies
\be\label{stFF}
\cF_0(\rho,\kappa)^*\cF_0(\rho,\kappa)={\bf 1}_{\hat{ \cH} }+
 \frac{\rho\sin(\pi^2/\rho\kappa)}{2\pi} \sum_{j=1}^n\Big(
\chi^{(e,j)}\otimes \overline{ \chi^{(e,j)}}-\chi^{(o,j)}\otimes \overline{ \chi^{(o,j)}}\Big), 
\ee
with $\chi^{(e,j)}$/$\chi^{(o,j)}$ given by~\eqref{chie}/\eqref{chio}. For~$\rho\kappa$ in the interval~$I_m$~\eqref{Im}, it satisfies
\be\label{FstF}
\cF_0(\rho,\kappa)\cF_0(\rho,\kappa)^*={\bf 1}_{ \cH }+\frac{\kappa\sin(\pi^2/\rho\kappa)}{\pi}\sum_{j=0}^{m-1} (-)^j\Psi^{(j)}\otimes \Psi^{(j)},
\ee
with $\Psi^{(j)}$ defined by~\eqref{Psij} and~\eqref{hj}. Furthermore, at the interval endpoints we have
\be\label{cFmfree}
 \cF_0(\rho,\kappa)=(-)^{m-1}\cF_0,\ \ \rho\kappa =\pi/m,\ \ \  \ m=2,3,\ldots,
\ee
where $\cF_0$ is the Fourier transform.
\end{proposition}
\begin{proof}
We obtain~\eqref{stFF} from Theorem~A.1 and~\eqref{hatRfin}. Likewise, \eqref{FstF} follows from Theorem~B.1 and~\eqref{Rfin}, except at the midpoint of~$I_m$, at which~$\hat{w}_0(k)$~\eqref{wod} has a double pole. However, the pole on the rhs of~\eqref{hatres} is cancelled by a zero coming from~\eqref{Lameq}--\eqref{Lamdi}, so that $R(r,r')$ is regular for~$\rho\kappa=\pi/(m+1/2)$. By continuity, therefore, \eqref{FstF} holds on all of~$I_m$. Finally, \eqref{cFmfree} follows from~\eqref{psi0free}.
\end{proof}

For $N>0$ we need to keep track of more than one pole moving out of the $\im r$- and $\im k$-strips as $a_-$ decreases. We shall only detail the case where the `highest' $r$-pole has moved out, so as to explicitly reveal the isometry breakdown  for the interval
\be\label{isovio}
a_-\in (Na_+,(N+1/2)a_+)\Leftrightarrow \rho\kappa \in (N\pi, (N+1/2)\pi).
\ee
Thus we can still work with the pole locations~$r_j$~\eqref{rj} for $j=0,\ldots,N-1$, but now we need 
\be
r_N^{(1)}:=i\pi(N+1/2)/\kappa-i\rho\Leftrightarrow x_N^{(1)}:=i(N+1/2)a_+-ia_-.
\ee
From~\eqref{KN} and~\eqref{SigN} we deduce
\be
K_N(\de (Na_+-ia_-),y)=(-)^Ne_+(\de y)K_N(\de iNa_+,y)=(-)^Ne_+(\de y)\prod_{j=N+1}^{2N} 2s_-(ija_+),
\ee 
so via~\eqref{lamN} we obtain
\be
\mu^{\tau}_N(\de x_N^{(1)},y)=\tau v_N(y)e_+(\de\tau y)e_-(\de(y-i(N+1)a_+)/2)\prod_{j=N+1}^{2N+1} 2is_-(ija_+).
\ee
Recalling~\eqref{vhy}, this entails
\begin{multline}
\mu^{\tau}_N(\de a_- r_N^{(1)}/\rho,a_-k/\kappa)=\tau \exp(\de\tau \rho k)\exp\big(\de\big(\pi k/2\kappa-i(N+1)
\pi^2/2\rho\kappa\big)\big)
\\
\times \frac{\prod_{j=N+1}^{2N+1}\sinh\big(ij \pi^2/\rho\kappa\big)}{\prod_{j=1}^{N+1}\sinh\big(\pi k/\kappa-ij \pi^2/\rho\kappa\big)}.
\end{multline}
As should be the case, this reduces to~\eqref{mu0res} for $N=0$ and~$n=1$. Using the residue (cf.~\eqref{wr})
\be
w_N=\rho\Big(2\pi\prod_{j=1}^{2N+1}2\sinh(ij\pi^2/\rho\kappa)\Big)^{-1},
\ee
 and following the same steps as for $N=0$, this readily yields the arbitrary-$N$ and $n=1$ counterpart of~\eqref{stFF}. Specifically, we obtain
\be\label{stFFN}
\cF_N(\rho,\kappa)^*\cF_N(\rho,\kappa)={\bf 1}_{\hat{ \cH} }+(-)^N
 \frac{\rho\prod_{j=N+1}^{2N+1}\sin(j\pi^2/\rho\kappa)}{2\pi\prod_{j=1}^{N}\sin(j\pi^2/\rho\kappa)}  \Big(
\chi_N^{(e,1)}\otimes \overline{ \chi_N^{(e,1)}}-\chi_N^{(o,1)}\otimes \overline{ \chi_N^{(o,1)}}\Big), 
\ee
with even and odd functions
\bea\label{chieN}
\chi_N^{(e,1)}(k) 
&  \equiv & \frac{1}{\prod_{j=1}^N2\sinh(\pi k/\kappa+ij\pi^2/\rho\kappa)} 
\nonumber \\
& \times & \frac{\sinh(\rho k)}
{\cosh(\pi k/2\kappa+i(N+1)\pi^2/2\rho\kappa)} \left(\begin{array}{c}
1 \\ 
-1 
\end{array}\right),
\eea
\bea\label{chioN}
\chi_N^{(o,1)}(k)  
& \equiv & \frac{1}{\prod_{j=1}^N2\sinh(\pi k/\kappa+ij\pi^2/\rho\kappa)} 
\nonumber \\
& \times &  \frac{\sinh(\rho k)}
{\sinh(\pi k/2\kappa+i(N+1)\pi^2/2\rho\kappa)} \left(\begin{array}{c}
1 \\ 
1 
\end{array}\right).
\eea 

An inspection of the proof of Theorem~4.3 reveals that the reasoning can be applied to the interval~\eqref{isovio} as well. This readily yields
\be
\cF_N(\rho,\kappa)\cF_N(\rho,\kappa)^*={\bf 1}_{ \cH }+(-)^N\frac{\kappa\prod_{j=N+1}^{2N+1}\sin (j\pi^2/\rho\kappa)}{\pi\prod_{j=1}^N\sin (j\pi^2/\rho\kappa)} \Psi_N\otimes \Psi_N,
\ee
with $\Psi_N$ given by \eqref{PsiN} and~\eqref{wN}. This concludes our account of isometry violation for the interval~\eqref{isovio}. 

The reader who has followed us to this point will realize that all ingredients are in place to handle the intervals arising when  $a_-$ is further decreased, but we shall not pursue this.  


\begin{appendix}

\section{The  transform $\cF$}\label{AppA}

Our general transform $\cF$ involves two Hilbert spaces
\be\label{cH1}
\cH\equiv L^2(\R, dr),
\ee
and 
\be\label{cH2}
\hat{\cH}\equiv L^2((0,\infty), dk)\otimes \C^2,
\ee
with inner products
\be
(d,e)_1\equiv \intR dr\,\overline{d(r)}e(r),\ \ \ d,e\in\cH,
\ee
\be
(f,g)_2\equiv \sum_{\de=+,-} \intp dk\, \overline{f_{\de}(k)}g_{\de}(k),\ \ \ f,g\in\hat{\cH}.
\ee
(Physically speaking, these spaces can be thought of as the reduced position and momentum space of a particle pair; the variables~$r$ and~$k$ are dimensionless and related to center-of-mass position~$x=x_1-x_2$ and momentum~$p=p_1-p_2$ by $r=\nu x$ and $k=p/2\hbar\nu $.)  The transform is at first defined on the dense $\hat{\cH}$-subspace~$\hat{\cC}$ of $\C^2$-valued smooth functions $f=(f_+,f_-)$ with compact support in~$(0,\infty)$, which are assumed to be mapped into~$\cH$:
\be\label{cF}
\cF\, :\, \hat{\cC}\equiv C_0^\infty((0,\infty))^2\subset \hat{\cH} \to \cH.
\ee
Specifically, we have
\be\label{cFdef}
(\cF f)(r)=\frac{1}{\sqrt{2\pi}}\intp dk \sum_{\de=+,-} F_{\de}(r,k)f_{\de}(k).
\ee
The two transform kernels are defined in terms of one function $\Psi(r,k)$:
\be\label{Fp}
F_+(r,k)=\Psi(r,k),
\ee 
\be\label{Fm}
F_-(r,k)=-\Psi(-r,k).
\ee

For the case of no interaction we have
\be\label{free}
\Psi(r,k)=\exp(irk),
\ee
and then we denote the corresponding transform by~$\cF_0$. Thus, $\cF_0$ amounts to the Fourier transform, with~$\hat{f}\in L^2(\R,dk)$ corresponding to~$f\in\hat{\cH}$ via
\be\label{ident}
\hat{f}(k)=
\left\{
\begin{array}{ll}
f_+(k),  &  k>0, \\
-f_-(-k),  &  k<0.
\end{array}
\right.
\ee
(The phase of $F_-$ is a matter of convention, the minus sign being chosen for later purposes.)

We proceed with an initial list of assumptions on the function $\Psi(r,k)$. (This list will be supplemented in Appendix~B.) To begin with, we assume $\Psi(r,k)$ is a smooth function on $\R^2$ that satisfies
\be\label{Psiconj}
\overline{\Psi(r,k)}=\Psi(r,-k),\ \ \ (r,k)\in\R^2,
\ee
and is of the form
\be\label{Psiform}
\Psi(r,k)=  w(r)^{1/2}\sumt m^{\tau}(r,k)\exp(i\tau rk).
\ee

The weight function $w(r)$ is a positive even function and the positive square root is taken in~\eqref{Psiform}. It extends to a meromorphic function that has period $i\rho$ with~$\rho>0$. Its only singularities in the period strip $\im r\in (0,\rho)$ are finitely many simple poles. If $r_0$ is one of these poles, it follows from evenness and $i\rho$-periodicity that $i\rho -r_0$ is another such pole whose   residue has opposite sign. Thus we can pair off the poles, obtaining $2L$ distinct poles $r_1,\ldots,r_{2L}$ related by
\be\label{poles}
r_{j+L}=i\rho -r_j,\ \ \ j=1,\ldots,L,
\ee
  with residues
\be\label{wres}
w_j\equiv {\rm Res}\,( w(r))|_{r=r_j},\ \ \ j=1,\ldots,2L,
\ee
satisfying
\be\label{resid}
w_{j+L}=-w_j,\ \ \ j=1,\ldots,L.
\ee

The two coefficients $m^{\pm}(r,k)$ are smooth and satisfy
\be\label{mconj}
\overline{m^{\tau}(r,k)}=m^{\tau}(r,-k),\ \ \ \tau=+,-,\ \ \  (r,k)\in\R^2,
\ee
in accord with~\eqref{Psiconj}. As functions of~$r$, they extend to entire functions that are either both $i\rho$-periodic or both $i\rho$-antiperiodic.

Next, we assume that the asymptotic behavior of $\Psi(r,k)$ for $\re r\to \pm\infty$  is given by
\be\label{mpas}
w(r)^{1/2}m^+(r,k)= 
\left\{
\begin{array}{ll}
T(k)+\Ogamp,  &  \limp, \\
 1+\Ogamn,  &  \limn,
\end{array}
\right.
\ee
\be\label{mnas}
w(r)^{1/2}m^-(r,k)= 
\left\{
\begin{array}{ll}
\Ogamp,  &  \limp, \\
- R(k)+\Ogamn,  &  \limn,
\end{array}
\right.
\ee
and that we also have
\be\label{asextra}
\partial_r(w(r)^{1/2}m^\tau (r,k))=O(\exp(\mp\gamma r)),\ \ \tau=+,-,\ \ \re r\to \pm \infty.
\ee
Here, we have $\gamma>0$ and the implied constants can be chosen uniformly for $\im r$ and~$k$ in $\R$-compacts. (To leave no doubt regarding the meaning of the uniformity assumption, let us spell it out for the bound~\eqref{asextra}: For given compact subsets $K_1,K_2$ of~$\R$, it says that there exist positive constants~$C$ and~$R$ such that the modulus of the lhs is bounded above by $C\exp(-\gamma |\re r|)$ for all $k\in K_1$ and all~$r\in\C$ with $\im r\in K_2$ and $|\re r| \ge R$.)

The transmission and reflection coefficients $T(k)$ and $R(k)$ are assumed to be smooth functions on $\R$ satisfying
\be\label{trconj}
\overline{T(k)}=T(-k),\ \ \ \overline{R(k)}=R(-k),\ \ \ k\in\R,
\ee
\be\label{unit1}
T(-k)T(k)+R(-k)R(k)=1,
\ee
\be\label{unit2}
T(-k)R(k)+R(-k)T(k)=0.
\ee
As a consequence, the matrix multiplication operator on $\hat{\cH}$ (`$S$-matrix') given by
\be\label{Sm}
S(k)\equiv
\left(\begin{array}{cc}
T(k) & R(k) \\
R(k)  &  T(k) 
\end{array}\right),\ \ \ k>0,
\ee
is a unitary operator. We also point out that our initial assumption that $\cF$ maps $\hat{\cC}$ into $\cH$ readily follows from the asymptotic behavior we have just assumed.

We need two more assumptions on the coefficients. To this end we introduce
\be\label{Mtsdef}
M_{\alpha}^{\tau}(r,k,k')\equiv m^+(-r,-k)m^{\tau}(-\alpha r,k')+m^-(r,-k)m^{-\tau}(\alpha r,k'),\ \ \ \tau,\alpha=+,-.
\ee
Then we assume
\be\label{Mass}
M_{\alpha}^{\alpha}(r,k,k)=M_{\alpha}^{\alpha}(-r,k,k),\ \ \alpha=+,-,\ \  \ (r,k)\in\R^2.
\ee
By contrast to previous ones, these   assumptions may seem unintuitive. For now, we point out that
  \eqref{mpas} and~\eqref{mnas} imply
\be\label{M1as}
w(r)M_{+}^{+}(r,k,k')
= 
\left\{
\begin{array}{ll}
1+\Ogamp,  &  \limp, \\
 T(-k)T(k')+R(-k)R(k')+\Ogamn,  &  \limn,
\end{array}
\right.
\ee
\be\label{M2as}
 w(r)M_{-}^{-}(r,k,k')
=  
\left\{
\begin{array}{ll}
\Ogamp,  &  \limp, \\
 -T(-k)R(k')-R(-k)T(k')+\Ogamn,  &  \limn.
\end{array}
\right.
\ee
Consequently, the evenness assumptions~\eqref{Mass} may be viewed as generalizations of the unitarity assumptions~\eqref{unit1} and~\eqref{unit2}. 

Before stating a theorem that only involves the above assumptions on the various coefficients determining $\Psi(r,k)$, we add two simple choices that satisfy these assumptions. Besides $\rho$, they involve two positive parameters~$\kappa$ and~$\varphi$.   The parametrization of the $k$-dependence and the choice of numerical factors   anticipate the further assumptions made in Appendix~B.

Specifically, we set
\be\label{wsp}
w(r)\equiv 1/4\sinh(\pi r/\rho+i\varphi)\sinh(\pi r/\rho-i\varphi),
\ee
\be\label{mnu}
m_{\sigma}^{\tau}(r,k)\equiv \frac{\ell_{\sigma}^{\tau}(r,k)}{2i\sinh (\pi k/\kappa-2i\varphi)},\ \ \ \tau,\sigma=+,-, 
\ee
\be\label{ellnum}
\ell_{\sigma}^-(r,k)\equiv 2i\sigma\sinh(2i\varphi)\exp(-\pi r/\rho),\ \ \ \sigma=+,-,
\ee
\be\label{ellnup}
\ell_{\sigma}^+(r,k)\equiv 2i\big(\exp(-\pi r/\rho)\sinh(\pi k/\kappa-2i\varphi) -
\exp(\pi r/\rho) \sinh(\pi k/\kappa)\big),\ \ \sigma=+,-,
\ee
and denote the associated special transforms by $\cF_{\sigma}(\varphi)$, $\sigma=+,-$.  Note first that all of these functions are $i\pi$-periodic in $\varphi$, so we may and will restrict $\varphi$ to the period interval $(0,\pi]$. For the special choices $\varphi=\pi$ and $\varphi=\pi/2$, the weight function $w(r)$ does not satisfy the simple pole restriction, but one readily verifies 
\be\label{phiexc}
 w(r)^{1/2}\sumt m_{\pm}^{\tau}(r,k)\exp(i\tau rk)=\left\{
 \begin{array}{ll}
-\exp(irk), & \varphi =\pi, \\
\exp(irk), & \varphi =\pi/2.
 \end{array}\right.
\ee
Thus we have $\cF_{\pm}(\pi)=-\cF_0$ and $\cF_{\pm}(\pi/2)=\cF_0$, cf.~\eqref{free}. 
Choosing next 
\be\label{phiint}
\varphi \in (0,\pi/2)\cup (\pi/2,\pi),
\ee
it is routine to check that all assumptions are satisfied, with $L=1$, $\gamma=2\pi/\rho$, and 
\be\label{tnu}
T_{\pm}(k)=\frac{\sinh(\pi k/\kappa)}{\sinh(2i\varphi-\pi k/\kappa )},
\ee
\be\label{rnu}
R_{\pm}(k)=\frac{\pm\sinh(2i\varphi )}{\sinh(2i\varphi-\pi k/\kappa )}.
\ee

We are now prepared  for the following theorem.

\begin{theorem}
Letting $f,g\in\hat{\cC}$, we have
\begin{multline}\label{cFres}
(\cF f,\cF g)_1= (f,g)_2+  i\sum_{j=1}^{L}w_j
\sum_{\de,\de'=+,-}\de\de'\intp dk\, \overline{f_{\de}(k)}\intp dk'\, g_{\de'}(k')
\\
\times   \sum_{\nu,\nu'=+,-}\frac{\exp(i r_j(\nu k-\nu' k'))}{1-\exp(\rho (\nu' k'-\nu k))} M_{\de \de'}^{\de\de'\nu\nu'}(\nu r_j,k,k').
\end{multline}
\end{theorem}
\begin{proof}
To start with, we stress that the~$4L$ residue integrals on the rhs are absolutely convergent due to our evenness assumptions~\eqref{Mass}. Indeed, fixing $j,\de,\de'$, the denominator has zeros on the integration region only when~$k=k'$ and $\nu =\nu'$. Now the two terms in the sum with $\nu'=\nu$ involve $M_{+}^{+}(\nu r_j,k,k')$ for the two cases $\de=\de'$,  and $M_{-}^{-}(\nu r_j,k,k')$ for $\de=-\de'$.  Hence \eqref{Mass} ensures the  cancellation of the poles arising for $k=k'$. 

Proceeding with the proof, we use Fubini's theorem to write
\be\label{Lamlim}
(\cF f,\cF g)_1= \lim_{\Lambda \to\infty}\sum_{\de,\de'=+,-}\de\de'\intp dk\, \overline{f_{\de}(k)}\intp dk'\, g_{\de'}(k')I_{\de \de'}(\Lambda ,k,k'),
\ee
with
\be
I_{\sigma}(\Lambda ,k,k')\equiv  \frac{1}{2\pi}\int_{-\Lambda }^\Lambda  dr\, \overline{\Psi( r,k)}\Psi(\sigma r,k'),\ \ \ \sigma=+,-,
\ee
cf.~\eqref{cFdef}--\eqref{Fm}. Using~\eqref{Psiconj} and~\eqref{Psiform}, we obtain
\be\label{Idd}
I_{\sigma}(\Lambda ,k,k')=\frac{1}{2\pi}\sumtt \int_{-\Lambda }^\Lambda  dr\, w(r)J^{\tau,\tau'}_{\sigma}(r,k,k'),
\ee
where
\be\label{Jdef}
J^{\tau,\tau'}_{\sigma}(r,k,k')\equiv  m^{\tau} ( r,-k)m^{\tau'} (\sigma r,k')\exp(ir(\sigma\tau'k'-\tau k)).
\ee
Fixing $\tau$ and $\tau'$, the integrand in~\eqref{Idd} picks up an $r$-independent multiplier when~$r$ is shifted by $i\rho$. Indeed, $w(r)$ is $i\rho$-periodic, while
\be
J^{\tau,\tau'}_{\sigma}(r+i\rho,k,k')=\exp(\rho(\tau k-\sigma\tau' k'))J^{\tau,\tau'}_{\sigma}(r,k,k').
\ee
To exploit this, we use Cauchy's theorem to get
\begin{multline}\label{Jint}
 \big(1-\exp(\rho(\tau k-\sigma\tau' k'))\big) \int_{-\Lambda }^\Lambda  dr\, w(r)J^{\tau,\tau'}_{\sigma}(r,k,k') 
 \\
=2\pi i\sum_{j=1}^{2L} w_j J^{\tau,\tau'}_{\sigma}(r_j,k,k')+B^{\tau,\tau'}_{\sigma}(\Lambda ,k,k'),
\end{multline}
where
\be\label{BttLam}
B^{\tau,\tau'}_{\sigma}(\Lambda ,k,k')\equiv
-\left( \int_\Lambda ^{\Lambda +i\rho}+\int_{-\Lambda +i\rho}^{-\Lambda } \right) dr\, w(r)J^{\tau,\tau'}_{\sigma}(r,k,k').
\ee
Here, we have chosen $\Lambda $ sufficiently large so that the rectangular contour with corners~$-\Lambda,\Lambda,\Lambda+i\rho,-\Lambda+i\rho$, encloses all of the $w$-poles $r_1,\ldots,r_{2L}$.

Using \eqref{poles}--\eqref{resid}, we can write the residue sum as
\begin{multline}\label{ressum}
\sum_{j=1}^{2L} w_j J^{\tau,\tau'}_{\sigma}(r_j,k,k')
 =   \sum_{j=1}^{L} w_j \Big(
 \exp(i r_j(\sigma\tau' k'-\tau k))
  m^{\tau}( r_j,-k)m^{\tau'}(\sigma r_j,k')
  \\
-\exp(\rho(\tau k-\sigma\tau' k'))\exp(-i r_j(\sigma\tau' k'-\tau k))
  m^{\tau}(- r_j,-k)m^{\tau'}(-\sigma r_j,k')\Big).
\end{multline}

As a consequence, we obtain
\begin{multline}
\sum_{\tau,\tau'=+,-}\sum_{j=1}^{2L}w_j J^{\tau,\tau'}_{\sigma}(r_j,k,k')\big/[1-\exp(\rho(\tau k-\sigma\tau' k'))]
\\
=\sum_{j=1}^L\sum_{\nu,\nu'=+,-}\frac{w_j}{1-\exp(\rho(\nu k-\nu' k'))}\Big[ \exp(ir_j(\nu' k'-\nu k))m^{\nu}( r_j,-k)m^{\sigma\nu'}(\sigma r_j,k')
\\
-\exp(\rho(\nu k-\nu' k'))\exp(-ir_j(\nu' k'-\nu k))m^{\nu}(- r_j,-k)m^{\sigma\nu'}(-\sigma r_j,k')\Big],
\end{multline}
where we have changed variables $\tau,\tau'\to \nu,\sigma\nu'$. When we now take $\nu,\nu'\to -\nu,-\nu'$ in the first sum, we get
\begin{multline}\label{crit1}
\sum_{j=1}^Lw_j\sum_{\nu,\nu'=+,-}\frac{\exp(ir_j(\nu k-\nu' k'))}{1-\exp(\rho(\nu' k'-\nu k))}
\\
\times\big[  m^{-\nu}( r_j,-k)m^{-\sigma\nu'}(\sigma r_j,k')+m^{\nu}(- r_j,-k)m^{\sigma\nu'}(-\sigma r_j,k')\big].
\end{multline}
Combining this with~\eqref{Lamlim}--\eqref{Jint}, the $4L$ residue integrals on the rhs of~\eqref{cFres}  result upon rewriting the second line of~\eqref{crit1} as $M_{\sigma}^{\sigma\nu\nu'}(\nu r_j,k,k')$, cf.~\eqref{Mtsdef}. (The latter equality is not immediate, but it can be readily verified case by case.) 

We proceed to rewrite~$B^{\tau,\tau'}_{\sigma}(\Lambda ,k,k')$. Taking $r\to -r+i\rho$ in the second integral, we use evenness and $i\rho$-periodicity of~$w(r)$ to obtain
\be
B^{\tau,\tau'}_{\sigma}(\Lambda ,k,k')=\int_\Lambda ^{\Lambda +i\rho}  dr\, w(r)\big( J^{\tau,\tau'}_{\sigma}(-r+i\rho,k,k') 
-J^{\tau,\tau'}_{\sigma}(r,k,k') \big).
\ee
Shifting $r$ over $i\rho/2$ and using~\eqref{Jdef}, this entails
\begin{multline}\label{boundint}
 \frac{B^{\tau,\tau'}_{\sigma}(\Lambda ,k,k')}{1-\exp(\rho(\tau k-\sigma\tau' k'))}=\frac{1}{2\sinh(\rho(\tau k-\sigma\tau' k')/2)}\int_{\Lambda -i\rho/2}^{\Lambda +i\rho/2}  dr\, w(r+i\rho/2)
\\
\times\sum_{\epsilon=+,-}\epsilon \exp(i\epsilon r(\sigma\tau' k'-\tau k))m^{\tau}(\epsilon(r+i\rho/2),-k)m^{\tau'}(\sigma\epsilon(r+i\rho/2),k').
\end{multline}
We now study the $\Lambda \to\infty$ limit of
\be\label{start}
\de\de'\intp dk\, \overline{f_{\de}(k)}\intp dk'\, g_{\de'}(k')\sumtt \frac{B^{\tau,\tau'}_{\de\de'}(\Lambda ,k,k')}{1-\exp(\rho(\tau k-\de\de'\tau' k'))},
\ee
 for the four cases~$\de,\de'=+,-$. 

First let $\de=\de'$. Then for the terms in~\eqref{start} with $\tau=-\tau'$, the denominator does not vanish on the integration region. Therefore, using~\eqref{boundint} and the asymptotics assumptions~\eqref{mpas}--\eqref{mnas}, it readily follows that they vanish for $\Lambda \to\infty$. (For the dominant asymptotics the $r$-integration is elementary, and one need only invoke the Riemann-Lebesgue lemma and dominated convergence, while the subdominant terms vanish by dominated convergence.) 

Consider next the terms with $\tau=\tau'$.  From~\eqref{boundint} and~\eqref{Mtsdef} we infer
\begin{multline}\label{firstcase}
\sum_{\tau=+,-}\frac{B^{\tau,\tau}_{+}(\Lambda ,k,k')}{1-\exp(\rho\tau (k-k'))}
=\frac{1}{2\sinh(\rho( k'- k)/2)}\int_{\Lambda -i\rho/2}^{\Lambda +i\rho/2}  dr\, w(r+i\rho/2)
\\
\times\sum_{\alpha=+,-}\alpha\exp(i\alpha r(k-k'))M_{+}^{+}(\alpha(r+i\rho/2),k,k').
\end{multline}
Now we change variables
\be\label{stt}
t=\rho k/2,\ \ t'=\rho k'/2,\ \ s=2r/\rho,
\ee
and set
\be
G_{\alpha}(s,t,t')\equiv w(\rho(s+i)/2))M_{+}^{+}(\alpha \rho (s+i)/2, 2t/\rho,2t'/\rho),\ \ \ \alpha=+,-.
\ee
This implies that the rhs of \eqref{firstcase} can be rewritten as
\be
\frac{\rho}{4\sinh(t'-t)}\int_{2\Lambda /\rho -i}^{2\Lambda /\rho +i}ds\, \sum_{\alpha=+,-}\alpha\exp(i\alpha s(t-t'))G_{\alpha}(s,t,t').
\ee
In view of our evenness assumptions~\eqref{Mass}, the functions $G_{\pm}(s,t,t')$ satisfy~\eqref{Geq}. Moreover, \eqref{M1as} implies  that $G_{\pm}(s,t,t')$ also satisfy~\eqref{Gpmas}--\eqref{rhojpart}, with~$\eta=\gamma\rho/2$, and
\be
A_1(t,t')=1,\ \ \ \ \ A_2(t,t')=T(-2t/\rho)T(2t'/\rho)+R(-2t/\rho)R(2t'/\rho).
\ee
Finally, our unitarity assumption \eqref{unit1} entails that \eqref{Aeq} is obeyed, with~$A(t)=1$.
 Putting
\be
\phi(t,t')\equiv \frac{1}{2\pi\rho}\,\overline{f_{\de}(2t/\rho)}g_{\de}(2t'/\rho),
\ee
it follows that all assumptions of Lemma~C.1 are fulfilled.
 As a result, we deduce
\be\label{deplim}
\lim_{\Lambda \to\infty}\int_{(0,\infty)^2} dk\, dk'\,\frac{\overline{f_{\de}(k)}  g_{\de}(k')}{2\pi}\sumtt \frac{B^{\tau,\tau'}_{+}(\Lambda ,k,k')}{1-\exp(\rho(\tau k-\tau'k'))}=\intp dk\, \overline{f_{\de}(k)}g_{\de}(k).
\ee

 It remains to show that for the two remaining cases $\de'=-\de$ the pertinent limits vanish. 
As before, this is easily checked when $\tau'=\tau$. 
Consider now the case  $\tau'=-\tau$. As the analog of~\eqref{firstcase}, we then get
\begin{multline}\label{thirdcase}
\sum_{\tau=+,-}\frac{B^{\tau,-\tau}_{-}(\Lambda ,k,k')}{1-\exp(\rho\tau (k-k'))}
=\frac{1}{2\sinh(\rho( k'- k)/2)}\int_{\Lambda -i\rho/2}^{\Lambda +i\rho/2}  dr\, w(r+i\rho/2)
\\
\times\sum_{\alpha=+,-}\alpha\exp(i\alpha r(k-k'))M_{-}^{-}(\alpha(r+i\rho/2),k,k').
\end{multline}
Using again the variable change~\eqref{stt}, but now putting
\be
G_{\alpha}(s,t,t')\equiv w(\rho(s+i)/2))M_{-}^{-}(\alpha \rho (s+i)/2, 2t/\rho,2t'/\rho),\ \ \ \alpha=+,-,
\ee
\be
\phi(t,t')\equiv \frac{1}{2\pi\rho}\,\overline{f_{\de}(2t/\rho)}g_{-\de}(2t'/\rho),
\ee
it follows once more that all assumptions of Lemma~C.1 are fulfilled, now with
\be
A_1(t,t')=0,\ \ \ A_2(t,t')=-T(-2t/\rho)R (2t'/\rho)-R(-2t/\rho)T(2t'/\rho),
\ee
and $A(t)=0$ (cf.~\eqref{Mass}, \eqref{M2as}, and \eqref{unit2}).
Thus the cases $\de'=-\de$ yield limit zero, so that  the proof is complete. 
\end{proof}

To apply the theorem in the main text, it is expedient to rewrite the residue sums in an  alternative form, which is detailed in the following corollary.

\begin{corollary}
Setting
\begin{multline}\label{Rdk}
\hat{R}_{\de,\de'}(k,k')\equiv i\de\de'\sum_{j=1}^{L} w_j 
  \sum_{\nu,\nu'=+,-}\frac{ 1}{1-\exp(\rho (\nu' k'-\nu k))} 
  \\
  \times
\big[  \mu^{-\de\nu}(\de r_j,-k)\mu^{-\de'\nu'}(\de' r_j,k')+\mu^{\de\nu}(-\de r_j,-k)\mu^{\de'\nu'}(-\de' r_j,k')\big],  
\end{multline}
where
\be\label{mudef}
\mu^{\tau}(r,k)\equiv \exp(i\tau rk)m^{\tau}(r,k),\ \ \ \tau=+,-,
\ee
the transform $\cF$ is isometric if and only if the residue sums $\hat{R}_{\de,\de'}(k,k')$, $\de,\de'=+,-$, vanish.
\end{corollary}
\begin{proof}
	This readily follows  from~\eqref{cFres} by invoking~\eqref{crit1} with $\sigma=\de\de'$.
 \end{proof}
 
For the special cases~\eqref{wsp}--\eqref{ellnup}, we have
\be
r_1=i\rho\varphi/\pi\in i(0,\rho),\ \ \ w_1=\frac{\rho}{4\pi\sinh(2i\varphi)},\ \ \ \varphi \in (0,\pi/2)\cup (\pi/2,\pi),
\ee
\be
\ell_{\sigma}^-(\pm r_1,k)=2i\sigma\sinh(2i\varphi)\exp(\mp i\varphi),
\ee
\be
\ell_{\sigma}^+(\pm r_1,k)=-2i\sinh(2i\varphi)\exp(\pm (\pi k/\kappa-i\varphi)).
\ee
From this we obtain the following ratios:
\begin{multline}\label{ratios}
M_{\alpha}^{\alpha}(-r_1,k,k')/M_{\alpha}^{\alpha}(r_1,k,k')=e^{\pi(k'-k)/\kappa},\ \  
M_{\alpha}^{-\alpha}(r_1,k,k')/M_{\alpha}^{\alpha}(r_1,k,k')=-\sigma e^{\pi k'/\kappa}, 
\\ 
M_{\alpha}^{-\alpha}(-r_1,k,k')/M_{\alpha}^{\alpha}(r_1,k,k')=-\sigma e^{-\pi k/\kappa},\ \ \ \alpha=+,-.
\end{multline}
Hence the residue sum in~\eqref{cFres} is proportional to
\begin{multline}
\frac{1}{1-\exp(\rho(k'-k))}+\frac{\exp((2ir_1+\pi/\kappa)(k'-k))}{1-\exp(\rho(-k'+k))}
\\
-\sigma\frac{\exp((2ir_1+\pi/\kappa)k')}{1-\exp(\rho(-k'-k))}
-\sigma\frac{\exp(-(2ir_1+\pi/\kappa)k)}{1-\exp(\rho(k'+k))}.
\end{multline}
For this to vanish we can choose either 
\be\label{phi1}
\sigma=+,\ \ \ 2ir_1+\pi/\kappa=0,\ \ \ \rho\kappa>\pi/2, 
\ee
or
\be\label{phi2}
\sigma=+,-,\ \ \ 2ir_1+\pi/\kappa=-\rho,\ \ \ \rho\kappa>\pi,
\ee
with the inequalities due to the requirement that~$r_1$ belong to $i(0,\rho)$. 
Indeed, for these choices the asserted vanishing comes down to the identities
\be\label{id1}
\frac{1}{1-A'/A}+\frac{1}{1-A/A'}-\frac{1}{1-1/A'A}-\frac{1}{1-A'A}=0,
\ee
\be\label{id2}
\frac{1}{1-A'/A}+\frac{A}{A'}\,\frac{1}{1-A/A'} -\frac{1}{A'}\,\frac{\sigma}{1-1/A'A} -A\,\frac{\sigma}{1-A'A}=0,\ \ \sigma=+,-,
\ee
which are easily checked.
Thus, setting
\be\label{rk0}
\phi_0\equiv\frac{\pi^2}{2\rho\kappa},\ \ \ \rho\kappa>\pi/2,
\ee
\be\label{rke}
\phi_e\equiv\frac{\pi^2}{2\rho\kappa}+\frac{\pi}{2},\ \ \ \rho\kappa>\pi,
\ee
we deduce that the three transforms $\cF_+(\phi_0)$, $\cF_+(\phi_e)$ and $\cF_-(\phi_e)$ are isometric.
In case the restrictions on $\rho\kappa$ are not satisfied, we need to subtract a suitable multiple of $i\rho$ from $i\rho\phi_j/\pi$, $j=0,e$, to obtain a pole location $r_1$ with $\im r\in (0,\rho)$. Generically, this yields a nonvanishing residue sum, hence isometry  breakdown.

The transform $\cF_+(\phi_0)$ amounts to the $N=0$ transform in the main text. The `extra' transforms~$\cF_{\pm}(\phi_e)$ go to show that our assumptions allow realizations beyond the main text. In Section~5 we shall elaborate on the issue of  isometry violation. In particular, Proposition~5.1 encodes the salient features of~$\cF_+(\phi_0)$ for $\rho\kappa\le\pi/2$.


\section{The  transform $\cF^*$ }\label{AppB}

In this appendix we retain the assumptions on $\cF$ made in Appendix~A. As we have pointed out, they imply in particular that $\cF$ maps~$\hat{\cC}$ into~$\cH$. However, they do not imply that~$\cF$ is a bounded operator. (To be sure, boundedness is plain when the residue sum vanishes.) In fact, it might not even follow from the assumptions made thus far that~$\cF$ has an adjoint that is densely defined.

From the additional assumptions in this appendix it shall follow that~$\cF^*$ is indeed densely defined. In particular, we shall see that we have
\be\label{cTst}
\cF^*\, :\, \cC\equiv C_0^\infty(\R)\subset \cH \to \hat{\cH}.
\ee 
Of course, it follows without further assumptions that for~$h\in \cC$ we obtain
\be
(h,\cF f)_1=\frac{1}{\sqrt{2\pi}}\sum_{\de=+,-}\intR dr\, \overline{h(r)} 
\intp dk \,  F_{\de}(r,k)f_{\de}(k),\ \ \ f\in \hat{\cC},
\ee
cf.~\eqref{cFdef}. From this and~\eqref{Fp}--\eqref{Psiconj} we easily deduce
\be\label{cTstdef}
(\cF^*h)_{\de}(k)=\frac{\de}{\sqrt{2\pi}} \intR dr\, \Psi(\de r,-k)h(r),\ \ \ \de=+,-,\ \ \ k>0,
\ee
whenever both integrals yield functions of~$k$ that are square-integrable on~$(0,\infty)$. But we can only ensure this property by making more assumptions. 

These extra assumptions concern the $k$-dependence of $m^{\pm}(r,k)$. They will allow us to study  $\cF^*$ in much the same way as $\cF$ itself. We would like to stress, however, that it is at this point that we part company with the nonrelativistic framework, inasmuch as the Jost functions~\eqref{Jost} do not have the periodicity in the spectral variable~$k$ we are about to require. 

First, we assume that the coefficients  $m^{+}(r,k)$ and $m^{-}(r,k)$ are $i\kappa$-periodic and $i\kappa$-antiperiodic functions of~$k$, resp., with $\kappa>0$. Second, they are of the form
\be\label{mtform}
m^{\tau}(r,k)= v(k)\ell^{\tau}(r,k),\ \ \ \tau=+,-.
\ee
Here, $\ell^{\pm}(r,k)$ are assumed to be entire in~$k$, whereas $v(k)$ is meromorphic and either $i\kappa$-periodic or $i\kappa$-antiperiodic. Third, $v(k)$ satisfies
\be\label{vass}
\overline{v(k)}=v(-k),\ \ \ k\in\R.
\ee
 Note that these assumptions are satisfied for the special cases given by~\eqref{mnu}--\eqref{ellnup}, with $v(k)$ being $i\kappa$-antiperiodic.

Introducing the dual weight function
\be\label{hwdef}
\hat{w}(k)\equiv   |v(k)|^2=v(-k)v(k),\ \ \ k\in\R,
\ee
it follows from the assumptions just made that it  extends to a meromorphic $i\kappa$-periodic function. We assume that its  
 only singularities in the period strip $\im k\in (0,\kappa)$ are finitely many simple poles. Just as for $w(r)$, it then follows that we can pair off the poles, obtaining $2\hat{L}$ distinct poles $k_1,\ldots,k_{2\hat{L}}$ related by
\be\label{whpoles}
k_{j+\hat{L}}=i\kappa -k_j,\ \ \ j=1,\ldots,\hat{L},
\ee
  with residues
\be\label{whres}
\hat{w}_j\equiv {\rm Res}\,( \hat{w}(k))|_{k=k_j},\ \ \ j=1,\ldots,2\hat{L},
\ee
satisfying
\be\label{whresid}
\hat{w}_{j+\hat{L}}=-\hat{w}_j,\ \ \ j=1,\ldots,\hat{L}.
\ee 

Clearly, these assumptions are satisfied for~\eqref{mnu}--\eqref{ellnup}, with~$\hat{L}=1$, and
\be\label{hw}
\hat{w}(k)\equiv \frac{1}{4\sinh (\pi k/\kappa -2i\varphi)\sinh (\pi k/\kappa +2i\varphi)},
\ee
provided $\varphi$ is restricted by~\eqref{phiint} and not equal to $ \pi/4$ or $3\pi/4$. (The latter cases yield double poles at $k= i\kappa/2$.)

Next, we assume that the $|\re k|\to \infty$ asymptotics of the coefficients is given by
\be\label{mpask}
 m^+(r,k)= 
\left\{
\begin{array}{ll}
 C(r)+O(\exp(-\hat{\gamma} k)),  &  \re k\to\infty, \\
 \overline{C(r)}+O(\exp(\hat{\gamma} k)),  &  \re k\to -\infty,
\end{array}
\right.
\ee
\be\label{mmask}
 m^-(r,k)= O(\exp(\mp \hat{\gamma} k)),  \ \ \   \re k\to\pm \infty,
 \ee
\be\label{mexk}
\partial_km^{\tau}(r,k)=O(\exp(\mp \hat{\gamma} k)),  \ \ \ \tau=+,-,\ \ \   \re k\to\pm \infty,
\ee
 with $\hat{\gamma}>0$ and the implied constants uniform for $\im k$ and~$r$ in $\R$-compacts. Moreover, we assume that the function  $C(r)$ extends from the real line to an entire $i\rho$-periodic or $i\rho$-antiperiodic function satisfying
\be\label{cw}
|C(r)|^2=1/w(r),\ \ \ r\in\R.
\ee
Introducing
\be\label{Ur}
U(r)^{1/2}\equiv C(r)w(r)^{1/2},\ \ \ r\in\R,
\ee
it follows from the above assumptions that $U(r)$ extends to a meromorphic $i\rho$-periodic function that is a phase for real~$r$, and that we have a reflectionless asymptotics
\be\label{Psikas}
\Psi(r,k)= 
  \left\{
\begin{array}{ll}
 U(r)^{1/2}\exp(irk)+O(\exp(-\hat{\gamma}k)),  &  \re k\to\infty, \\
U(r)^{-1/2}\exp(irk) +O(\exp(\hat{\gamma} k)),  &  \re k\to -\infty,
\end{array}
\right.
\ee
uniformly for $\im k$ and $r$ in $\R$-compacts.

From this asymptotic behavior it easily follows that $\cF^*$ is defined on~$\cC$ and given by~\eqref{cTstdef}, as announced. Actually, it may well follow from our assumptions that $\cF$ must be bounded. We have not tried to show this, however, since we do not need this property for the general analysis undertaken in this appendix. (It is not hard to see that the example transforms $\cF_{\pm}(\varphi)$ and the transforms arising in the main text are bounded, cf. the paragraph containing~\eqref{cFform}.)  

Our final assumptions generalize the evenness assumptions~\eqref{Mass}. Introducing
\be\label{Lttp}
L^{\tau,\tau'}(k,r,r')\equiv \ell^{\tau}(r,k)\ell^{\tau'}(r',-k)+\ell^{-\tau}(-r,k)\ell^{-\tau'}(-r',-k),\ \ \ \tau,\tau'=+,-,
\ee
it follows from the above that $L^{\tau,\tau}(k,r,r')$ is $i\kappa$-periodic in~$k$, whereas $L^{\tau,-\tau}(k,r,r')$ is $i\kappa$-antiperiodic. Also,
it is plain that we have  symmetries
\be\label{keven}
L^{\tau,\tau'}(k,r,r')=L^{-\tau,-\tau'}(k,-r,-r'),\ \ \  L^{\tau,\tau'}(k,r,r)=L^{\tau',\tau}(-k,r,r).
\ee
But the evenness properties
\be\label{Lass}
L^{\tau,\tau'}(k,r,r')=L^{\tau,\tau'}(-k,-r,-r'),\ \ \ \tau,\tau'=+,-,
\ee
amount to further assumptions on the coefficients, strengthening~\eqref{Mass}. Indeed, for $\tau'=\tau=+$ and $r'=r$,   \eqref{Lass} amounts to~\eqref{Mass} with $\alpha=+$, while for $\tau'=-\tau=-$ and $r'=-r$,   \eqref{Lass} reduces to~\eqref{Mass} with $\alpha=-$.

It is easy to check that the assumptions on the $|\re k|\to\infty$ asymptotics are satisfied for the example cases \eqref{mnu}--\eqref{ellnup}, with~$\hat{\gamma}=\pi/\kappa$, and
\be
C(r)= -2\exp (i\varphi)\sinh (\pi r/\rho +i\varphi),\ \ U(r)= \exp (2i\varphi)\frac{\sinh (\pi r/\rho +i\varphi)}{\sinh (\pi r/\rho -i\varphi)}.
\ee
With due effort, the assumptions~\eqref{Lass} can also be checked for these special cases.

\begin{theorem}
With the above assumptions in effect, let $f,g\in \cC$. Then we have
\begin{multline}\label{cFstres}
(\cF^* f,\cF^* g)_2= (f,g)_1+i\sum_{j=1}^{\hat{L}} \hat{w}_j 
  \intR dr\, \overline{f(r)}\intR dr'\, g(r')w(r)^{1/2}w(r')^{1/2}
\\
\times   \sum_{\tau,\tau'=+,-}\frac{\exp(i k_j(\tau r- \tau'r'))}{1-\tau\tau'\exp(\kappa (\tau' r' - \tau r))}L^{\tau,\tau'}(  k_j,r,r').
\end{multline}
\end{theorem}
\begin{proof}
The proof proceeds along the same lines as the proof of Theorem~A.1. We begin by noting that we have an identity
\be\label{idextra}
L^{\tau,\tau}(k,r,r)=L^{-\tau,-\tau}(k,r,r),
\ee
as readily follows from~\eqref{Lttp} and~\eqref{Lass}. Thus  the residue integrals involving~$\tau'=\tau$ are absolutely convergent. (For~$\tau'=-\tau$ absolute convergence is immediate, since then the denominator has no zeros.)

Next, we use \eqref{cTst} and Fubini's theorem to write
\be\label{Lalim}
(\cF^* f,\cF^* g)_2= \lim_{\Lambda \to\infty} \intR dr\, \overline{f(r)}\intR dr'\, g(r')I(\Lambda ,r,r'),
\ee
with
\be
I (\Lambda ,r,r')\equiv  \frac{1}{2\pi}\int_{0}^\Lambda  dk\, \big(\Psi(r,k) \Psi( r',-k)+\Psi(-r,k) \Psi(- r',-k)\big),
\ee
cf.~\eqref{Fp}--\eqref{Psiconj}. Using~\eqref{Psiform}, \eqref{mtform}--\eqref{hwdef}, and our evenness assumptions~\eqref{Lass}, we now deduce
\be\label{IhJ}
I (\Lambda ,r,r')=\frac{1}{4\pi}w(r)^{1/2}w(r')^{1/2}\sumtt \int_{-\Lambda }^\Lambda  dk\,  \hat{w}(k)J^{\tau,\tau'} (k,r,r'),
\ee
where
\be\label{Jddef}
 J^{\tau,\tau'} (k,r,r')\equiv L^{\tau,\tau'}(k,r,r')  \exp(ik (\tau r-\tau' r')).
\ee

Fixing $\tau$ and $\tau'$, the integrand in~\eqref{IhJ} picks up a $k$-independent multiplier when~$k$ is shifted by $i\kappa$. Indeed, this is clear for the plane wave factor, whereas $\hat{w}(k)L^{\tau,\tau'}(k,r,r') $ is  $i\kappa$-periodic/$i\kappa$-antiperiodic for $\tau\tau'=+$/$\tau\tau'=-$. Using Cauchy's theorem, we therefore get
\begin{multline}\label{hJint}
 \big(1-\tau\tau'\exp(\kappa(\tau' r'-\tau r))\big)\int_{-\Lambda }^\Lambda  dk\, \hat{w}(k)J^{\tau,\tau'}(k,r,r') 
 \\
= 2\pi i\sum_{j=1}^{2\hat{L}} \hat{w}_j J^{\tau,\tau'}(k_j,r,r')+B^{\tau,\tau'}(\Lambda ,r,r'),
\end{multline}
where
\be\label{BdttLam}
B^{\tau,\tau'} (\Lambda ,r,r')\equiv
-\left( \int_\Lambda ^{\Lambda +i\kappa}+\int_{-\Lambda +i\kappa}^{-\Lambda } \right) dk\, \hat{w}(k)J^{\tau,\tau'}(k,r,r').
\ee
As before, we choose $\Lambda $ sufficiently large so that the   contour   encloses all of the $\hat{w}$-poles.

We proceed to rewrite   the residue sum   as
\begin{multline}\label{ressum2}
\sum_{j=1}^{2\hat{L}} \hat{w}_j J^{\tau,\tau'}(k_j,r,r')
 =   \sum_{j=1}^{\hat{L}} \hat{w}_j \Big(
 \exp(i k_j(\tau r-\tau' r'))
 L^{\tau,\tau'}( k_j,r,r')
  \\
-\tau\tau'\exp(\kappa(\tau' r'-\tau r))\exp(-i k_j(\tau r-\tau' r'))
L^{\tau,\tau'}(- k_j,r,r')
  \Big),
\end{multline}
where we used \eqref{whpoles}--\eqref{whresid}. Hence we deduce
\begin{multline}
\sum_{\tau,\tau'=+,-}\sum_{j=1}^{2\hat{L}}\hat{w}_j J^{\tau,\tau'} (k_j,r,r')\big/[1-\tau\tau'\exp(\kappa(\tau' r'-\tau r))]
\\
=\sum_{j=1}^{\hat{L}}\sum_{\tau,\tau'=+,-}\frac{ \hat{w}_j}{1-\tau\tau'\exp(\kappa(\tau' r'- \tau r))} 
\big[\exp(ik_j(\tau r-\tau' r'))L^{\tau,\tau'}( k_j,r,r')
\\
-\tau\tau'\exp(\kappa(\tau' r'- \tau r))\exp(-ik_j(\tau r-\tau' r'))L^{-\tau,-\tau'}( k_j,r,r')\big],
\end{multline}
where we have used an identity equivalent to~\eqref{Lass}, viz.,
\be\label{Lasseq}
L^{\tau,\tau'}(-k,r,r')=L^{-\tau,-\tau'}(k,r,r').
\ee
In the second sum we now take $\tau,\tau'\to -\tau,-\tau'$, thus  arriving at
 \be\label{crit2}
2\sum_{j=1}^{\hat{L}} \sum_{\tau,\tau'=+,-}\frac{ \hat{w}_j }{1-\tau\tau'\exp(\kappa(\tau' r'- \tau r))} 
\exp(ik_j(\tau r-\tau' r'))L^{\tau,\tau'}( k_j,r,r').
\ee
Combining this with~\eqref{Lalim}--\eqref{hJint}, we readily obtain the residue integrals on the rhs of~\eqref{cFstres}.

Next, we change variables to rewrite \eqref{BdttLam} as
\be
B^{\tau,\tau'} (\Lambda ,r,r')=
 \int_{\Lambda-i\kappa/2} ^{\Lambda +i\kappa/2}  dk\, \hat{w}(k+i\kappa/2)\big(J^{\tau,\tau'}(-k+i\kappa/2,r,r')-J^{\tau,\tau'}(k+i\kappa/2,r,r')\big).
\ee
From this we infer, using \eqref{Jddef},
\begin{multline}\label{bndint}
 \frac{B^{\tau,\tau'} (\Lambda ,r,r')}{1-\tau\tau'\exp(\kappa(\tau' r'-\tau r))} =\frac{\exp(\kappa(\tau' r'-\tau r)/2)}{1-\tau\tau'\exp(\kappa(\tau' r'-\tau r))} \int_{\Lambda -i\kappa/2}^{\Lambda +i\kappa/2}  dk\, \hat{w}(k+i\kappa/2)
\\
\times\sum_{\sigma=+,-}\sigma \exp(i\sigma k(\tau' r'-\tau r))L^{\tau,\tau'}(-\sigma k+i\kappa/2,r,r') .
\end{multline}

We now study the $\Lambda\to\infty$ limit of the four terms
\be
T^{\tau,\tau'}(\Lambda)\equiv \frac{1}{4\pi}\intR dr\, \overline{f(r)}\intR dr'\, g(r')w(r)^{1/2}w(r')^{1/2}\frac{B^{\tau,\tau'} (\Lambda ,r,r')}{1-\tau\tau'\exp(\kappa(\tau' r'-\tau r))},
\ee
contributing to the rhs of~\eqref{Lalim}. For the cases $\tau'=-\tau$, it follows from the uniform asymptotics~\eqref{mpask}--\eqref{mmask} that we have
\be
\lim_{\Lambda\to\infty}T^{\tau,-\tau}(\Lambda)=0,\ \ \ \tau=+,-.
\ee
(Note that in \eqref{bndint} we may replace $\hat{w}(k+i\kappa/2)$ by $\hat{w}(-\sigma k+i\kappa/2)$, since $\hat{w}$ is even and $i\kappa$-periodic.) For the cases $\tau'=\tau$, the rhs of~\eqref{bndint} can be rewritten as
\begin{multline}\label{brhs}
\frac{1}{2\sinh(\kappa (r'-r)/2)}  \int_{\Lambda -i\kappa/2}^{\Lambda +i\kappa/2}  dk\, \sum_{\alpha=+,-}\alpha 
  \exp(i\alpha  k( r- r'))
  \\
  \times\hat{w}(\alpha k+i\kappa/2)L^{+,+}(\alpha k+i\kappa/2,r,r'),
 \end{multline}
 where we used \eqref{Lasseq} with $\tau'=\tau=-$ and $i\kappa$-periodicity. 
 
We now change variables
\be\label{stt2}
t=\kappa r/2,\ \ t'=\kappa r'/2,\ \ s=2k/\kappa,
\ee
and set
\be
G_{\alpha}(s,t,t')\equiv \hat{w}(\alpha \kappa (s+i)/2)L^{+,+}(\alpha \kappa (s+i)/2, 2t/\kappa,2t'/\kappa),\ \ \ \alpha=+,-.
\ee
Then \eqref{brhs} can be rewritten as
\be
\frac{\kappa}{4\sinh(t'-t)}\int_{2\Lambda /\kappa -i}^{2\Lambda /\kappa +i}ds\, \sum_{\alpha=+,-}\alpha\exp(i\alpha s(t-t'))G_{\alpha}(s,t,t').
\ee
 Thanks to~\eqref{keven},  the functions $G_{\pm}(s,t,t')$ satisfy~\eqref{Geq}. Moreover, our assumptions \eqref{mpask}--\eqref{mexk} about the $\re k\to\pm \infty$ asymptotics of the coefficients $m^{\tau}(r,k)$  ensure that $G_{\pm}(s,t,t')$ also satisfy~\eqref{Gpmas}--\eqref{rhojpart}, with $\eta=\hat{\gamma}\kappa/2$, and
 \be
A_1(t,t')=C(2t/\kappa)\overline{C(2t'/\kappa)},\ \ \ \ \ A_2(t,t')= \overline{A_1(t,t')}. 
\ee
  Putting
\be
\phi(t,t')\equiv \frac{1}{4\pi\kappa}\,\overline{f(2t/\kappa)}g(2t'/\kappa)w(2t/\kappa)^{1/2}w(2t'/\kappa)^{1/2},
\ee
it follows that the assumptions of Lemma~C.1 are obeyed, with
\be
A(t)=|C(2t/\kappa)|^2.
\ee
 In view of~\eqref{cw}, we therefore get
\be
\lim_{\Lambda\to\infty}T^{\tau,\tau}(\Lambda)=(f,g)_1/2,\ \ \ \tau=+,-,
\ee
and hence the theorem follows.
\end{proof}

Just as for the residue sum in Theorem~A.1,  in the main text it is more expedient to use an alternative form for the residue sum in Theorem~B.1, as specified by the following corollary.

\begin{corollary}
Setting
\be\label{defR}
R(r,r')\equiv iw(r)^{1/2}w(r')^{1/2}\sum_{j=1}^{\hat{L}} \hat{w}_j 
  \sum_{\tau,\tau'=+,-}\frac{ \Lambda^{\tau,\tau'}(k_j,r,r')}{1-\tau\tau'\exp(\kappa (\tau' r'-\tau r))},
  \ee
  where
\be\label{Lamdef}
\Lambda^{\tau,\tau'}(k,r,r')\equiv  \lambda^{\tau}(r, k)\lambda^{\tau'}(r',- k)+\lambda^{-\tau}(- r,k)\lambda^{-\tau'}(-r',- k), 
\ee
\be\label{lamdef}
\lambda^{\tau}(r,k)\equiv \exp(i\tau rk)\ell^{\tau}(r,k),
\ee
the operator $\cF^*$ is isometric if and only if the residue sum~$R(r,r')$ vanishes.
\end{corollary}
\begin{proof}
This easily follows from~\eqref{cFstres} and~\eqref{Lttp}.  
\end{proof}

For the examples~\eqref{wsp}--\eqref{ellnup}, we should distinguish two cases for the simple pole locations~$k_1$ and $i\kappa-ik_1$ in the period strip  $\im k\in(0,\kappa)$, viz.,
\be\label{kap1}
k_1=2i\kappa\varphi/\pi,\ \ \ \varphi\in(0,\pi/2), \ \ \ \varphi\ne \pi/4,
\ee
and
\be\label{kap2}
k_1=2i\kappa\varphi/\pi -i\kappa,\ \ \ \varphi\in(\pi/2,\pi),\ \ \ \varphi\ne 3\pi/4.
\ee
Also, we obtain 
\be
\ell_{\sigma}^-(r,\pm k_1)=2i\sigma\sinh(2i\varphi)\exp(-\pi r/\rho),
\ee
\be
\ell_{\sigma}^+(r, k_1)=\mp 2i\sinh(2i\varphi)\exp(\pi r/\rho),
\ee 
\be
\ell_{\sigma}^+(r, -k_1)=\mp 2i\sinh(2i\varphi)\big(2\cosh(2i\varphi)\exp(-\pi r/\rho)-\exp(\pi r/\rho)\big),
\ee
so that \eqref{Lttp} yields
\be\label{Lp}
L^{\tau,\tau}( k_1,r,r')=-8\sinh(2i\varphi)^2\cosh(2i\varphi)\exp(\tau \pi (r-r')/\rho),\ \ \tau=+,-,
\ee
\be\label{Lm}
L^{\tau,-\tau}( k_1,r,r')=\pm 8\sigma\sinh(2i\varphi)^2\cosh(2i\varphi)\exp(\tau \pi (r+r')/\rho),\ \ \sigma,\tau=+,-,
\ee 
with the upper/lower sign referring to~\eqref{kap1}/\eqref{kap2}. From this we calculate ratios
\be
L^{-,-}/L^{+,+} =e^{2\pi(r'-r)/\rho},\ \ 
L^{+,-}/L^{+,+} =\mp \sigma e^{2\pi r'/\rho},
\ \
L^{-,+} /L^{+,+} =\mp \sigma e^{-2\pi r/\rho}.
\ee

With~\eqref{kap1} in force, the residue sum in~\eqref{cFstres}  is therefore proportional to
\begin{multline}
\frac{1}{1-\exp(\kappa(r'-r))}+\frac{\exp((2ik_1+2\pi/\rho)(r'-r))}{1-\exp(\kappa(-r'+r))}
\\
-\sigma\frac{\exp((2ik_1+2\pi/\rho)r')}{1+\exp(\kappa(-r'-r))}
-\sigma\frac{\exp(-(2ik_1+2\pi/\rho)r)}{1+\exp(\kappa(r'+r))}.
\end{multline}
For this to vanish we can choose
\be
\sigma =+,\ \ \ ik_1+\pi/\rho=0\Leftrightarrow \varphi =\pi^2/2\rho\kappa=\phi_0,
\ee
cf.~\eqref{rk0}. (Indeed, vanishing boils down to~\eqref{id1} with $A,A'\to iA, iA'$.) As a consequence, the transform $\cF_+(\phi_0)$ is unitary for $\phi_0=\pi^2/2\rho\kappa\in(0,\pi/2)$, $\phi_0\ne \pi/4$. Using a continuity argument, it follows that $\cF_+(\phi_0)$ is unitary for the double-pole case $\phi_0=\pi/4$, too.

Next, we consider the case~\eqref{kap2}. Then~\eqref{hw} yields the residue
\be\label{hw1}
\hat{w}_1=\frac{\kappa}{4\pi\sinh(4i\varphi)}.
\ee
From~\eqref{Lp}--\eqref{Lm} we now deduce
\be
R(r,r')=\cN(\varphi)w(r)^{1/2}w(r')^{1/2}S_{\sigma}(\varphi;r,r'),
\ee
where we have introduced
\be
\cN(\varphi)\equiv \frac{\kappa\sin(2\varphi)}{\pi}\in (-\infty,0),\ \ \ \varphi\in(\pi/2,\pi),
\ee
and
\be\label{Sdef}
S_{\sigma}(\varphi;r,r')\equiv \sum_{\tau=+,-}\Big(\frac{\exp(\tau (ik_1+\pi/\rho)( r- r'))}{1-\exp(\tau\kappa (r' - r))} 
+\sigma
\frac{\exp(\tau (ik_1+\pi/\rho)( r+ r'))}{1+\exp(-\tau\kappa (r' + r))}  \Big).
\ee

We proceed to study the two $\varphi$-choices $\phi_0$, $\phi_e\in(\pi/2,\pi)$, for which we already know that $\cF_+(\phi_0)$ and $\cF_{\pm}(\phi_e)$ are isometric. (Recall $\rho\kappa\in(\pi/2,\pi)$ in the first case and $\rho\kappa\in (\pi,\infty)$ in the second one, cf.~\eqref{rk0}--\eqref{rke}.) For the case
 $\varphi=\phi_0$ we obtain, putting 
\be\label{AAp}
A\equiv \exp(\kappa r),\ \ \ A'\equiv \exp(\kappa r'),
\ee
and setting  $ik_1+\pi/\rho=\kappa$ in~\eqref{Sdef},
\begin{multline}\label{S0}
S_+(\phi_0;r,r')=\frac{A}{A'}\,\frac{1}{1-A'/A}+\frac{A'}{A}\,\frac{1}{1-A/A'} +AA'\,\frac{1}{1+1/A'A}+\frac{1}{A'A}\,\frac{1}{1+A'A} 
\\
=\frac{A'}{A}+\frac{A}{A'}+A'A+\frac{1}{A'A}=4\cosh(\kappa r)\cosh(\kappa r').
\end{multline}
Introducing the function (cf.~\eqref{wsp})
\be\label{Psi0}
\Psi_0(r)\equiv 2\cosh(\kappa r) w_0(r)^{1/2}, \ \ \ w_0(r)\equiv 1\big/4\sinh \Big(\frac{\pi}{\rho}\big( r+i\frac{\pi}{2\kappa}\big)\Big)\sinh \Big(\frac{\pi}{\rho}\big( r-i\frac{\pi}{2\kappa}\big)\Big),
\ee
 we get
\be
\Psi_0(\cdot)\in \cH, \ \ \rho\kappa \in(\pi/2,\pi).
\ee
Thus we can rewrite~\eqref{cFstres} as
\be\label{F0st}
(\cF_+(\phi_0)^* f,\cF_+(\phi_0)^* g)_2= (f,g)_1+\frac{\kappa \sin (\pi^2/\rho\kappa)}{\pi}(f,\Psi_0)_1(\Psi_0,g)_1.
\ee

Now $\cF_+(\phi_0)$ is isometric, so $\cF_+(\phi_0)^*$ has a continuous extension from~$\cC$ to a partial isometry such that
\be 
\cF_+(\phi_0)^*\cF_+(\phi_0)={\bf 1}_{\hat{\cH} },\ \ \ \cF_+(\phi_0)\cF_+(\phi_0)^*={\bf 1}_{ \cH }-P,
\ee
where $P$ is the projection on the orthocomplement of $\cF_+(\phi_0)(\hat{\cH} )$.
 In particular, this implies \eqref{F0st} holds true for all $f,g\in\cH$, which yields
 \be
 \cF_+(\phi_0)\cF_+(\phi_0)^*={\bf 1}_{ \cH }+\frac{\kappa \sin (\pi^2/\rho\kappa)}{\pi}\Psi_0\otimes \Psi_0.
 \ee 
Hence we have~$\cF_+(\phi_0)^*\Psi_0=0$ and  
\be\label{inpr0}
(\Psi_0,\Psi_0)_1=-\frac{\pi}{\kappa \sin (\pi^2/\rho\kappa)},\ \ \rho\kappa \in(\pi/2,\pi).
\ee
By continuity these formulas are also valid for the double-pole case $\phi_0=3\pi/4$.

Let us now summarize our findings regarding~$\cF_+(\phi_0)$.

\begin{proposition}
The transform~$\cF_+(\phi_0) $ defined by~\eqref{cF}--\eqref{Fm} with
\begin{multline}\label{cFp0}
\Psi(r,k)\equiv \frac{w_0(r)^{1/2}}{\sinh\Big( \frac{\pi}{\kappa}\big( k-\frac{i\pi}{\rho}\big)\Big)}
\Big( e^{irk}\Big[ e^{-\pi r/\rho}\sinh\Big( \frac{\pi}{\kappa}\big( k-\frac{i\pi}{\rho}\big)\Big) -
e^{\pi r/\rho}\sinh\big( \frac{\pi}{\kappa} k \big)\Big]
\\
+e^{-irk}e^{-\pi r/\rho}\sinh\big( \frac{i\pi^2}{\rho\kappa}  \big)\Big),\ \ \ \rho\kappa\in(\pi/2,\infty),
\end{multline}
 is unitary for $\rho\kappa\in [\pi,\infty)$. For $\rho\kappa\in (\pi/2,\pi)$ it is isometric and satisfies
\be
\cF_+(\phi_0)\cF_+(\phi_0)^*={\bf 1}_{ \cH }-\Psi_0\otimes \Psi_0/(\Psi_0,\Psi_0)_1.
\ee
Here, $w_0(r)$ and $\Psi_0(r)$ are defined by~\eqref{Psi0} and the inner product is given by~\eqref{inpr0}.
\end{proposition}

As mentioned below~\eqref{psizero}, the function~$\Psi(r,k)$ in this proposition is a reparametrized version of the $N=0$ function $\psi(a_+;x,y)$ of the main text. We also point out that from~\eqref{cFp0} one can read off that for $\rho\kappa\in[\pi,\infty)$ the function $\Psi(r,k)$ is analytic in the strip $\im k \in [0,\kappa)$. By contrast, for $\rho\kappa\in(\pi/2,\pi)$ it has a pole in this strip at~$k=i\pi/\rho-i\kappa$, and  the bound state $\Psi_0(r)$ is proportional to the residue  at this pole. Finally, we repeat that Proposition~5.1 yields a complete picture of the transform~$\cF_+(\phi_0) $ for~$\rho\kappa\le \pi/2$.

Next, we study the choice $\varphi=\phi_e$. Then we need to set $ik_1+\pi/\rho=0$ in~\eqref{Sdef}, yielding
\be
S_{\sigma}(\phi_e;r,r')= \frac{1}{1-A'/A}+ \frac{1}{1-A/A'} +\frac{\sigma}{1+1/A'A}+ \frac{\sigma}{1+A'A}. 
\ee
For $\sigma=-$ this vanishes, implying that the transform~$\cF_-(\phi_e)$ is unitary.

By contrast, we have
\be
S_{+}(\phi_e;r,r')=2.
\ee
Defining (cf.~\eqref{wsp})
\be\label{we}
 w_e(r)\equiv  1\Big/4\cosh \Big(\frac{\pi}{\rho}\big( r+i\frac{\pi}{2\kappa}\big)\Big)\cosh \Big(\frac{\pi}{\rho}\big( r-i\frac{\pi}{2\kappa}\big)\Big),
\ee
 we get
\be\label{Psie}
\Psi_e(r)\equiv   w_e(r)^{1/2} \in \cH, \ \ \rho\kappa \in(\pi,\infty).
\ee
Thus we can rewrite~\eqref{cFstres} as
\be
(\cF_+(\phi_e)^* f,\cF_+(\phi_e)^* g)_2= (f,g)_1-\frac{2\kappa \sin (\pi^2/\rho\kappa)}{\pi}(f,\Psi_e)_1(\Psi_e,g)_1.
\ee
As before, we deduce from this that we have  $\cF_+(\phi_e)^*\Psi_e=0$ and
\be\label{inpre}
(\Psi_e,\Psi_e)_1=\frac{\pi}{2\kappa \sin (\pi^2/\rho\kappa)},\ \ \ \rho\kappa\in(\pi,\infty).
\ee
Once more, by continuity this is also true for   $\phi_e=3\pi/4$.

We proceed to summarize these results.

\begin{proposition}
Consider the transforms~$\cF_{\sigma}(\phi_e) $, $\sigma=+,-$, defined by~\eqref{cF}--\eqref{Fm} with
\begin{multline}\label{cFe}
\Psi(r,k)\equiv \frac{w_e(r)^{1/2}}{\sinh\Big( \frac{\pi}{\kappa}\big( k-\frac{i\pi}{\rho}\big)\Big)}
\Big( e^{irk}\Big[ e^{-\pi r/\rho}\sinh\Big( \frac{\pi}{\kappa}\big( k-\frac{i\pi}{\rho}\big)\Big) +
e^{\pi r/\rho}\sinh\big( \frac{\pi}{\kappa} k \big)\Big]
\\
+\sigma e^{-irk}e^{-\pi r/\rho}\sinh\big( \frac{i\pi^2}{\rho\kappa}  \big)\Big), \ \ \ \rho\kappa\in(\pi,\infty),
\end{multline}
where $w_e(r)$ is given by \eqref{we}. The transform $\cF_{-}(\phi_e) $
 is unitary, whereas $\cF_{+}(\phi_e) $ is isometric and satisfies
\be
\cF_+(\phi_e)\cF_+(\phi_e)^*={\bf 1}_{ \cH }-\Psi_e\otimes \Psi_e/(\Psi_e,\Psi_e)_1.
\ee
Here, $\Psi_e(r)$ is defined by~\eqref{Psie} and the inner product is given by~\eqref{inpre}.
\end{proposition}

From \eqref{cFe} one sees that for $\sigma=-$ the function $\Psi(r,k)$ is analytic in the strip $\im k\in[0,\kappa)$. For $\sigma=+$ it has a pole in this strip at $k=i\pi/\rho$, and the bound state $\Psi_e(r)$ is proportional to the residue of~$\Psi(r,k)$ at this pole.

To conclude this appendix, we add one more explicit example satisfying all of the assumptions. We begin by noting that when we have an additional assumption
\be\label{addass}
m^-(r,k)=0,
\ee
then we obtain a reflectionless transform, and it also follows that we have
\be\label{Mr0}
M^-_{\alpha}(r,k,k')=0,\ \ \ \ \  M^+_{\alpha}(r,k,k')=m^+(-r,-k)m^+(-\alpha r,k'),\ \ \alpha=+,-,
\ee
\be\label{Lr0}
L^{\tau,-\tau}(k,r,r')=0,\ \ L^{\tau,\tau}(k,r,r')=\ell^+(\tau r,k)\ell^+(\tau r',-k),\ \ \tau=+,-,
\ee
cf.~\eqref{Mtsdef} and \eqref{Lttp}. Therefore, the critical evenness assumptions~\eqref{Mass} and~\eqref{Lass} reduce to the sole assumption
\be\label{L0ass}
\ell^+(r,k)\ell^+(r',-k)=\ell^+(-r,-k)\ell^+(-r',k).
\ee
Moreover, in view of \eqref{ressum}, the condition for the transform to be isometric reduces to 
\be\label{iso0}
0=\sum_{j=1}^Lw_j
\big[
  \mu^{+}(\de r_j,-k)\mu^{+}(\de' r_j,k')
-\exp(\rho(\de k-\de' k')) 
  \mu^{+}(-\de r_j,-k)\mu^{+}(-\de' r_j,k')\big],
  \ee
where we have used the notation~\eqref{mudef}.

We are now prepared for our last example, which yields a reflectionless transform.

\begin{proposition}
The transform~$\cF_a $ defined by~\eqref{cF}--\eqref{Fm} with
\be\label{cFa}
\Psi(r,k)\equiv w_a(r)^{1/2} e^{irk}\Big[ e^{-\pi r/\rho}+e^{\pi r/\rho}\sinh\Big( \frac{\pi}{\kappa}\big( k+\frac{i\pi}{\rho}\big)\Big)\Big/\sinh\Big( \frac{\pi}{\kappa}\big( k-\frac{i\pi}{\rho}\big)\Big) \Big], 
\ee 
\be\label{wa}
w_a(r)\equiv 
1\Big/4\cosh \Big(\frac{\pi}{\rho}\big( r+i\frac{\pi}{\kappa}\big)\Big)\cosh \Big(\frac{\pi}{\rho}\big( r-i\frac{\pi}{\kappa}\big)\Big),
\ee
  is isometric for $\rho\kappa\in (2\pi,\infty)$, and satisfies
\be
\cF_a\cF_a^*={\bf 1}_{ \cH }-\Psi_a\otimes \Psi_a/(\Psi_a,\Psi_a)_1,\ \ \ \Psi_a(r)\equiv w_a(r)^{1/2},
\ee
where
\be
(\Psi_a,\Psi_a)_1=\frac{\pi}{\kappa \sin (2\pi^2/\rho\kappa)},\ \ \ \rho\kappa\in(2\pi,\infty).
\ee
\end{proposition}
\begin{proof}
Here we have
\be\label{vhell}
v(k)=1\Big/2i \sinh\Big( \frac{\pi}{\kappa}\big( k-\frac{i\pi}{\rho}\big)\Big), \ \ell^+(r,k)=
2ie^{\pi r/\rho}\sinh\Big( \frac{\pi}{\kappa}\big( k+\frac{i\pi}{\rho}\big)\Big)+(\rho\to -\rho).
\ee
From this the evenness assumption \eqref{L0ass} is readily verified, and all other assumptions are clearly satisfied as well, with
\be
T_a(k)=\sinh\Big( \frac{\pi}{\kappa}\big( k+\frac{i\pi}{\rho}\big)\Big)\Big/\sinh\Big( \frac{\pi}{\kappa}\big( k-\frac{i\pi}{\rho}\big)\Big),\ \ \ \ \ R_a(k)=0, 
\ee 
\be
C_a(r)=2\exp(i\pi^2/\rho\kappa)\cosh(\pi r/\rho+i\pi^2/\rho\kappa).
\ee
Consequently, Theorems~A.1 and~B.1 apply, so it remains to study the residue sums.

For the residue sum on the rhs of~\eqref{iso0}  we have $L=1$ and we can take
\be\label{r1}
r_1=i\pi/\kappa +i\rho/2\in i[0,\rho),\ \ \ \rho\kappa\in(2\pi,\infty).
\ee
From~\eqref{vhell} we then obtain
\be\label{ellp}
\ell^+(\nu r_1,k)=-2\nu \sinh(2i\pi^2/\rho\kappa) \exp(\nu\pi k/\kappa),
\ee
which entails
\be
\mu^+(\nu r_1,k)=-2\nu \sinh(2i\pi^2/\rho\kappa)v(k) \exp(-\nu  \rho k/2).
\ee
From this we see that \eqref{iso0} holds true, so~$\cF_a$ is isometric.

Turning to the residue sum in Theorem~B.1, we have $\hat{L}=1$ and
\be
k_1=i\pi/\rho\in i[0,\kappa),\ \ \ \hat{w}_1=-i\kappa/4\pi\sin(2\pi^2/\rho\kappa),\ \ \rho\kappa\in(2\pi,\infty).
\ee
From \eqref{vhell} we obtain
\be
\ell^+(r,\nu k_1)=-2\nu \sin(2\pi^2/\rho\kappa)\exp(\nu\pi r/\rho),
\ee
so \eqref{Lr0} yields
\be
L^{\tau,\tau}( k_1,r,r')=-4\sin(2\pi^2/\rho\kappa)^2\exp(\tau\pi(r-r')/\rho).
\ee

As a result, we have
\bea
 \sum_{\tau=+,-}\frac{\exp(i\tau k_1( r- r'))}{1-\exp(\tau\kappa (r' - r))}L^{\tau,\tau}(  k_1,r,r') & = & -4\sin(2\pi^2/\rho\kappa)^2\sum_{\tau=+,-}\frac{1}{1-\exp(\tau\kappa (r' - r))}
  \nonumber \\
& = & -4\sin(2\pi^2/\rho\kappa)^2.
 \eea
Therefore, \eqref{cFstres} becomes
\be
(\cF_a^* f,\cF_a^* g)_2= (f,g)_1-\frac{\kappa\sin(2\pi^2/\rho\kappa)}{\pi}(f,\Psi_a)_1(\Psi_a,g)_1,
\ee
whence the proposition follows.
\end{proof}

Once more, the bound state $\Psi_a(r)$ is proportional to the residue of $\Psi(r,k)$ at the pole $k=i\pi/\rho$ in the strip $\im k\in [0,\kappa)$.
Admittedly, the transform~$\cF_a$ may seem to come out of the blue. We have included it, because the function $\Psi(r,k)$~\eqref{cFa} corresponds to the function $\Psi(2a_-;x,y)$ in the main text, cf.~\eqref{psiNr=0}. In Section~4 of~\cite{R00} the reflectionless transforms associated with $\Psi((N+1)a_-;x,y)$, $N\in\N$, were already analyzed, and  $\Psi(r,k)$ arises from the function given by Eq.~(4.29) in~\cite{R00}; as such, it satisfies the eigenvalue equation
\begin{multline}\label{aeigen}
\left( \frac{\cosh(\pi r/\rho -i\pi^2/\rho\kappa)}{\cosh(\pi r/\rho)}\exp(i\pi \partial_r/\kappa)+(i\to -i)\right)w_a(r)^{-1/2}\Psi(r,k)
\\
=2\cosh(\pi k/\kappa)w_a(r)^{-1/2}\Psi(r,k).
\end{multline}
(To tie this in with $H_o$~\eqref{Ho}, first take $\rho \leftrightarrow \pi/\kappa$, then put $\kappa =1$. We have deviated from the reparametrization \eqref{xryk} by swapping~$a_+$ and~$a_-$, so that the $r$- and $k$-periodicity assumptions apply to the reflectionless wave function~\eqref{cFa}.) 

In fact, all of the assumptions (including~\eqref{addass}) are obeyed by the arbitrary-$N$ attractive eigenfunctions from Section~4 in~\cite{R00}. In particular, the evenness assumption~\eqref{L0ass} is obeyed due to parity features of the latter.
 

\section{A boundary lemma}
 
The following lemma yields a template for handling the $\Lambda \to\infty$ limits of the boundary terms arising in the analysis of the transform $\cF$ and its adjoint. It is adapted from the proof of Theorem~2.1 in Appendix~A of~\cite{R00}.

In order to handle $\cF$ and~$\cF^*$ at once, we  start from two $\C$-valued functions~$G_{\pm}(s,t,t')$ that are defined on~$\{\re s>0\}\times \R^2$ and have the following features. They are analytic in~$s$ and smooth in~$t,t'$. Their dominant asymptotics as $\re s\to\infty$ is given by smooth functions~$A_j(t,t')$, $j=1,2$, in the sense that  
\be\label{Gpmas}
G_+(s,t,t')=A_1(t,t')+\rho_1(s,t,t'),\ \ \ G_-(s,t,t')=A_2(t,t')+\rho_2(s,t,t'),
\ee
where
\be\label{rhoj}
\rho_j(s,t,t')=O(\exp(-\eta  s)),\ \ \re s\to\infty,\ \ \eta>0,
\ee
\be\label{rhojpart}
\partial_3\rho_j(s,t,t')=O(\exp(-\eta  s)),\ \ \re s\to\infty,\ \ \eta>0,
\ee
with implied constants that are uniform for $\im s,t,t'$ in compact subsets of~$\R$. Finally, we assume
\be\label{Geq}
G_+(s,t,t)=G_-(s,t,t),\ \ \ \re s>0,\ \ t\in\R,
\ee
and
\be\label{Aeq}
A_j(t,t)=A(t),\ \ j=1,2,\ \ \ t\in\R.
\ee
We are now prepared for our boundary lemma.

\begin{lemma}
Letting $\phi(t,t')\in C_0^\infty (\R^2)$, define
\be\label{IR}
I_R\equiv \int_{\R^2}dtdt'\, \phi(t,t') \frac{B_R(t,t')}{\sinh(t'-t)},
\ee
where
\be\label{BR}
B_R(t,t')\equiv \int_{R-i}^{R+i}ds\, \sum_{\alpha=+,-}\alpha G_{\alpha}(s,t,t')\exp(i\alpha s(t-t')) ,\ \ \ R>0.
\ee
Then we have
\be
\lim_{R\to\infty} I_R=4\pi \intR \phi(t,t)A(t)dt.
\ee
\end{lemma}
\begin{proof}
We begin by noting that by~\eqref{Geq} the integrand in~\eqref{BR} is a smooth function of $(t,t')\in\R^2$ that vanishes for $t=t'$, so that $B_R(t,t')$ has the same properties. Therefore the integrand in~\eqref{IR} belongs to $C_0^\infty(\R^2)$. Hence $I_R$ is well defined.

Next, we write
\be\label{Btel}
B_R(t,t')= \sum_{j=1}^4
\int_{R-i}^{R+i}ds\, b_j(s,t,t'),
\ee
\be\label{b1}
b_1\equiv i\sin
(s(t-t'))A_+(t,t'),\ \ \ 
b_2\equiv i\sin (s(t-t'))\rho_{+}(s,t,t'),
\ee
\be
b_3\equiv \cos
(s(t-t'))A_-(t,t'),\ \ \ 
b_4\equiv \cos (s(t-t'))\rho_{-}(s,t,t'),
\ee
where
\be\label{Apm}
A_{\pm}(t,t')    \equiv A_1(t,t')\pm A_2(t,t'),  
\ee
\be\label{rhopm}
\rho_{\pm}(s,t,t')    \equiv \rho_1(s,t,t')\pm \rho_2(s,t,t'). 
\ee
Each of the terms in the sum on the rhs of \eqref{Btel}
is a smooth function of $t$ and $t'$ that vanishes
for $t=t'$. Thus the integrals
\be\label{defIj}
I_j(R)    \equiv    \int_{\R^2} dt dt'\,
\frac{\phi(t,t')}{\sinh(t'-t)} \int_{R-i}^{R+i}ds\, b_j(s,t,t'),\ \
j=1,\ldots,4,
\ee
are well defined and   it suffices to prove
\be\label{I1}
\lim_{R\to\infty}I_1(R)=4\pi \intR \phi(t,t)A(t)dt,
\ee
\be\label{Ij}
\lim_{R\to\infty}I_j(R)=0,\ \ \ j=2,3,4.
\ee

In order to prove (\ref{I1}), we use (\ref{defIj}) and
(\ref{b1}) to calculate
\be
I_1(R)= \int_{\R^2} dt dt' \,\phi(t,t')\frac{2\sin
R(t-t')}{t-t'}A_+(t,t').
\ee
Invoking the tempered distribution limit
\be
\lim_{R\to\infty}\frac{\sin Rx}{x}=\pi \delta(x),
\ee
and using our assumption~\eqref{Aeq}, we now deduce (\ref{I1}).

We continue by studying the integral $I_2(R)$. It can be written 
\be
i\int_{\R^2} dt dt' \,\phi(t,t') \frac{t-t'}{\sinh
(t'-t)}\int_{R-i}^{R+i}ds\,\frac{\sin
(s(t-t'))}{t-t'}\rho_{+}(s,t,t').
\ee
The integrand of the $s$-integral can be estimated by using
\be
\left| \frac{\sin
s(t-t')}{t-t'}\right|=\frac{1}{2}\left|
\int_{-s}^sdxe^{ix(t-t')}\right| \le |s|e^{|t-t'||\im
s|},
\ee
\be
\rho_{+}(s,t,t')=O(\exp (-\eta s)),\ \ \ \re s\to \infty,
\ee
where the latter bound is uniform for $\im s, t,t'$ in $\R$-compacts, cf.~\eqref{rhoj} and \eqref{rhopm}. Hence we easily deduce \eqref{Ij} for
$j=2$.

Consider next $I_3(R)$. This integral equals 
\be
2\int_{\R^2} dt dt' \,\phi(t,t')\cos
R(t-t') \frac{A_-(t,t') }{t-t'}.
\ee
Thus its integrand equals $\cos R(t-t')$ times a
function in $C_0^{\infty}(\R^2)$. Its $R\to\infty$ limit
then vanishes by virtue of the Riemann-Lebesgue lemma.

Finally, we take $j=4$ in (\ref{defIj}) and write 
\be\label{Ifin}
I_4(R)=\int_{\R^2} dt dt' \,\phi(t,t') \frac{t-t'}{\sinh
(t'-t)}
\int_{R-i}^{R+i}ds\,\cos
(s(t-t'))\frac{\rho_{-}(s,t,t')}{t-t'}.
\ee
Now due to our assumptions~\eqref{Geq} and~\eqref{Aeq}, the function $\rho_{-}(s,t,t')$ vanishes for $t=t'$. For $(t,t')$ belonging to the support of $\phi$ we therefore have
\be\label{rhomin}
\left|\frac{\rho_{-}(s,t,t')}{t-t'}\right| =\left|
\frac{1}{t-t'}\int_t^{t'}du\,\partial_3\rho_{-}(s,t,u)\right|\le
\max_{(t,t',\theta)\in {\rm supp}( \phi)\times [0,1]}
|\partial_3\rho_{-}(s,t,t+\theta(t'-t))|.
\ee
Invoking our assumption~\eqref{rhojpart}, we infer that the rhs of
(\ref{rhomin}) is $O(\exp (-\eta\re s))$ for $\re s\to
\infty$, uniformly for $\im s,t,t'$ in $\R$-compacts. Clearly,
this entails that (\ref{Ifin}) has limit 0 for $R\to
\infty$, completing the proof of the lemma.
\end{proof}


\section{Time-dependent scattering theory}

In the previous appendices we have not introduced any dynamics, yet we have freely referred to well-known objects from scattering theory, including reflection/transmission coefficients and $S$-matrix. In this appendix we shall explain the relation to time-dependent scattering theory in the general setting of Appendices~A and~B. 

First, we recall that this setting solely involves a number of features of the function $\Psi(r,k)$ in terms of which the transform $\cF$ is defined, cf.~\eqref{cF}--\eqref{Fm}. These features were chosen such that they form sufficient hypotheses for Theorems~A.1 and B.1 to hold true. We have already shown that there are nontrivial concrete realizations of these assumptions, yielding the transforms~$\cF_{\pm}(\varphi)$, cf.~\eqref{wsp}--\eqref{ellnup}. These assumptions are in force throughout this appendix.

To connect the general transform~$\cF$ to time-dependent scattering theory, however, an additional assumption is critical. This assumption is that $\cF$ is an isometry, or equivalently, that the residue sums~\eqref{Rdk} vanish. Admittedly, at face value this   seems an assumption of a quite inaccessible nature. The special cases worked out below Corollary~A.2 show that it is not vacuous, but they only involve the simplest case~$L=1$. Indeed, without the further examples coming from the main text, it would be far from obvious that there exist nontrivial arbitrary-$L$ functions $\Psi(r,k)$ satisfying all of the pertinent assumptions. 

At any rate, in this appendix we are not concerned with explicit realizations of our assumptions, this being the focus of the main text. More is true: We need not even restrict attention to the special dynamics that pertains to the relativistic hyperbolic Calogero-Moser systems. 

Specifically, we consider a vast class of dynamics to which the latter belongs. We start from multiplication operators on $\hat{\cH}$ of the diagonal form
\be\label{hatM}
(\hat{M}f)_{\de}(k)=\mu(k)f_{\de}(k),\ \  \ \de=+,-,\ \ \ k>0.
\ee
Here, $\mu(k)$ is any real-valued, smooth, even function on~$\R$ whose (odd) derivative 
satisfies
\be\label{muder}
\mu'(k)>0,\ \ \  k>0.
\ee
Thus $\mu(k)$ is strictly increasing on~$(0,\infty)$, but not necessarily unbounded. (For example, the function $1-\exp(-k^2)$ satisfies the assumptions.) On its maximal multiplication domain $\cD(\hat{M})\subset \hat{\cH}$, the operator~$\hat{M}$ is self-adjoint, and the subspace $\hat{\cC}=C_0^{\infty}((0,\infty))^2\subset \cD(\hat{M})$ is a core (domain of essential self-adjointness).

We can now define an operator $M$ on the subspace
\be\label{DMdef}
\cD(M)\equiv \cF(\cD(\hat{M}))\subset \cH,
\ee
by setting
\be\label{Mdef}
M\cF f\equiv \cF \hat{M} f,\ \ f\in \cD(\hat{M}).
\ee
Since $\cF$ is isometric (by asumption), this yields a self-adjoint operator~$M$ on the Hilbert space 
\be\label{ranF}
\cH^r\equiv \cF(\hat{\cH}).
\ee
 Furthermore, $\cF(\hat{\cC})$ is a core for~$M$. In case $\cH^r$ is a proper subspace of $\cH$ (so that $\cF$ is not unitary), we define $M$ to be equal to an arbitrary self-adjoint operator on the orthogonal complement of~$\cH^r$; this choice plays no role in the scattering theory of this appendix, but it will be made definite in the main text for the dynamics at issue there, cf.~the last paragraph of Section~4. 

With the dynamics $M$ thus defined as a self-adjoint operator on $\cH$, we now consider the associated `interacting' unitary time evolution 
\be
\cU(t)\equiv \exp(-itM),\ \ t\in\R,
\ee
in relation to a `free' evolution defined by using the Fourier transform~$\cF_0$, cf.~\eqref{free}. Specifically, with~$\cF_0$ in the role of~$\cF$, we obtain a self-adjoint operator $M_0$ on~$\cH$ with associated time evolution
\be\label{U0t}
\cU_0(t)\equiv \exp(-itM_0),\ \ t\in\R.
\ee
This evolution can be compared to $\cU(t)$ in the usual sense of time-dependent scattering theory. We recall that this amounts to studying the (strong) limits of the unitary family $\cU(-t)\cU_0(t)$, yielding the isometric wave operators
\be\label{Wpm}
\cW_{\pm}\equiv \lim_{t\to \pm \infty}\cU(-t)\cU_0(t),
\ee
in case the limits exist~\cite{RS79}.

The crux is now that our assumptions already suffice  to prove that these limits do exist. Moreover, the transform $\cF$ is substantially equivalent to the incoming wave operator $\cW_-$, in the sense that it is equal to $\cW_-\cF_0$.
To avoid possible confusion, we should  stress that this equality does not give rise to an isometry proof for $\cF$. Indeed, we need to \emph{assume} that~$\cF$ is isometric to begin with, so as to obtain a unitary evolution~$\cU(t)$, cf.~\eqref{DMdef}--\eqref{Mdef}. 

Before stating the pertinent theorem, we specify the function~$\mu(k)$ that corresponds to the defining relativistic Calogero-Moser dynamics in the main text: It reads
\be\label{muCM}
\mu_{CM}(k)=2\cosh(\rho k).
\ee
The action of the associated operator $M_{CM}$ on the core $\cF(\hat{\cC})$ is that of the analytic difference operator
\be\label{HCM}
H_{CM}=\exp(i\rho \partial_r)+\exp(-i\rho \partial_r).
\ee
Indeed, $H_{CM}$ yields the eigenvalues~\eqref{muCM} when acting on~$\Psi(r,k)$. To be more specific, in the main text the $w(r)$-poles are on the imaginary axis and the coefficients~$w(r)^{1/2}m^{\pm}(r,k)$ multiplying the plane waves~$\exp(\pm irk)$ are $i\rho$-periodic for $|\re r|>0$. (Note that the assumptions~\eqref{mpas}--\eqref{mnas} entail that this also holds true in the present  axiomatic setting when we let $|\re r|>|\re r_j|$, $j=1,\ldots,2L$, so as to avoid the branch points.) 

\begin{theorem}
Assuming the transform~$\cF$ is isometric,  define dynamics by~\eqref{hatM}--\eqref{U0t}. Then the strong limits \eqref{Wpm} exist. They are explicitly given by 
\be\label{Wm}
\cW_-=\cF \cF_0^*,
\ee
and
\be\label{Wp}
\cW_+=\cW_-\cF_0S(\cdot)^*\cF_0^*,
\ee
where $S(k)$ is the unitary matrix multiplication operator~\eqref{Sm} on~$\hat{\cH}$.
\end{theorem}
\begin{proof}
We first prove that~$\cW_-$ exists and is given by~\eqref{Wm}. Since~$\cF$ is isometric and the time evolutions are unitary, we need only show that the $\cH$-norm
\be\label{gm}
 \| \big( e^{itM}e^{-itM_0}\cF_0-\cF\big)f\|_1=\|   (\cF_0-\cF)e^{-it\hat{M}}f\|_1,
\ee
with $f=(f_+,f_-)$ an arbitrary function in the dense subspace~$\hat{\cC}$ of~$\hat{\cH}$, vanishes for~$t\to-\infty$. To this end, consider
\bea
 \big(  (\cF_0-\cF)e^{-it\hat{M}}f\big)(r) &=&\frac{1}{\sqrt{2\pi}}\intp dk\, \sum_{\de=+,-}\de \big(e^{i\de rk}-\Psi(\de r,k)\big)e^{-it\mu(k)}f_{\de}(k)
  \nonumber \\
 &  =: & \frac{1}{\sqrt{2\pi}}\sum_{\de=+,-}\de \psi_{\de}(r).
\eea
Taking $r\to -r$ in $\psi_-(r)$, we see that it suffices to prove that the $\cH$-norm of the function
\be
g(t;r)\equiv \intp dk\,   \big(e^{i rk}-\Psi( r,k)\big)e^{-it\mu(k)}f(k),\ \ \ f\in C_0^\infty((0,\infty)),
\ee
vanishes for $t\to -\infty$. Recalling~\eqref{Psiform}, we obtain
\be
g(t;r)=\sum_{\tau=+,-}\tau g^{\tau}(t;r),
\ee
where
\be
g^+(t;r)\equiv \intp dk\,   e^{i rk}(1-w(r)^{1/2}m^+(r,k))e^{-it\mu(k)}f(k),
\ee
\be
 g^-(t;r)\equiv \intp dk\,  e^{-irk}w(r)^{1/2}m^-(r,k)e^{-it\mu(k)}f(k).
\ee
From this we see that we need only change variables $k\to x\equiv \mu(k)$ and invoke the Riemann-Lebesgue lemma to obtain convergence to zero of $g^{\pm}(t;r)$ for $t\to -\infty$ and fixed~$r$. (Note that our assumption~\eqref{muder} ensures that this change of variables is well defined.)

As a consequence, it remains to supplement this pointwise convergence with $L^1(\R)$ dominating functions for the functions~$|g^{\pm}(t;r)|^2$. To do so, we split up the integration over~$\R$ into intervals $[-R,R]$, $(-\infty,-R]$, and $[R,\infty)$, where~$R$ is chosen large enough so that we may invoke the  asymptotics \eqref{mpas}--\eqref{mnas}. On the first interval the functions~$|g^{\pm}(t;r)|^2$ are clearly bounded uniformly in~$t$, so this contribution vanishes for $t\to -\infty$.

Next, we bound~$|g^+(t;r)|^2$ on the two tail intervals. For the left interval we use~\eqref{mpas} to obtain an exponentially decreasing dominating function. On the right interval we get two contributions from~\eqref{mpas}, the second one again yielding an exponentially decreasing dominating function. Thus we are left with obtaining a suitable bound for the function
\be\label{gright}
\intp dk\,   e^{i rk}(1-T(k))e^{-it\mu(k)}f(k),\ \ r\in[R,\infty),
\ee
with $t\le -1$, say. To this end, we write the exponentials as
\be\label{exp1}
(r-t\mu'(k))^{-1}(-i\partial_k)\exp(irk-it\mu(k)),
\ee
noting that by our assumption~\eqref{muder}  the denominator $r-t\mu'(k)$ is bounded  away from zero on the compact support of~$f(k)$ for~$t\le -1$ and $r\ge R$. 
Integrating by parts and estimating in the obvious way, this yields an $O(1/r)$-majorization that is uniform for~$t\le -1$. Thus the modulus squared of the function~\eqref{gright} is bounded above by $C/r^2$ for all $t\le -1$, so that by the dominated convergence theorem its $L^2([R,\infty))$-norm vanishes for~$t\to -\infty$.  

It remains to bound~$|g^-(t;r)|^2$ on the tail intervals. On the right one we get an exponentially decreasing dominating function from~\eqref{mnas}. On the left, the second term in~\eqref{mnas} yields again an exponentially decreasing dominating function, so it remains to bound
\be\label{gleft}
\intp dk\,   e^{-i rk}R(k)e^{-it\mu(k)}f(k),\ \ r\in(-\infty,-R],
\ee
uniformly for $t\le -1$. Writing
\be\label{exp2}
(r+t\mu'(k))^{-1}(i\partial_k)\exp(-irk-it\mu(k)),
\ee
and integrating by parts, we readily obtain a uniform $O(1/r)$-majorization that suits our purpose. As a result, we have now proved existence of~$\cW_-$ and its explicit form~\eqref{Wm}. 

Next, we show that~$\cW_+$ exists as well, and that $\cW_+\cF_0$ equals $\cF 
S(\cdot)^*$. (In view of~\eqref{Wm}, this amounts to~\eqref{Wp}.) Proceeding along the same lines as before, we study
\be
 \| \big( e^{itM}e^{-itM_0}\cF_0-\cF S(\cdot)^*\big)f\|_1=\|  (\cF_0-\cF S(\cdot)^*)e^{-it\hat{M}}f\|_1,\ \ \ f\in\hat{\cC},
\ee
for $t\to\infty$. We have
\be
 \big(  (\cF_0-\cF S(\cdot)^*)\exp(-it\hat{M})f\big)(r)= \frac{1}{\sqrt{2\pi}}\sum_{\de=+,-}\de \phi_{\de}(r),
 \ee
where
\be
\phi_{\de}(r)\equiv \intp dk\, \big(e^{i\de rk}-\Psi(\de r,k)T(-k)+\Psi(-\de r,k)R(-k)\big)e^{-it\mu(k)}f_{\de}(k). 
\ee
Taking $r\to -r$ in $\phi_-(r)$, we deduce that we need only show that the $L^2(\R)$-norm of the function
\be
h(t;r)\equiv \intp dk\,   \big(e^{i rk}-\Psi( r,k)T(-k)+\Psi(- r,k)R(-k)\big)e^{-it\mu(k)}f(k),\ \ \ f\in C_0^\infty((0,\infty)),
\ee
vanishes for $t\to \infty$. Now from~\eqref{Psiform} we get
\be
h(t;r)=\sum_{\tau=+,-}\tau h^{\tau}(t;r),
\ee
where
\be
h^+(t;r)\equiv \intp dk\,   e^{i rk}\big(1-w(r)^{1/2}m^+(r,k)T(-k)+w(r)^{1/2}m^-(-r,k)R(-k)\big)e^{-it\mu(k)}f(k),
\ee
\be
 h^-(t;r)\equiv \intp dk\,  e^{-irk}\big(w(r)^{1/2}m^-(r,k)T(-k)-w(r)^{1/2}m^+(-r,k)R(-k)\big)e^{-it\mu(k)}f(k).
\ee

As before, the Riemann-Lebesgue lemma implies that the functions~$h^{\pm}(t;r)$ vanish for $t\to \infty$ and fixed~$r$, so it remains to supply dominating functions for $t\ge 1$, say. This can be done by adapting the above reasoning for the incoming wave operator. More in detail, for~$h^+(t;r)$ we can invoke the unitarity relation~\eqref{unit1} when $r\ge R$ and use~\eqref{exp1} for integration by parts when $r\le -R$, whereas~$h^-(t;r)$ can be handled by using~\eqref{unit2} for $r\le -R$ and~\eqref{exp2} for $r\ge R$. Thus our proof is now complete.
\end{proof}

In the language of time-independent scattering theory, this theorem reveals that the integral kernel of $\cF$ is the incoming wave function
\be\label{incom}
\Psi^{in}(r,k)=
\left(\begin{array}{c}
\Psi(r,k) \\ 
-\Psi(-r,k) 
\end{array}\right),\ \ \ k>0,
\ee
cf.~\eqref{cFdef}--\eqref{Fm}, whose relation to the outgoing one 
\be\label{outgo}
\Psi^{out}(r,k)=\left(\begin{array}{c}
\Phi(r,k) \\ 
-\Phi(-r,k) 
\end{array}\right),\ \ \ k>0,
\ee
is given by
\be\label{Phi}
\Phi(r,k)=T(-k)\Psi(r,k)-R(-k)\Psi(-r,k).
\ee
Equivalently, we have
\be
\Psi^{in}(r,k)=S(k)\Psi^{out}(r,k),
\ee
with the $S$-matrix $S(k)$ given by~\eqref{Sm}. 

It follows from our assumptions that each of the above dynamics defined via~$\cF$ is invariant under the usual parity operator
\be
(\cP f )(r)\equiv f(-r),\ \ \ f\in\cH.
\ee
Indeed, we readily calculate that the operator
\be 
\hat{\cP}\equiv \cF^* \cP \cF, 
\ee
is given by
\be\label{pari}
(\hat{\cP}g)(k)=
\left(\begin{array}{cc}
0 & -1 \\ 
-1 & 0  
\end{array}\right)g(k),\ \ \ g\in\hat{\cH}.
\ee
This implies that the range~$\cH^r$ of $\cF$ is left invariant by~$\cP$, and since $\hat{M}$ commutes with~$\hat{\cP}$, we also have $[\cP,\exp(itM)]=0$. (To be quite precise, this holds when we define $M$ on the orthocomplement~$(\cH^r)^{\perp}$ of~$\cH^r$ in such a way that it also commutes with~$\cP$ on~$(\cH^r)^{\perp}$.)

By contrast, the state of affairs for the customary time reversal operator
\be\label{cT}
(\cT f)(r)\equiv \overline{f(r)},\ \ f\in\cH,
\ee
is not clear in the present axiomatic context (as opposed to the main text, as we shall see shortly). We proceed to elaborate on this. First, we can  easily calculate
\be
\hat{\cT}_0 \equiv \cF_0^* \cT \cF_0,
\ee
yielding
\be
(\hat{\cT}_0g)(k) =\left(\begin{array}{cc}
0 & -1 \\ 
-1 & 0  
\end{array}\right)\overline{g(k)},\ \ \ g\in\hat{\cH}.
\ee
Thus $\hat{\cT}_0$ commutes with all of the dynamics~$\hat{M}$ given by~\eqref{hatM}. Clearly, this entails
\be\label{cTM0}
\cT\exp(-itM_0)=\exp(itM_0)\cT,\ \ \ t\in\R.
\ee  

The difficulty is now that we do not know whether our assumptions imply~\eqref{cTM0} with $M_0\to M$. 
To explain what is involved, let us assume that this is indeed the case.
Then, the definition~\eqref{Wpm} of the wave operators entails
\be
\cT \cW_-=\cW_+\cT.
\ee
Using~\eqref{Wm} and~\eqref{Wp}, we readily deduce
\be\label{intid}
\cF=\cT\cF S(\cdot)^*\hat{\cT}_0.
\ee
When we now compare the kernels of the transforms in~\eqref{intid}, then
we obtain 
\be
\Psi(r,k)=T(k)\overline{\Psi(-r,k)}-R(k)\overline{\Psi(r,k)}.
\ee
In view of our standing assumption~\eqref{Psiconj}, this amounts to the identity
\be\label{cTsym}
\Psi(r,k)=T(k)\Psi(-r,-k)-R(k)\Psi(r,-k).
\ee
On account of~\eqref{incom}--\eqref{Phi}, this identity can also be rewritten as the relation
\be\label{outin}
\Psi^{out}(r,k)=\left(\begin{array}{cc}
0 & -1 \\ 
-1 & 0  
\end{array}\right)\overline{\Psi^{in}(r,k)}.
\ee
Clearly, this argument can be reversed: Assuming~\eqref{cTsym} holds true, we deduce
\be\label{cTM}
\cT\exp(-itM)=\exp(itM)\cT,\ \ \ t\in\R.
\ee

For the transforms coming from the main text, the time reversal identity~\eqref{cTsym} is indeed valid, cf.~\eqref{psirev}. Moreover, it can also be verified for the two transform kernels~\eqref{cFe} in Proposition~B.4. But we do not know whether~\eqref{cTsym}  is a necessary consequence of the assumptions we made in Theorem~D.1.

From now on, we add to the assumptions of Theorem~D.1 the extra assumption that~$\cF$ is unitary (equivalently, that~$\cH^r$ equals~$\cH$, cf.~\eqref{ranF}). As we shall show next, this implies  that the function~$U(r)$ arising from the large-$|k|$ asymptotics of the coefficients (cf.~\eqref{Psikas}) may be viewed as the $S$-matrix of a large class of `dual' dynamics. This class arises by starting from real-valued, smooth, odd functions~$d(r)$, which satisfy
\be\label{derpos}
d'(r)>0,\ \ \ r\in\R.
\ee
Thus $d(r)$ is strictly increasing, but need not be unbounded. (For example, the function $\tanh r$ satisfies the assumptions.)

Any such function gives rise to a self-adjoint multiplication operator~$D$ on its natural domain~$\cD(D)\subset \cH$. Since $\cF$ is assumed to be unitary, we can now define a self-adjoint operator $\hat{D}$ on~$\hat{\cH}$  
by  
\be\label{Dhdef}
\hat{D}\cF^* f\equiv \cF^* D f,\ \ f\in \cD(D).
\ee 

We continue to compare the associated `interacting'  unitary time evolution 
\be
\hat{\cU}(t)\equiv \exp(-it\hat{D}),\ \ t\in\R,
\ee
to the `free' evolution defined by using~$\cF_0^*$. Thus, replacing~$\cF^*$  by~$\cF_0^*$ in~\eqref{Dhdef}, we obtain a self-adjoint operator $\hat{D}_0$ on~$\hat{\cH}$, which yields a time evolution
\be
\hat{\cU}_0(t)\equiv \exp(-it\hat{D}_0),\ \ t\in\R.
\ee
As before, our goal is to show that the dual wave operators,
\be\label{Whpm}
\hat{\cW}_{\pm}\equiv \lim_{t\to \pm \infty}\hat{\cU}(-t)\hat{\cU}_0(t),
\ee
 exist, and to clarify their relation to~$\cF^*$ and the unitary operator
 \be\label{Udef}
 (Uf)(r)\equiv U(r)f(r),\ \ \ f\in\cH.
 \ee

Before doing so, we detail the function~$d(r)$ that arises from the dual relativistic Calogero-Moser dynamics. It is given by
\be\label{dCM}
d_{CM}(r)=2\sinh(\kappa r),
\ee
and the corresponding operator $\hat{D}_{CM}$ acts on the core $\cF^*(\cC)$ as the analytic difference operator
\be\label{HhCM}
\hat{H}_{CM}=\exp(-i\kappa \partial_k)-\exp(i\kappa \partial_k).
\ee
To see that this gives rise to  the `eigenvalues'~\eqref{dCM}, recall that we have assumed (above~\eqref{mtform}) that~$m^+(r,k)$ is $i\kappa$-periodic in~$k$, whereas~$m^-(r,k)$ is assumed to be $i\kappa$-antiperiodic.

\begin{theorem}
The strong limits \eqref{Whpm} exist and are given by 
\be\label{Whm}
\hat{\cW}_{\pm}=\cF^* U(\cdot)^{\mp 1/2} \cF_0,
\ee
where $U(r)$ is the unitary multiplication operator on~$\cH$ given by the square of~\eqref{Ur}.
\end{theorem}
\begin{proof}
Our proof is patterned after the proof of Theorem~D.1.
To show that~$\hat{\cW}_-$ exists and is given by~\eqref{Whm}, we start from the $\hat{\cH}$-norm
\be\label{gmd}
 \| \big( e^{it\hat{D}}e^{-it\hat{D}_0}\cF_0^*-\cF^*U(\cdot)^{1/2}\big)f\|_2=\|   (\cF_0^*-\cF^*U(\cdot)^{1/2})e^{-itD}f\|_2,\ \ f\in\cC.
\ee
In order to prove it vanishes for~$t\to-\infty$, we consider
\be\label{difde}
 \big( \cF_0^*-\cF^*U(\cdot)^{1/2})e^{-itd(\cdot)}   f\big)_{\de}(k) =\frac{\de}{\sqrt{2\pi}}\int_{\R} dr\,   \big(e^{-i\de rk}-\Psi(\de r,-k)U(r)^{1/2}\big)e^{-itd(r)}f(r). 
 \ee
By virtue of the Riemann-Lebesgue lemma, this function vanishes for $t\to -\infty$ and~$k>0$ fixed, so we need only exhibit a suitable dominating function in~$L^1((0,\infty),dk)$.
 
 For $\de=+$ we should look at the functions 
\be\label{gpp}
\hat{g}_+^+(t;k)\equiv \int_{\R} dr\,   e^{-i rk}(1-w(r)^{1/2}m^+(r,-k)U(r)^{1/2})e^{-itd(r)}f(r),
\ee
\be\label{gpm}
 \hat{g}_+^-(t;k)\equiv \int_{\R} dr\,  e^{irk}w(r)^{1/2}m^-(r,-k)U(r)^{1/2}e^{-itd(r)}f(r).
\ee
To supply dominating functions for~$|\hat{g}_+^{\pm}(t;k)|^2$, we split   integration over~$(0,\infty)$ into intervals $(0,R]$  and $[R,\infty)$, where~$R$ is chosen large enough for the  asymptotics \eqref{mpask}--\eqref{mmask} to be valid. Since the functions~$|\hat{g}_+^{\pm}(t;k)|^2$ are bounded uniformly in~$t$ on~$(0,R]$, the contribution of this interval vanishes for $t\to -\infty$.

To bound~$|\hat{g}_+^+(t;k)|^2$ on $[R,\infty)$, we need only invoke~\eqref{mpask} and~\eqref{Ur}. Indeed, from this we deduce that the dominant contribution cancels, so we are left with an exponentially decreasing dominating function. For~$|\hat{g}_+^-(t;k)|^2$ the existence of such a function is immediate from~\eqref{mmask}, so it now follows that the norm of the function~\eqref{difde} with~$\de=+$ vanishes for~$t\to -\infty$.

For the choice~$\de=-$, we should majorize the functions
\be\label{gmp}
\hat{g}_-^+(t;k)\equiv \int_{\R} dr\,   e^{i rk}(1-w(r)^{1/2}m^+(-r,-k)U(r)^{1/2})e^{-itd(r)}f(r),
\ee
\be\label{gmm}
 \hat{g}_-^-(t;k)\equiv \int_{\R} dr\,  e^{-irk}w(r)^{1/2}m^-(-r,-k)U(r)^{1/2}e^{-itd(r)}f(r).
\ee
As before, boundednes for $k\in (0,R]$   is plain, and 
just as for~\eqref{gpm}, we get an exponentially decreasing dominating function for~\eqref{gmm} on~$[R,\infty)$ right away from~\eqref{mmask}. 

Invoking once again~\eqref{mpask} and~\eqref{Ur}, we deduce that it remains to bound  the function
\be\label{ghright}
\int_{\R} dr\,   e^{i rk}\big(1-w(r)\overline{C(-r)}C(r)\big)e^{-itd(r)}f(r),\ \ k\in[R,\infty),
\ee
for $t\le -1$. Writing
\be\label{exp3}
(k-td'(r))^{-1}(-i\partial_r)\exp(irk-itd(r)),
\ee
we observe that by our assumption~\eqref{derpos}  the denominator  is bounded  away from zero on the compact support of~$f(r)$ for~$t\le -1$ and $k\ge R$. 
It easily follows that when we integrate by parts  we can obtain an $O(1/k)$-bound that is uniform for~$t\le -1$. Thus the  $L^2([R,\infty),dk)$-norm of~\eqref{ghright} vanishes for~$t\to -\infty$.  

The upshot is that we have completed the proof that $\hat{W}_-$ exists and is given by~\eqref{Whm}. The proof for~$\hat{W}_+$ only involves some obvious changes, so we omit it.   
\end{proof}

\end{appendix}

\vspace{5mm}

\noindent
{\Large\bf Acknowledgments}

\vspace{8mm}

\noindent
 We would like to thank the referees for their comments, which helped us to improve the exposition.

\vspace{4mm}

\bibliographystyle{amsalpha}

\end{document}